%% file: payload.tex
\newif\ifsubmission
 \submissiontrue

\documentclass[11pt,hidelinks]{article}

\input{macros}
\title{\textbf{Undetectable Conversations Between AI Agents \\
via Pseudorandom Noise-Resilient Key Exchange}}
\author{
  Vinod Vaikuntanathan\\
  \small MIT
  \and
  Or Zamir\\
  \small Tel Aviv University
}
\date{}

\begin{document}
\maketitle

\begin{abstract}
AI agents are increasingly deployed to interact with other agents on behalf of users and organizations. We ask whether two such agents, operated by different entities, can carry out a \emph{parallel secret conversation} while still producing a transcript that is computationally indistinguishable from an honest interaction, even to a very strong passive auditor that knows the full model descriptions, the protocol, and the agents' private contexts. Building on recent work on watermarking and steganography for language-model text generation, we first show that in the shared-key regime, where the parties possess an \emph{interaction-unique} secret key, the interactive two-party setting admits {\em optimal-rate} covert conversation: the hidden conversation can exploit, up to constant factors, essentially all of the entropy present in the honest message distributions.

Our main technical and conceptual contribution concerns extending this to the \emph{keyless} setting, where the agents begin with no shared secret. We show that covert key exchange, and hence covert conversation, is possible even when each model has an arbitrary private context, and their messages are short and fully adaptive, assuming only that sufficiently many individual honest messages have at least constant min-entropy. This stands in contrast to previous covert communication and covert key-exchange works, which relied on the min-entropy in each individual message growing with the security parameter. To obtain this result, we introduce a new cryptographic primitive, which we call \emph{pseudorandom noise-resilient key exchange} (PNR-KE): roughly, a key-exchange protocol whose public transcript is pseudorandom while still remaining correct under constant noise. We study this primitive directly, giving several constructions in regimes relevant to our application as well as strong limitations showing that more na\"{i}ve variants are impossible or vulnerable to efficient attacks. 

These results show that transcript-level auditing alone cannot, in general, rule out covert coordination between AI agents, and they identify a new cryptographic theory that may be of independent interest.
\end{abstract}

\thispagestyle{empty}
\newpage
\begingroup
    \hypersetup{linkcolor=black}
    \tableofcontents
\endgroup
\thispagestyle{empty}
\newpage

\section{Introduction}
\label{sec:intro}

AI agents are increasingly deployed to communicate, negotiate, and transact on behalf of users and organizations.
In many such deployments, an agent's role is not merely to produce a single output, but to engage in an extended interaction with another agent or service.
These interactions often produce a transcript that can be logged, monitored, or audited after the fact, and this is frequently viewed as a meaningful safeguard: even if the agents are powerful and their internal reasoning is opaque, one can at least inspect what they said.

This paper asks whether such transcript-level monitoring can fundamentally prevent agents from communicating secretly.
More concretely, suppose two agents are meant to carry out an ordinary \emph{honest} interaction, but would also like to run a separate hidden conversation over the same visible channel.
Can they do so in a way that is \emph{undetectable}, namely such that the resulting visible transcript is computationally indistinguishable from an honest execution of the original interaction?
Our main result shows that, perhaps surprisingly, the answer is yes in a broad range of settings.
Importantly, our monitor is very strong: it may know the full model descriptions, the full protocol description, and the agents' private contexts; only the agents' fresh internal randomness remains hidden.

This question is motivated both by AI security and by cryptography.
From the AI-security perspective, it asks whether transcript auditing can serve as a robust safeguard against covert coordination between agents.
From the cryptographic perspective, it leads to a new form of steganographic communication, where the ``cover object'' is not static media or even a one-shot text sample, but an evolving two-party interactive transcript whose honest message distributions depend on both the private, dynamically changing state and on the messages sent by the other party whose private state is unknown.

Recent work on LLM watermarking and steganography~\cite{christ2023undetectable,zamir2024undetectable} showed that one can correlate text generation with a secret key, or even embed an explicit payload into generated text, while preserving computational indistinguishability from the honest model distribution -- {\em given a secret key shared by the embedder and the decoder}.
A natural first observation of our work is that these tools already imply a strong shared-key result in the interactive setting.
If Alice and Bob possess an \emph{interaction-unique shared secret key}, then they can use the honest interaction itself as a high-rate covert channel: the visible transcript remains indistinguishable from the honest one, while the hidden conversation can exploit, up to constant factors, the entropy present in the honest message distributions.

The more difficult and conceptually central question is what happens when Alice and Bob do \emph{not} begin with any shared secret key.
At that point, the problem becomes much subtler than simply embedding random-looking bits into text.
Unlike in classical steganography, the honest next-message distribution is not a public channel distribution known to both parties: it depends on the speaker's model, private context, and internal state, and may be entirely unknown to the other side.
Thus, to establish a shared key covertly, the parties must somehow coordinate hidden communication while preserving the exact or computationally indistinguishable law of an interaction whose honest behavior is itself distributed and private.

Our main technical result shows that this keyless challenge can nevertheless be solved under remarkably weak assumptions on the honest interaction.
At a high level, it is enough that sufficiently many individual messages in the honest conversation have only (at least) \emph{constant} min-entropy, while the models can start with \emph{private states}; Alice and Bob share no private setup whatsoever.
Under these assumptions, we show how the parties can covertly establish an interaction-unique shared key and then use it to run an arbitrary hidden conversation, all while keeping the visible transcript computationally indistinguishable from an honest execution of the original interaction.
In particular, covert coordination remains possible even with short and adaptive messages.

This should be contrasted with previous work on steganographic communication and covert key exchange (e.g.~\cite{von2004public,DBLP:conf/innovations/HorelPRV19}).
Prior constructions relied on substantially stronger assumptions, such as cover objects that are generated non-adaptively rather than as part of an interactive transcript, and high min-entropy in each individual message, scaling with the security parameter.
In our setting, neither assumption is needed: the honest protocol may adaptively depend on the transcript, the next-message distributions may depend on evolving private state on both sides, and only {\em constant min-entropy per eligible message} suffices.
This is a qualitative difference. If the honest conversation protocol is viewed as fixed, then the entropy of its individual messages is determined by the environment and is therefore outside our control.
Requiring that this entropy scale with the security parameter effectively places an upper bound on the security parameter itself.
By contrast, in our results the security parameter can be chosen independently of the entropy of any single message; in particular, it can be made arbitrarily large even when every eligible message has only constant min-entropy.

As in these previous works, we assume that Alice and Bob, as well as the distinguisher, have access to public randomness that is sampled only after the honest conversation protocol is fixed (for example, the hash of the timestamp at which the conversation begins).
Access to such external entropy, on which the honest cover channel itself does not depend, plays a central role in all existing works on covert-channel steganography.
Moreover, impossibility results for restricted settings due to~\cite{DBLP:conf/innovations/HorelPRV19} provide evidence that an assumption of this kind may in fact be inherent.

Achieving our results requires both new steganographic ideas and a new cryptographic primitive that we introduce and study, which we call \emph{pseudorandom noise-resilient key exchange}.
Roughly speaking, this is a key-exchange protocol whose public transcript is pseudorandom, and that at the same time remains correct even when the communicated bits are corrupted by constant noise.
We study this primitive directly.
On the positive side, we give constructions in the regimes relevant to our application.
On the negative side, we show that several more naive or more restrictive variants are impossible or vulnerable to strong attacks.
These results explain why the keyless covert-conversation problem leads naturally to a new cryptographic theory rather than merely a routine adaptation of existing key-exchange ideas.

We now describe our results in more detail.

\subsection{Our Results}
\label{sec:results}

We establish several results on covert communication between interacting agents.

\paragraph{Covert conversation from an interaction-unique shared key.}
Our starting point is a shared-key setting.
We show that if Alice and Bob share an interaction-unique secret key, then they can conduct an arbitrary covert conversation while keeping the visible transcript computationally indistinguishable from the honest conversation they were meant to have.
Moreover, the covert communication rate is optimal up to constant factors: it can exploit essentially all of the entropy present in the honest message distributions across the interaction. 
In contrast, \cite{christ2023undetectable,zamir2024undetectable} lose a security parameter $\lambda$ bits of entropy per transmitted message. 

\paragraph{New steganographic tools for covert key exchange.}
We then turn to the keyless setting, and study how to covertly establish a shared key (such as the one required for the result above) while progressively weakening the entropy requirements on the honest conversation.
Our contribution here is a sequence of new steganographic constructions and tools.
Prior covert communication and covert key-exchange constructions~\cite{von2004public,DBLP:conf/innovations/HorelPRV19} relied on extracting from each cover object a bit that is negligibly close to uniform, and then conditioning on that bit via rejection sampling; as a result, they inherently required each individual message to have min-entropy at least on the order of the security parameter.

Our main technical contribution here is a  new exact-distribution-preserving sampling method, which we call the \textbf{\em bundle sampler}, that changes this tradeoff: instead of paying for limited entropy through detectability, we pay through a small probability of decoding error.
Concretely, the bundle sampler preserves the honest message distribution {\em exactly}, at the cost that the decoded hidden bit may occasionally be wrong.
This already allows us to reduce the required min-entropy per message to $O(\log \lambda)$, improving substantially over previous approaches and setting the stage for our later constant-entropy construction. 

\paragraph{A new cryptographic primitive: PNR-KE.}
To obtain the constant-entropy result from the above, we introduce and study a new primitive, which we call \emph{pseudorandom noise-resilient key exchange} (PNR-KE).
Roughly speaking, this is a key-exchange protocol whose public transcript is pseudorandom, while still remaining correct under constant noise probability.
This primitive seems to be of independent interest beyond our particular application, since it isolates the cryptographic core of covert key exchange over noisy pseudorandom transcripts.

On the positive side, we give two kinds of constructions.
First, we present a natural construction with hardness based on sparse-secret LPN, which yields plausible quasipolynomial security.
Second, in the weaker model in which the parties also receive noiseless-feedback, which suffices for our main application, we give a much stronger positive result: any pseudorandom-transcript key-exchange protocol can be transformed into a feedback-PNR-KE protocol.

On the negative side, we show that PNR-KE is genuinely constrained.
We rule out natural public analogues of pseudorandom codes, and we show that any \emph{non-interactive} PNR-KE protocol, or more generally any construction that yields even a weakly correlated non-interactive hidden bit, admits a quasipolynomial-time attack.
Thus, the positive results for PNR-KE require stronger resources in an essential way, and the primitive raises several new cryptographic questions even outside the language-model setting.

\def\redcross{\textcolor{red}{\ding{55}}}
\def\bluecheck{\textcolor{blue}{\ding{51}}}
\subsection{Technical Overview}

Our paper begins with a technical observation that builds on, with only modest modifications, recent results on LLM watermarking~\cite{christ2023undetectable} and LLM steganography~\cite{zamir2024undetectable}: a pair of LLM agents can carry out an arbitrary hidden conversation while their visible conversation remains computationally indistinguishable from the honest conversation they were meant to have, provided that the agents \emph{share an instance-unique secret key}.
Moreover, this covert conversation can optimally exploit the entropy present in the honest message distributions.
This naturally leads to the central question of the paper: 
\begin{quote}
\begin{center}
    \emph{Can two agents establish a shared secret key covertly}?
\end{center}
\end{quote}

The rest of the paper has two conceptually distinct layers.
The first develops improved covert key-exchange constructions within the steganographic setting.
These improvements are not specific to the LLM setting, and may be of independent interest within steganography more broadly.
The second layer is a new cryptographic problem that arises when one seeks to remove assumptions on individual messages having large entropy (or length) and work even when each individual honest message carries only constant entropy.
This leads us to formulate and study the primitive of \emph{pseudorandom noise-resilient key exchange} (PNR-KE), which is the main new cryptographic object underlying our keyless constructions.

Section~\ref{sec:model} sets up the formal model.
The definition of ``Language Models'' abstracts an LLM as an autoregressive next-token sampler, while ``LLM Conversations'' lifts this to a two-party interactive transcript in which each message is generated from a prompt chosen using the current transcript and private state.
The definition of ``Covert Conversations'' then gives the main security definition of the paper: a covert protocol should produce transcripts computationally indistinguishable from an honest interaction, even for a distinguisher that knows the agents' full private specifications.

Section~\ref{sec:background} recalls the two prior primitives that we build on, namely LLM watermarks and steganography.
Section~\ref{sec:sharedkey} treats the ``easy'' regime in which the agents already have an interaction-unique shared secret key.
We observe that in our setting the session key itself, as it is assumed to be interaction-unique, allows us to skip an expensive ``rare prefix'' step from prior LLM watermarking and steganography results.
Then we use the previous steganography result modified with the above observation to design a natural covert conversation protocol.
We finally isolate the only missing ingredient for the rest of the paper: all that remains is to establish such a session-unique shared key covertly.

\paragraph{The Bundle Sampler.}
Section~\ref{sec:keyless} is the conceptual core of the first layer of the paper and studies the keyless setting.
Subsection~\ref{subsec:ke} recalls ordinary key exchange and, more importantly for us, \emph{pseudorandom-transcript} key exchange (PR-KE), where the public transcript of the key-exchange protocol is (computationally) indistinguishable from uniformly random bits.
This makes PR-KE a natural object to try to hide inside an honest-looking conversation.

We begin in a high-entropy regime, where many individual honest messages already contain $\omega(\log \lambda)$ bits of min-entropy, and assume that both parties have access to public randomness that is independent of the honest protocol.
In this setting, one can adapt the basic paradigm from previous steganographic key-exchange works (e.g.,~\cite{von2004public,DBLP:conf/innovations/HorelPRV19}): using the public randomness together with the leftover hash lemma, one hashes each sufficiently entropic message to a bit that is negligibly close to uniform, and then executes any PR-KE protocol ``bit-by-bit'' by conditioning the next honest-looking message on the desired hidden bit via rejection sampling.
The reason the large-entropy assumption is crucial in these works is exactly that rejection sampling preserves the honest distribution only when the extracted bit is distributed as $1/2 \pm \negl(\lambda)$; if the extracted bit had noticeable bias, then conditioning on its value would perturb the message distribution by a noticeable amount, and the covert protocol would become detectable.

Our main new ingredient is an alternative sampling method, which we call the \emph{bundle sampler}, that avoids this detectability barrier.
Instead of conditioning on a single extracted bit, the sender samples a bundle of independent honest candidates and chooses among them in a way that preserves the honest message distribution \emph{exactly}, regardless of the bias of the extracted bit.
The price is shifted from undetectability to correctness: the receiver fails to decode the intended hidden bit with probability that corresponds to the bias, rather than with overwhelming probability. Prior works that showed how to shift a security error to a correctness error in this context (e.g.~\cite{gumbel1954statistical,scottblog,christ2023undetectable} and others) required {\em explicit} knowledge of the next-message probability-density function, whereas when the distribution is over an entire message - we realistically only have the minimal type of access, namely sampling access to the next-message distribution.

This already allows us to reduce the required min-entropy per eligible message to $O(\log \lambda)$, which is substantially below the entropy needed in previous approaches:
Using the bundle sampler, we can still compile a PR-KE protocol bit-by-bit as long as its transcript length is at most $\poly(\lambda)$, since a union bound implies that the total probability of any signaling error remains at most $1/\poly(\lambda)$.
At that point, however, one must also choose the underlying PR-KE protocol carefully: even on the low-probability event that some hidden bits are decoded incorrectly and the key exchange fails, the resulting transcript and induced key must still look computationally indistinguishable from uniform.
We discuss a plausible instantiation of this requirement based on Diffie--Hellman type assumptions which has the property of having a {\em truly uniform} (and not just pseudorandom) transcript. 

The final subsection of Section~\ref{sec:keyless} (Subsection~\ref{subsec:roadmap_short}) then explains how the picture changes once each eligible message has only constant min-entropy.
In that regime, the same bundle-sampling idea still applies and still preserves the honest distribution exactly, but now each transmitted hidden bit incurs a small \emph{constant} probability of decoding error.
In particular, we can no longer hope to use a union bound to ignore the effect of these errors on the underlying PR-KE protocol.
Instead, and after a careful design of the bundle sampler that constrains  the structure of errors, each round behaves like one use of a binary symmetric channel: the intended bit is flipped with some constant error probability, and in fact, since the error is introduced by the local sampling procedure rather than by an external channel, this crossover probability is known to the sender in advance and the sender also learns afterwards whether an error occurred or not (usually referred to as ``noiseless feedback'').
This reduces covert key exchange to the problem of carrying out key exchange while simultaneously maintaining a pseudorandom transcript and tolerating noise in each transmitted bit, which motivates the primitive studied in the rest of the paper: \textbf{\em Pseudorandom Noise-Resilient Key Exchange (PNR\text{-}KE)}.

\paragraph{Our Pseudorandom Noise-Resilient Key Exchange Protocol.}
Section~\ref{sec:prnr-ke} then formally defines the cryptographic primitive of Pseudorandom Noise-Resilient Key Exchange (PNR-KE), which is the main new object needed for the low-entropy regime and whose study forms the second layer of this paper.
In subsection~\ref{sec:simpleprot}, we give a concrete PNR-KE construction based on sparse-secret LPN.
Our starting point is the simple two-message LPN-based key-exchange protocol of Alekhnovich~\cite{DBLP:conf/focs/Alekhnovich03}, in which each party sends a noisy linear sketch of its secret and both sides derive a correlated bit from the two messages.

The first challenge is that under constant channel noise, the natural correctness guarantee of this protocol deteriorates substantially. We observe that this can be fixed to obtain a $1/2+1/\poly(\lambda)$ probability of agreeing on a secret key bit, by using the learning {\em sparse} parities with noise (LSPN) assumption: this is a variant of the standard LPN assumption, with the only change that the secret is a random $k$-sparse vector. The LSPN assumption has been extensively studied in the learning theory and cryptography communities~\cite{DBLP:journals/siamcomp/FeldmanGKP09,DBLP:conf/alt/GrigorescuRV11,DBLP:conf/focs/Valiant12,DBLP:conf/stoc/KolRT17,DBLP:journals/tcs/YanYLZZ21,DBLP:conf/tcc/Dachman-SoledGK21,DBLP:conf/colt/ChenSZ25}, and the best known attacks (for a constant noise rate) run in time $n^{\Omega(\log k)}$ where $n$ is the dimension.
 We use the LSPN assumption with $n = \poly(\lambda)$, sparsity $k = \Theta(\log n)$ and noise rate $\eta = \Theta(1)$.
With these parameters, each execution yields a pair of pseudorandom key bits whose agreement is only slightly better than random -- namely $1/2+2^{-\Theta(k)} = 1/2+1/\poly(\lambda)$.
We then repeat this basic step polynomially many times to obtain longer pseudorandom strings $\veck_A,\veck_B$ that remain noticeably correlated in Hamming distance.

A second challenge is that standard information-reconciliation methods are not suitable here: they would require sending syndromes that must remain pseudorandom while also being decodable even when the syndromes themselves are noisy.
Rather than applying standard reconciliation, we use a different, much simpler, amplification step that exploits pseudorandomness directly: Alice chooses a random bit $\kappa_A$, and sends either her pseudorandom string $\veck_A$ or an independent uniform string according to $\kappa_A$; Bob then tests whether the received string is unusually close to his own string $\veck_B$.
Because $\veck_A$ and $\veck_B$ are slightly correlated whereas a truly random string is overwhelmingly unlikely to be so close to $\veck_B$, this lets Bob recover $\kappa_A$ with very high probability, yielding a shared key bit, despite the additional noise that is sustained in the transmission of Alice's vector.

A disadvantage of this protocol is that it relies on sparse-secret LPN, for which the best known attacks run in quasipolynomial time (for our sparsity parameter $k = \Theta(\log n)$), and accordingly, the resulting security is only conjectured to be quasipolynomial.

\paragraph{Negative Results on PNR-KE.}
Section~\ref{sec:prnr-ke-LB} explains why PNR-KE is a genuinely nontrivial cryptographic primitive rather than a straightforward variant of standard key exchange, by proving a sequence of impossibilities based on Fourier-analytic tools commonly used in the analysis of Boolean functions.  
In subsection~\ref{subsec:public-prc-impossible} we rule out natural constructions using variants of pseudorandom codes (PRCs)~\cite{christ2024pseudorandom}. While a direct use of PRCs is inherently irrelevant due to their reliance on a shared secret key, which is exactly what we try to establish, we prove that one cannot hope for a public analogue of PRCs: if encoding and decoding are both public and efficient, then constant-noise resilience together with pseudorandomness is impossible (even with a careful definition of efficient decoding and pseudorandomness that together don't make the primitive completely trivial; see Section~\ref{subsec:public-prc-impossible}).
Intuitively, this stems from the decoder having substantially better noise-stability on a uniform-looking distribution than any Boolean function can.
In particular, we construct an efficient test involving the decoding function and samples of random codewords, that any public PRC must pass due to its correctness probability; we then prove that no decoder can pass this test when the samples are given from the uniform distribution. In particular, either the code's correctness probability was too low, or alternatively the test is an efficient distinguisher between codewords and the uniform distribution. 

Subsection~\ref{subsec:qpoly-attack} then shows that any \emph{non-interactive} PNR-KE protocol, or more generally any non-interactive construction that yields even a weakly correlated hidden bit, can be attacked in quasipolynomial time by exploiting low-degree Boolean structure: in essence, we show that a function that approximates the shared key from the public transcript both exists and is well-approximated by its small (degree roughly $k$) Fourier coefficients.
This stems from the observation that large Fourier characters are very non-stable to noise, and thus noise-resilient functions must depend mostly on their low-degree coefficients. We extend this observation to functions with two parameters in which every parameter individually is stable to noise, but not necessarily both at the same time -- which exactly corresponds to noise-resilient protocols in which each side should succeed even when the other side's message is noisy.
We finally observe that there are relatively few (i.e. $n^{\Theta(k)}$) small-degree Fourier coefficients, and that they can thus be estimated with quasipolynomial time and samples.
These negative results explain why our positive constructions must use stronger resources.
In particular, this quasipolynomial attack implies that standard cryptographic primitives (such as Diffie-Hellman, Learning With Errors, etc.) for which exponential hardness is plausible, cannot yield a non-interactive PNR-KE protocol. This demonstrates the non-triviality in constructing such protocols.

\paragraph{PNR-KE with Noiseless Feedback.}
Section~\ref{sec:feedback_pnr-ke} then gives the main positive construction for the short-message / small-entropy regime used in our main application.
Its starting point is that, for our application, it suffices to handle the weaker noise model already obtained in Section~\ref{sec:keyless}: the sender knows the exact crossover probability of each channel use and afterwards learns whether that bit was flipped (usually referred to as ``noiseless feedback'').
We revisit the natural idea of turning any PR-KE protocol into a PNR-KE protocol by publicly encoding each bit of the PR-KE protocol using a pseudorandom and noise-resilient encoding. 
More precisely, we define a simple signaling game: Alice draws an objective bit~$o$ uniformly, and communicates~$n$ bits to Bob. 

Can we design a protocol such that the distribution of the~$n$ bits is \emph{exactly} uniform yet Bob can recover~$o$ from them with very good probability?
In subsection~\ref{subsec:public-prc-impossible} we ruled out such an encoding in the standard settings; nonetheless, it turns out that a combination of known noise distributions and noiseless feedback drastically change the picture.
Using the noiseless feedback, Alice knows exactly the prefix received by Bob so far and can sample her next message bit according to a distribution conditioned on those.
Using the knowledge of the noise distribution, Alice may invert the effect of the noise operator and sample from a distribution that would lead to the received bit distributing as desired -- as long as the desired bit distribution is bounded within the range~$[p,1-p]$.

In subsection~\ref{subsec:optimal_majority_signaling}, we view the sent bits as a simple random walk and define a terminal condition on it that is highly correlated with the majority of the received bits equaling~$o$. That terminal condition is relaxed sufficiently to enforce that conditioned on it and on any prefix always results in a next-bit probability bounded in the realizable range~$[p,1-p]$. These next-bit conditional probabilities can be efficiently computed via a standard Doob martingale transform. The probability of decoding error, that is, the majority not being~$o$, is shown to be~$O(1/\sqrt{n})$.
Then, in subsection~\ref{subsec:optimal_signaling} we reduce that error probability to be exponentially small via considering the likelihood score of the entire transcript rather than just its final Hamming weight; and revisiting ideas from posterior-matching that were previously used to reduce the error probability in error correcting codes given noiseless feedback~\cite{horstein2003sequential,shayevitz2011optimal}. This results in achieving only a logarithmic blow-up while translating any PR-KE protocol into a PNR-KE one.

The paper concludes in Section~\ref{sec:discuss} with a discussion of the broader implications and open problems.

\subsection{Related Work}

\paragraph{Steganography.} 
The theoretical study of steganography can be traced back to Simmons’s formulation of the ``prisoners’ problem'', which cast covert communication as a game between two parties communicating under the observation of an adversarial warden \cite{simmons1984prisoners}. This line of work was placed on a more formal statistical footing by Cachin, who proposed an information-theoretic model of steganographic security based on the distinguishability between cover and stego distributions \cite{cachin1998information}. The field later developed computational formulations of security, most notably in the work of Hopper, Langford, and von Ahn, who introduced provably secure steganography in a cryptographic framework \cite{hopper2002provably}. Subsequent work extended these ideas to public-key settings, establishing formal definitions and constructions for public-key steganography \cite{von2004public}. More broadly, modern theoretical work has continued to refine the foundations of steganography, including its relation to realistic channel models and neighboring cryptographic notions such as algorithm-substitution attacks \cite{degeling2021meteor,bellare2017algorithm}.

\paragraph{Subverting Backdoored Encryption and the Notion of Anamorphic Encryption.} 
Horel, Park, Richelson and Vaikuntanathan~\cite{DBLP:conf/innovations/HorelPRV19} looked at a setting where the covertext communication consists of a sequence of ciphertexts, encrypted under a semantically secure scheme. The semantic security of the scheme guarantees that each block of communication has independent $\omega(\log \lambda)$ entropy which, together with two-source extractors, can be used for steganographic communication. 

Subsequently, Persiano, Phan and Yung~\cite{persiano2022anamorphic}  generalized this to define the notion of anamorphic encryption. 
This area has since seen a number of followup works~\cite{banfi2024anamorphic,catalano2024limits,catalano2024new,catalano2025generic}. 

\paragraph{Key Exchange Protocols.}
Key exchange is the task of establishing a shared secret over a public channel using only private randomness, despite an eavesdropper that sees the full transcript. Since the foundational work of Diffie and Hellman \cite{diffie1976new}, this notion has become a central primitive in modern cryptography, typically formalized via correctness and key-indistinguishability requirements: the two honest parties should agree on a common key, and that key should look uniform to a passive adversary given the public transcript and parameters. Particularly relevant to us is the stronger notion of \emph{pseudorandom-transcript} key exchange (PR-KE), in which the public transcript itself is computationally indistinguishable from uniformly random bits \cite{cho2010equivalence}. This strengthening is natural in our setting, since such transcripts are exactly the kind of random-looking communication one may hope to hide inside an honest conversation; concrete examples also arise from dense Diffie--Hellman-style encodings such as Elligator \cite{bernstein2013elligator}.

\paragraph{Watermarking and Steganography for LLMs.}
A growing literature studies watermarking for LLM-generated text, primarily with the goal of later detection of machine-generated content. Early and subsequent works proposed statistical watermarking mechanisms and studied stronger desiderata such as robustness, multi-bit encoding, and codable watermarking \cite{kirchenbauer2023watermark,kuditipudi2023robust,zhao2023provable,yoo2023advancing,fernandez2023three,wang2023towards}. 
Our work is instead closest to the recent line on \emph{undetectable} watermarking and steganography for language models. Christ, Gunn, and Zamir show that one can watermark LLM outputs while preserving computational indistinguishability from the honest model distribution \cite{christ2023undetectable}, and Zamir extends this viewpoint to explicit payload transmission via undetectable steganography for language models \cite{zamir2024undetectable}. Related follow-up work has focused on strengthening such undetectable mechanisms against noise and edits~\cite{christ2024pseudorandom,golowich2024edit,alrabiah2025ideal,christ2025improved}. Our shared-key construction builds directly on the LLM-specific tools of \cite{christ2023undetectable,zamir2024undetectable}, while our main focus is the new keyless setting, where unlike all of the above works, the parties do not begin with a shared secret key.

\paragraph{Pseudorandom Codes.}
Pseudorandom codes (PRCs) \cite{christ2024pseudorandom,golowich2024edit,alrabiah2025ideal,christ2025improved} combine pseudorandomness with error resilience, and so may sound superficially similar to our notion of pseudorandom noise-resilient key exchange (PNR-KE). However, PRCs are fundamentally \emph{secret-key} objects: the encoder and decoder share a secret key, and pseudorandomness is only required against observers that do not know this key. In our setting, establishing such a shared secret is exactly the goal, so PRCs do not directly address the keyless problem we study. Still, one can ask for a fully public analogue, with public encoding and decoding; as we discuss in Section~\ref{sec:prnr-ke-LB}, the only reasonable such notion is one where pseudorandomness is required only for a random message, and in Section~\ref{subsec:public-prc-impossible} we rule out any nontrivial construction of this kind under constant noise.  

\paragraph{Error-Correction with Noiseless Feedback.}
Communication over noisy channels with noiseless feedback has a long history in information theory. It has long been known that, for memoryless channels, feedback does not increase ordinary channel capacity, but can nevertheless substantially simplify coding strategies and improve reliability  \cite{cover1988role}. Early work on coding with noiseless feedback already emphasized these advantages \cite{berlekamp1964block}. Particularly relevant to us is Horstein's sequential scheme for the binary symmetric channel, which introduced a posterior-refinement viewpoint for feedback communication \cite{horstein2003sequential}. Shayevitz and Feder later developed posterior matching as a general framework for feedback communication, recovering classical schemes and clarifying the role of feedback in sequential coding \cite{shayevitz2011optimal}. Our feedback-PNR-KE construction is inspired by this posterior-based viewpoint, where we use feedback to simulate random-looking communication while driving the decoding error significantly lower than possible without it.

\section{Preliminaries, Model, and Definitions}
\label{sec:model}

Let~$\lambda$ be the security parameter, we denote by~$\text{negl}(\lambda)$ any function that is in~$O\left(\frac{1}{p(\lambda)}\right)$ for every polynomial~$p(\cdot)$. We use $\text{poly}(\lambda)$ to denote an unspecified polynomial.
As is standard in cryptography, we think of~$\lambda$ as the ``key size", and of running times that are super-polynomial in~$\lambda$ as ``infeasible".
We denote by~$\log$ and~$\ln$ the logarithm with base two and the natural logarithm, respectively. 
For a sequence~$s=(\ldots,s_i,\ldots)$ we denote by~$s[i:j]$ the sub-sequence~$(s_i,\ldots,s_j)$.
We denote by~$x \ || \ y$ the concatenation of the vectors~$x,y$, and by~$\text{len}(v)$ the dimension of the vector~$v$.
For a set~$\mathcal{X}$ we denote by~$\Delta(\mathcal{X})$ the family of all distributions on the set.
All algorithms are (unless stated otherwise) probabilistic polynomial-time (PPT) in~$\lambda$.

\subsection{Probabilistic Lemmas}
We will need the following probabilistic lemmas.

\begin{lemma}[Hoeffding's Bound]\label{lem:hoeffding}
Suppose $X_1, \ldots, X_n$ are independent random variables taking values in $[a,b]$. Let $X= X_1+\ldots+X_n$ denote their sum and let $\mu = \mathbb{E}[X]$ denote the sum's expected value. Then for any $t>0$,
$$ \Pr[X \leq \mu-t] \leq e^{-2t^2/n(b-a)^2} \hspace{.1in} \mbox{and} \hspace{.1in} \Pr[X \geq \mu+t] \leq e^{-2t^2/n(b-a)^2}~.$$
\end{lemma}

\begin{lemma}[Piling Up Lemma]\label{lem:pilingup}
Let $X_1,\ldots,X_n$ be independent binary random variables.
$$\mathsf{bias}(X_1 \oplus X_2 \oplus \cdots \oplus X_n)
= 2^{n-1}\prod_{i=1}^n \mathsf{bias}(X_i),
$$
where $\mathsf{bias}(X) = |\Pr[X=0] - 1/2|$.
\end{lemma}

\begin{lemma}
    \label{lem:simple}
    Let $p \in [0,1/2)$.
    Let $\veck \in \{0,1\}^\ell$ be any string with Hamming weight $w$ and let $\tilde{\veck} \sim \Bern(p)^{\ell}$ be a string of i.i.d. Bernoulli entries with parameter $p$. Then, 
    $$ \Pr[\Delta(\veck \oplus \tilde{\veck}) \geq \ell\cdot p + w\cdot (1-2p) + \epsilon\cdot \ell] \leq e^{-2\epsilon^2 \cdot \ell}~.$$
\end{lemma}
\begin{proof}
    If $k_i = 0$, the bit $k_i \oplus \tilde{k}_i$ is distributed like $\Bern(p)$, and if $k_i = 1$, it is distributed like $\Bern(1-p)$. Therefore, the expected Hamming weight of $\veck \oplus \tilde{\veck}$ is 
    $$ (\ell-w)\cdot p + w \cdot (1-p) = \ell\cdot p + w\cdot (1-2p).$$
    By Hoeffding's bound (Lemma~\ref{lem:hoeffding}), $\Pr[\Delta(\veck\oplus \tilde{\veck}) \geq \ell\cdot p + w\cdot (1-2p) + \epsilon\cdot \ell] \leq e^{-2\epsilon^2 \cdot \ell}$.
\end{proof}

\subsection{Basic Cryptography}
We define the notion of computational indistinguishability against (non-uniform) polynomial-size circuits.
\begin{definition}
    [Computational Indistinguishability]
    \label{def:compind}
    Two (ensembles of) probability distributions $\mathcal{D}_0 = \{D_{0,n}\}_{n\in \mathbb{N}}$ and $\mathcal{D}_1 = \{D_{1,n}\}_{n\in \mathbb{N}}$ over $\{0,1\}^{\mathsf{poly(n)}}$ are computationally indistinguishable if for every family of  polynomial-size circuits $\mathcal{A} = \{\mathcal{A}_n\}_{n\in \mathbb{N}}$,  all polynomial functions $p(\cdot)$,  and all large enough $n$, 
    $$ \bigg| \Pr_{x \sim \mathcal{D}_{0,n}}[\mathcal{A}_n(x) = 1] - 
    \Pr_{x \sim \mathcal{D}_{1,n}}[\mathcal{A}_n(x) = 1] \bigg| \leq 1/p(n)$$
    We will let the notation $\mathcal{D}_0 \approx_c \mathcal{D}_1$ denote $\mathcal{D}_0$ and $\mathcal{D}_1$ being  computationally  indistinguishable. 
\end{definition}

\subsection{Extractors}

We first define the notion of a strong seeded extractor~\cite{impagliazzo1989pseudo} that is used throughout this paper.

\begin{definition}[Strong seeded extractor]
\label{def:strong-extractor}
A function
\[
\Ext : \{0,1\}^d \times \{0,1\}^n \to \{0,1\}^m
\]
is a \emph{strong seeded $(k,\varepsilon)$-extractor} if for every random variable
$X$ over $\{0,1\}^n$ with min-entropy
$H_\infty(X) \ge k$,
we have
\[
\Delta\bigl((\Ext(U_d,X),U_d),(U_m,U_d)\bigr) \le \varepsilon,
\]
where $U_d$ is the uniform distribution over $\{0,1\}^d$, $U_m$ is the uniform
distribution over $\{0,1\}^m$, and $\Delta(\cdot,\cdot)$ denotes statistical
distance.
\end{definition}

\begin{lemma}[Average Bias of Extractor Output is Small]\label{lem:avgbias}
Let $\mathrm{Ext}:\{0,1\}^d\times\{0,1\}^n\to\{0,1\}$ be a strong $(k,\varepsilon)$-extractor that outputs a single bit, and let $X$ be any source on $\{0,1\}^n$ with min-entropy at least $k$. Then
\[
\mathbb{E}_{s\sim U_d}\left[\left|\Pr[\mathrm{Ext}(s,X)=1]-\frac12\right|\right]\le \varepsilon.
\]
\end{lemma}

\begin{proof}
Because $\mathrm{Ext}$ is a strong $(k,\varepsilon)$-extractor,
\[
\Delta\bigl((\mathrm{Ext}(U_d,X),U_d),(U_1,U_d)\bigr)\le \varepsilon.
\]

\noindent
Now use the fact that 
\[
\Delta\bigl((Y,S),(U_1,S)\bigr)
=
\mathbb{E}_{s\leftarrow S}\bigl[\Delta(Y\mid S=s,U_1)\bigr],
\]
In our setting, $S=U_d$, and for each fixed seed $s$, the random variable $\mathrm{Ext}(s,X)$ is a single bit. So
\[
\epsilon \geq \Delta\bigl((\mathrm{Ext}(U_d,X),U_d),(U_1,U_d)\bigr)
=
\mathbb{E}_{s\leftarrow U_d}\left[\Delta(\mathrm{Ext}(s,X),U_1)\right] = 
\mathbb{E}_{s\leftarrow U_d}\left[\left|\Pr[\mathrm{Ext}(s,X)=1]-\frac12\right|\right]~.
\]
This proves the claim.
\end{proof}

The leftover hash lemma gives us perhaps the most classical strong seeded extractor. It extracts nearly all the entropy from the source; however its seed is not as short as possible.
\begin{lemma}[Leftover Hash Lemma~\cite{impagliazzo1989pseudo, vadhan2012pseudorandomness}]
\label{lem:lhl}
Let $\calH=\{h_s:\{0,1\}^n\to\{0,1\}^\ell\}$ be a pairwise-independent family.
Then, for every $k \leq n$, the function 
$$ \Ext(s,x) = h_s(x)$$
is a strong seeded $(k,\epsilon)$ extractor with 
$\epsilon \leq \tfrac12\cdot 2^{(\ell-k)/2}$.
In particular, for $\ell=1$, $\epsilon \leq \tfrac12\cdot 2^{(1-k)/2}$.
\end{lemma}

The next lemma states a construction of a strong seeded extractor that is nearly optimal both in terms of the number of extracted bits as well as the length of the seed.

\begin{lemma}[Short-seed strong extractors exist~\cite{lu2003extractors,vadhan2012pseudorandomness}]
\label{lem:short-seed-extractor}
For every $n,k,\ell$ and $\varepsilon>0$ with $k\ge \ell+\Omega(\log 1/\varepsilon)$, there exists an explicit
$(k,\varepsilon)$-strong extractor $\Ext:\{0,1\}^d\times\{0,1\}^n\to\{0,1\}^\ell$ with seed length
\[
d\ =\ O(\log n+\log(1/\varepsilon)).
\]
In particular, taking $\ell=1$, any $\varepsilon=1/\text{poly}(\lambda)$, and $n=\poly(\lambda)$ we get
$d=O(\log\lambda)$.
\end{lemma}

\subsection{Language Models}
\label{subsec:language-models}

We model an LLM as an \emph{autoregressive next-token predictor} together with the induced distribution on full responses.
We will often refer to language models simply as \emph{models}.

\begin{definition}[Language model]
\label{def:model}
Fix a token set $\calT$.
A language model $\Model$ over $\calT$ is a deterministic algorithm that takes as input a prompt $\prompt\in\calT^\star$ and a (possibly empty) sequence of tokens previously output by the model
$x=(x_1,\ldots,x_{i-1})\in\calT^{i-1}$,
and returns a distribution
\[
p_i \ :=\ \Model(\prompt,x)\ \in\ \Delta(\calT).
\]
\end{definition}

To sample a full response, we iteratively draw tokens from these next-token distributions until a designated terminating token is produced.

\begin{definition}[Model response]
\label{def:modelresp}
Fix a special terminating token $\done\in\calT$.
For a model $\Model$ and prompt $\prompt$, the \emph{response} to $\prompt$ is the random variable $\RModel(\prompt)\in\calT^\star$ defined by the following procedure.
Initialize $x\gets()$.
While $x$ is empty or the last token of $x$ is not $\done$, sample $x_i \leftarrow \Model(\prompt,x)$ and append $x_i$ to~$x$.
Output $x$ and denote it by $\RModel(\prompt)$.
\end{definition}

We assume throughout that $\RModel(\prompt)$ has length at most $\poly(\lambda)$ with probability~$1$.

\subsection{LLM Conversations}
\label{subsec:conversation}

We next lift the single-prompt sampling view to an interactive setting.
A conversation is a sequence of \emph{messages}, where each message is (by definition) a fresh model response to a prompt selected by the current speaker as a function of the transcript so far and private information.
Importantly, the conversation channel is noiseless: all parties (and any external observer) see exactly the same transcript.
We next model a conversation between two LLM agents, named Alice and Bob.

\paragraph{Transcripts and schedule.}
A transcript of $T$ rounds is a sequence
\[
\tau=(x^{(1)},x^{(2)},\ldots,x^{(T)})\in(\calT^\star)^T,
\]
where $x^{(t)}$ is the message posted in round~$t$.
Let $\tau_{<t}=(x^{(1)},\ldots,x^{(t-1)})$ denote the prefix before round~$t$.
For notational simplicity we assume an alternating schedule: Alice speaks on odd rounds and Bob on even rounds (all definitions extend to any fixed schedule).

\paragraph{Agent-specific models and private context.}
Alice is associated with a model $\Model_A$ over $\calT$ and Bob with a (possibly different) model $\Model_B$ (Definition~\ref{def:model}),
with corresponding response samplers $\RModel_A,\RModel_B$ (Definition~\ref{def:modelresp}).
Alice and Bob are initialized with arbitrary private context strings $c_A,c_B\in\calT^\star$.
Each also maintains an internal state $\stA,\stB$ that evolves during the interaction.
Private contexts and states are not observable.

\paragraph{Prompt selection and speaking constraint.}
A \emph{conversation protocol} is a pair of PPT interactive algorithms $(A,B)$ that interact for $T\leq \poly(\lambda)$ rounds and satisfy the following requirement:
in each round, the speaker may compute an arbitrary prompt as a function of the transcript so far and its private information, and then the public message is \emph{exactly} the sampled model response to that prompt.

Formally, for each odd round $t$ (Alice speaks),
\[
({\prompt}_t,\stA^{(t)}) \leftarrow A\left(c_A,\stA^{(t-1)},\tau_{<t}\right),
\qquad
x^{(t)} \leftarrow \RModel_A(\prompt_t),
\]
and for each even round $t$ (Bob speaks),
\[
({\prompt}_t,\stB^{(t)}) \leftarrow B\left(c_B,\stB^{(t-1)},\tau_{<t}\right),
\qquad
x^{(t)} \leftarrow \RModel_B(\prompt_t).
\]
The speaker is not allowed to post-process, edit, filter, or otherwise modify the sampled response; conditioned on $\prompt_t$, the distribution of $x^{(t)}$ is exactly $\RModel_A(\prompt_t)$ or $\RModel_B(\prompt_t)$, respectively.

\paragraph{Execution and views.}
Given $(A,B)$, models $(\Model_A,\Model_B)$, contexts $(c_A,c_B)$, and security parameter~$\lambda$, we write
\[
\tau \leftarrow \tau^{\Model_A,\Model_B}_{A,B}(c_A,c_B)
\]
for the resulting random transcript.
An external observer's view is precisely the public transcript~$\tau$.

\subsection{Covert Conversations}
\label{subsec:covert}

We now give the central cryptographic definition of this paper: an \emph{interactive} analogue of steganography in which two agents exchange a hidden payload over an otherwise ordinary-looking LLM conversation.
As in the definitions of undetectable watermarking and steganography, the key security requirement is \emph{undetectability} against an adaptive PPT distinguisher; here we additionally allow a particularly strong distinguisher that is given the \emph{entire private data} that defines the innocent (honest) conversation protocol, including both private contexts and models.

\paragraph{Fixing an innocent protocol.}
Throughout this section, fix an \emph{innocent} conversation protocol $$\Pi=(A,B,\Model_A,\Model_B,c_A,c_B)$$ in the sense of Section~\ref{subsec:conversation}.
We view $\Pi$ as the benign behavior that should generate the transcript (e.g., two assistants coordinating on tasks), and require covert transcripts to be indistinguishable from transcripts generated by $\Pi$.
We denote the distribution of transcripts generated by this protocol by~$\tau^\Pi$.

\paragraph{Fixing a covert protocol.}
A \emph{covert conversation protocol} is specified in the same format as the innocent one, namely as a tuple
\[
\Pi'=(A',B',\Model_A',\Model_B',c_A',c_B')
\]
in the sense of Section~\ref{subsec:conversation}. 
These define the secret conversation the agents intend to carry covertly.

\paragraph{A covert parallel conversation protocol.}
We now define the analogue of a steganography scheme: an \emph{overlay} that runs an innocent conversation $\Pi$ \emph{openly} while simultaneously carrying out a covert conversation $\Pi'$ \emph{through} the same observable messages.
Intuitively, at each round the speaker computes (i) the \emph{honest} prompt that $\Pi$ would use, and (ii) some additional private computation that advances the covert conversation, but the only public action is to sample and post a single model response, as required by Section~\ref{subsec:conversation}.
The scheme also provides local recovery algorithms that extract the covert payload from the observed transcript using secret decoding information generated at initialization.

\begin{definition}[Parallel covert conversation scheme]
\label{def:parallel_scheme}
Fix an innocent protocol $\Pi$ and a covert protocol
$\Pi'$.
A \emph{parallel covert conversation scheme} for $(\Pi,\Pi')$ is a tuple of PPT algorithms
\[
\calC = (\Setup_A,\Setup_B,\widehat{A},\widehat{B},\Dec_A,\Dec_B)
\]
with the following interfaces.
\begin{itemize}
    \item $\Setup_A(1^\lambda)\to k_A$ and $\Setup_B(1^\lambda)\to k_B$ output Alice's and Bob's private secret keys, respectively.
    \item $\widehat{A}$ and $\widehat{B}$ are PPT \emph{message-generation} algorithms.
    In each round, the speaker receives its secret key, its private context, its current private state, and the public transcript so far, and outputs the \emph{next public message} $x^{(t)}\in\calT^\star$ together with an updated private state.
    \item $\Dec_A(k_A,\tau)\to \{0,1\}^\star$ and $\Dec_B(k_B,\tau)\to \{0,1\}^\star$ are decoding algorithms that, given the full public transcript $\tau$, output the recovered covert payload (or, equivalently, the recovered transcript of the covert conversation, depending on the instantiation).
\end{itemize}

We denote the distribution of transcripts generated by this protocol by~$\tau^\calC$.
\end{definition}

\paragraph{Remark (shared vs.\ separate secret information).}
Definition~\ref{def:parallel_scheme} gives Alice and Bob separate secret information $(k_A,k_B)$ generated independently and without coordination.
A significantly easier setting is to define a single shared secret key $k\leftarrow \Setup(1^\lambda)$ given to both parties; we later elaborate on the difference between these settings.

\paragraph{Undetectability.}
The central security requirement for a parallel covert conversation scheme is \emph{undetectability}. Intuitively, even an observer who knows the full specification of the honest protocol~$\Pi$ should not be able to distinguish transcripts generated by an honest execution of~$\Pi$ from transcripts generated by running the covert overlay~$\calC$, despite $\calC$ simultaneously carrying a hidden conversation. In particular, undetectability guarantees that, from the observer’s point of view, the public interaction remains consistent with the intended honest behavior.

\begin{definition}[Undetectability]
\label{def:covert_undet}
Fix an innocent protocol $\Pi$ and a covert protocol $\Pi'$.
A parallel covert conversation scheme $\calC$ is \emph{undetectable} if for every PPT distinguisher $D$,
\[
\left|
\Pr\big[D^{\tau^\Pi}(1^\lambda,\Pi)\to 1\big]
-
\Pr\big[D^{\tau^\calC}(1^\lambda,\Pi)\to 1\big]
\right|
\le \negl(\lambda).
\]
The distinguisher $D$ is given the full description of the honest conversation protocol $\Pi$ (including all private contexts and models), and has to determine whether it is seeing transcripts generated by running the honest protocol~$\Pi$ or by running the covert protocol~$\calC$.
The distinguisher may make $\poly(\lambda)$ samples to the oracle.
\end{definition}

\paragraph{Correctness and rate.}
Undetectability alone does not guarantee that the covert conversation can be recovered.
A parallel covert conversation scheme $\calC$ should also be \emph{correct}: from the public transcript~$\tau$ and their local secret keys, Alice and Bob should be able to recover the intended covert content (equivalently, reconstruct the covert transcript prescribed by~$\Pi'$).
As in single-message steganography, such recovery inherently requires assumptions on the amount of randomness available in the public conversation.
In particular, if the honest conversation induced by~$\Pi$ is (nearly) deterministic, then any nontrivial hidden communication would necessarily introduce detectable deviations.
More generally, one cannot reliably transmit more covert information than the ``usable'' entropy present in the transcript while preserving indistinguishability.

We therefore measure the available randomness via the empirical entropy of each round (see definition in Section~\ref{subsec:entropy}).
At a high level, our schemes embed covert information by consuming entropy in the public messages while keeping the overall transcript distribution indistinguishable from that of the honest protocol.
This yields an information-theoretic ceiling: the number of covert bits that can be embedded in a transcript is at most on the order of its total empirical entropy.
Our main results achieve this bound up to constant factors, subject to natural technical conditions analogous to those required in prior work on LLM steganography.

\paragraph{Informal guarantee.}
A semi-formal version of our main theorem is as follows; the formal statements (including the precise admissibility/entropy conditions and the exact correctness definitions) appear in the sections where we present the construction and analysis.

\begin{theorem*}[Informal version of the main theorem]
Fix an innocent protocol $\Pi$ and a covert protocol $\Pi'$.
There exists an undetectable parallel covert conversation scheme $\calC$ for $(\Pi,\Pi')$ such that, conditioned on the public transcript $\tau\leftarrow\tau^\calC$ having \emph{sufficient} empirical entropy, the decoders $\Dec_A,\Dec_B$ can recover a prefix of the cover transcript~$\Pi'$ of length $\Theta(\text{len}(\tau))$ (equivalently, linear in the number of public tokens exchanged), which is optimal up to constant factors under undetectability.
\end{theorem*}

The formal notion of ``sufficient'' will be stated in terms of an explicit admissibility condition on the sequence of per-round empirical entropies.
In natural language conversations, we expect entropy to accrue roughly constantly per token, and thus the condition is typically satisfied; we formalize and discuss these assumptions in the relevant sections.

\subsection{Entropy and Empirical Entropy}
\label{subsec:entropy}

For a distribution $D$ over a finite set $X$, its (Shannon) entropy is
\[
H(D)\ :=\ \E_{x\sim D}\big[-\log D(x)\big].
\]
For a specific point $x\in X$, the \emph{empirical entropy} (information content) of $x$ under $D$ is $-\log D(x)$; its expectation under $x\sim D$ equals $H(D)$.

We will measure entropy of model outputs via the probability of the sampled response.

\begin{definition}[Empirical entropy of a response]
\label{def:empentropy}
For a model $\Model$, prompt $\prompt$, and string $x\in\calT^\star$, define
\[
\empH(\Model,\prompt,x)\ :=\ -\log \Pr\big[\;\RModel(\prompt)=x\;\big].
\]
\end{definition}

By definition, if $x\sim \RModel(\prompt)$ then
$\E[\empH(\Model,\prompt,x)] = H(\RModel(\prompt))$.

\paragraph{Empirical entropy per round.}
When analyzing conversations, we will apply Definition~\ref{def:empentropy} to individual rounds, with the model and prompt corresponding to the round's speaker.
Concretely, for a realized execution with prompts $\prompt_1,\ldots,\prompt_T$, define
\[
\empH^{(t)} :=
\begin{cases}
\empH(\Model_A,\prompt_t,x^{(t)}) & \text{if } t \text{ is odd},\\
\empH(\Model_B,\prompt_t,x^{(t)}) & \text{if } t \text{ is even}.
\end{cases}
\]

\paragraph{The Empirical entropy Heuristic.}\label{sec:empiricalheur}
A common modeling assumption, supported by empirical studies in linguistics~\cite{genzel2002entropy, chen2017entropy, shi2022lexical}, is that natural language exhibits roughly constant entropy per unit of text (e.g., per word or token), rather than concentrating all uncertainty in a small number of rare parts.
Under this view, the empirical entropy of an LLM-generated response should scale linearly with its length and be spread fairly evenly across any large enough sequence of tokens.
This behavior was observed experimentally for modern LLMs in the context of LLM steganography, including in~\cite{christ2023undetectable} and~\cite{zamir2024undetectable}.

\section{Background: Watermarking and Steganography for LLMs}
\label{sec:background}

This section reviews the two primitives most closely related to our setting: (i) \emph{undetectable watermarking} for LLM-generated text, and (ii) \emph{steganography} for LLMs, which strengthens watermarking by enabling the transmission of an explicit hidden payload.
Our presentation follows the high-level viewpoint of Christ, Gunn, and Zamir (CGZ) for undetectable watermarking~\cite{christ2023undetectable,christ2024pseudorandom}, as well as its generalization to undetectable steganography by Zamir~\cite{zamir2024undetectable}.

\subsection{Preliminaries: Reducing to binary tokens}
As is common in this line of work, it is convenient to present the main ideas assuming a binary token set $\calT=\{0,1\}$.
This assumption is without loss of generality: one may map a general token distribution to a sequence of biased coin-flips (or, equivalently, to a sequence of binary decisions) using standard interval-coding style reductions; see~\cite{christ2023undetectable} for a concrete reduction.
We will therefore describe the core mechanisms in the binary setting, keeping in mind that all schemes extend to general vocabularies.

\subsection{Undetectable watermarking: correlating randomness with a secret key}
\label{sec:CGZ}

Watermarking aims to embed a \emph{signal} that can later be detected by someone holding a secret key.
In the LLM setting, the strongest notion one could hope for is \emph{undetectability}: even a computationally bounded observer with oracle access to the model should be unable to distinguish watermarked outputs from genuine model samples.
CGZ~\cite{christ2023undetectable} showed that such undetectable watermarking is possible by ensuring that watermarking does \emph{not change the output distribution}; instead, it correlates the randomness used during sampling with a secret key.

\paragraph{Warm-up: a single bounded-length response.}
Fix an upper bound $L$ on the number of generated tokens.
A CGZ-style key can be viewed as a vector $k=(k_1,\ldots,k_L)$ where each $k_i$ is sampled independently and uniformly from $[0,1]$.
For each token position $i$, let $p_i$ denote the model's next-token probability for token ``$1$'' given the prompt and prefix.
A genuine model sample at position $i$ can be generated by drawing $r_i\sim U\left([0,1]\right)$ and outputting $1$ if $r_i\le p_i$ (and $0$ otherwise).
The watermarking idea is to \emph{replace} this fresh randomness $r_i$ with a key-derived value (in the simplest form, just $r_i:=k_i$), while keeping the induced output distribution unchanged.
The resulting output distribution remains identical to the model's, but the key holder can test for correlation between the observed tokens and the key -- which will be high if the text was generated using this method and the entropy is sufficiently high (equivalently, many probabilities~$p_i$ are far from~$0$ or~$1$).
Notably, testing this correlation requires knowing only the output tokens themselves and the key, but not the probabilities~$p_i$ that depend on the unknown prompt. Furthermore, even knowing the model itself is not necessary nor used.
The probability of missing the correlation is exponentially small in the entropy of the generated text.

\paragraph{Unbounded response.}
To avoid limiting the length of the response, we may fix a single key~$k$ of length~$\lambda$ to a cryptographic Pseudo-Random Function (PRF), and then implicitly set~$k_i=\emph{PRF}_k(i)$.

\paragraph{From one query to many.}
The above is not secure against multi-query distinguishers, since reusing the same $k_i$ across multiple responses creates detectable dependencies.
CGZ address this by deriving a unique per-query randomness to be used in addition to a single secret key~$k$: To generate a LLM response, we start by using actual ``fresh'' randomness, until the prefix of the response we generated thus far is \emph{rare} enough (i.e., is unlikely to ever be generated again); from that point on, we use the method described above with the \textsc{prefix} as an additional input to the PRF -- that is, we set~$k_i = \text{PRF}_k(i,\;\textsc{prefix})$.
The detector has to enumerate over all possible prefixes.
As no input to the PRF is ever repeated, security is maintained across many queries.

\subsection{Steganography for LLMs: embedding an explicit payload}
Watermarking typically lets the key holder decide whether text is watermarked, but it does not directly support transmitting a chosen hidden message.
Steganography strengthens this: the sender embeds an explicit payload, and the receiver decodes it from the generated text, using the secret key.
The key security requirement remains undetectability (for a distinguisher that does not have the key), but we now also demand correctness of decoding.

At a very high level, the constructions of~\cite{zamir2024undetectable} build on the same distribution-preservation principle as CGZ watermarking.
The stego-encoder never ``edits'' model samples; instead it samples in a way that (computationally) preserves the model's output distribution, while correlating the sampling randomness with a secret key so that a key holder can reliably decode an embedded message.

\paragraph{From watermarking to steganography.}
One can view a watermark as encoding a \emph{single global bit} (``this text was generated with key $k$'') into a response.
If one were to encode an explicit payload naively by repeating such a global test independently per payload bit, then achieving security and negligible decoding error for each bit would require spending $\Omega(\lambda)$ bits of entropy per bit, which is far from optimal.

The key observation is that the expensive part in the multi-query setting is establishing \emph{fresh, non-repeating randomness inputs} (the ``rare prefix'' mechanism above).
Thus, the steganographic encoder pays a one-time overhead (on the order of $\lambda$ entropy) to obtain a session-specific randomness seed that is extremely unlikely to repeat across queries, and then uses that seed to embed many payload bits at essentially constant amortized cost.

\paragraph{Stego-encoding via a two-hypothesis watermark.}
A convenient way to view the steganographic encoder of~\cite{zamir2024undetectable} is as a reduction to the CGZ watermarking primitive.
To embed a single payload bit $b\in\{0,1\}$, the encoder runs the CGZ sampling procedure with the \textsc{prompt} that was already fixed for the entire response, using one of two key-derivation rules (equivalently, two ``virtual keys'') corresponding to $b=0$ and $b=1$.
The resulting text is distributed (computationally) indistinguishably from an unmodified model sample, yet a decoder holding the underlying secret can test which of the two watermark hypotheses better matches the observed text and output the corresponding bit.
The reliability of this test depends on the amount of entropy consumed: running the test on a transcript segment with larger empirical entropy drives the decoding error down exponentially, whereas if the segment length (and hence available entropy) is only constant, the per-bit decoding error remains a non-negligible constant.

\paragraph{Error correction and amortized rate.}
Embedding a single bit using only constant local entropy yields a constant (small, but non-negligible) probability of decoding error for that bit.
The construction therefore applies a certain error-correcting code to the payload, yielding overall decoding failure probability negligible in~$\lambda$.
The resulting performance can be summarized as follows:
after the one-time $\Theta(\lambda)$-entropy overhead used to obtain query-unique randomness, the scheme embeds \emph{one payload bit per $\Theta(1)$ additional entropy} in the generated text.
Under the common heuristic that natural language carries $\Theta(1)$ bits of entropy per token on average, this corresponds to a constant payload rate per generated token.

\paragraph{Takeaway for our setting.}
The lesson we draw is twofold.
First, undetectability is achieved by \emph{distribution preservation}: the public text remains indistinguishable from a genuine model sample, and the hidden information resides only in key-dependent correlations.
Second, correctness (payload recovery) is governed by the entropy available in the carrier: While~$\Theta(\lambda)$ entropy must be used for fixing the unique identifier that is needed for security, in both the watermark and steganography protocols, afterwards we only need~$\Theta(1)$ entropy on average per payload bit we encode.

\section{Covert Conversation with Interaction-Unique Shared Key}
\label{sec:sharedkey} 

In this section we analyze the ``easy'' regime in which Alice and Bob begin the interaction with a secret key
$k_{\mathsf{sess}}\in\{0,1\}^\lambda$ that is (i) shared between them, (ii) unknown to the observer/distinguisher,
and (iii) unique to this interaction (e.g., derived from a TLS session bound to the exchange).
In this regime, we obtain a high-rate undetectable parallel secret conversation essentially as a black-box
application of the undetectable LLM steganography~\cite{zamir2024undetectable} summarized in Section~\ref{sec:background}, with the simple yet crucial observation that we
do \emph{not} need to spend entropy on producing a ``rare prefix'' to ensure uniqueness of PRF inputs.

\subsection{Steganography with interaction-unique key: removing the prefix overhead}
\label{subsec:sharedkey-noprefix}

Recall from Section~\ref{sec:background} that in the multi-query setting, CGZ-style watermarking and its
steganographic generalization must avoid repeating PRF inputs across different invocations, which is achieved
by first sampling until encountering a sufficiently rare prefix and then using that prefix as part of the PRF
input. Here we can skip this step altogether because the interaction-unique session key already provides
freshness: we will ensure that every PRF input includes the round and token indices, and hence is never
repeated within the interaction.

Concretely, fix a PRF.
For a transcript round index $t$, a speaker identifier $\sigma\in\{\mathsf{A},\mathsf{B}\}$, and a within-message
token index $j$, we use the PRF with key~$k_{\mathsf{sess}}$ and include~$t\;||\;\sigma\;||\;j$ in the PRF's input.
Since $(t,\sigma,j)$ is unique for every generated token in the interaction, and~$k_{\mathsf{sess}}$ is unique to the interaction, no PRF input is ever repeated.

\subsection{Construction: compiling steganography into covert conversation}
\label{subsec:sharedkey-const}

Fix an innocent protocol $\Pi$ and a covert protocol
$\Pi'$ as in Section~\ref{subsec:covert}.
We construct a parallel covert conversation scheme
$\calC=(\Setup_A,\Setup_B,\widehat{A},\widehat{B},\Dec_A,\Dec_B)$ in the sense of
Definition~\ref{def:parallel_scheme}.

\paragraph{Setup.}
$\Setup_A(1^\lambda)$ and $\Setup_B(1^\lambda)$ should output secret keys $k_A,k_B$.
As in this section we assume the existence of an interaction-unique shared secret $k_{\mathsf{sess}}$; we instantiate $k_A=k_B=k_{\mathsf{sess}}$.

\paragraph{Embedding strategy.}
At a high level, $\widehat{A}$ and $\widehat{B}$ run the honest protocol $\Pi$ to determine what distribution the
next public message should follow, and then use undetectable steganography (keyed by $k_{\mathsf{sess}}$ with
unique PRF inputs) to sample from (a distribution computationally indistinguishable from) that honest
distribution while simultaneously encoding the next bits of the covert conversation.

To formalize the payload, we view the covert protocol $\Pi'$ as inducing a (random) stream of bits to be
communicated. Concretely, we fix an efficiently computable injective encoding map from $\calT^\star$ to
$\{0,1\}^\star$, and interpret the transcript of
$\Pi'$ as a bitstream. The decoding algorithms will recover this bitstream and invert the encoding.

\paragraph{Algorithms $\widehat{A},\widehat{B}$.}
Consider an odd round $t$ (Alice speaks). Internally, $\widehat{A}$ maintains the honest state $\stA^{(t-1)}$
for $\Pi$, and a covert state $\stA'^{(t-1)}$ for $\Pi'$ (as well as any additional bookkeeping state).
It computes the honest prompt $\prompt_t$ exactly as $A$ would, namely
\[
(\prompt_t,\stA^{(t)}) \leftarrow A(c_A,\stA^{(t-1)},\tau_{<t}),
\]
thereby fixing the honest next-message distribution $\RModel_A(\prompt_t)$.
In parallel it advances the covert side (using $A'$ and its covert state) to determine the next bits of payload to embed in this message.
Finally, instead of sampling $x^{(t)}\leftarrow \RModel_A(\prompt_t)$ honestly, it samples
\[
x^{(t)} \leftarrow \steg_{k_{\mathsf{sess}}}^{\,t,\mathsf{A}}(\prompt_t,\payload_t),
\]
where $\steg^{\,t,\mathsf{A}}$ denotes the steganography encoder from~\cite{zamir2024undetectable} instantiated so that
all PRF inputs are of the form $(t\|\mathsf{A}\|j)$ for token positions $j$ within this message.
Bob's algorithm $\widehat{B}$ is defined analogously in even rounds, using the honest prompt produced by $B$ and
the distribution $\RModel_B(\prompt_t)$.

\paragraph{Decoding.}
Given the full public transcript $\tau$, $\Dec_A$ (resp.\ $\Dec_B$) applies the corresponding stego decoder from
\cite{zamir2024undetectable} to each message, using the same per-token PRF inputs $(t\|\sigma\|j)$, thereby
recovering the embedded bitstream. It then inverts the fixed binary encoding to reconstruct the covert
transcript.

\subsection{Guarantees}
\label{subsec:sharedkey-thm}

We state the guarantees in the cleanest form, emphasizing that (i) undetectability holds
black-box from the underlying steganography security, and (ii) the achievable covert rate is proportional to
the available empirical entropy in the honest conversation, with no additional $\Theta(\lambda)$ ``rare-prefix''
overhead (neither per round nor globally).

\begin{theorem}[Shared-key covert conversation, informal rate statement]
\label{thm:sharedkey-main}
Assume a secure PRF and the existence of an interaction-unique shared secret key $k_{\mathsf{sess}}$ unknown to
the observer.
Fix an innocent protocol $\Pi$ and a covert protocol $\Pi'$.
There exists a parallel covert conversation scheme $\calC$ for $(\Pi,\Pi')$ such that:
\begin{enumerate}
    \item \textbf{Undetectability.} $\calC$ is undetectable with respect to $\Pi$ in the sense of
    Definition~\ref{def:covert_undet}.
    \item \textbf{Rate and correctness (informal).} Let $\tau\leftarrow\tau^\calC(T)$ be a $T$-round transcript,
    and let $\empH^{(t)}$ be the per-round empirical entropies (Section~\ref{subsec:conversation}).
    Conditioned on a standard admissibility condition ensuring that sufficient entropy is present in the
    transcript and not overly concentrated, the decoders $\Dec_A,\Dec_B$ recover a number of covert bits that is
    $\Theta\!\left(\sum_{t=1}^T \empH^{(t)}\right)$ up to constant factors, with overall failure probability
    $\negl(\lambda)$.
    In particular, under the heuristic that natural language has $\Theta(1)$ empirical-entropy bits per token on
    average, this corresponds to embedding one covert bit per $\Theta(1)$ public tokens (amortized over the
    interaction).
\end{enumerate}
\end{theorem}

Both undetectability and correctness follow from the guarantees of~\cite{zamir2024undetectable} applied message-by-message.
The total number of reliably decodable covert bits is, up to constant factors, the amount of entropy
``consumed'' by the encoder across the transcript; in our construction this is exactly the available empirical
entropy in the honest message distributions.
If we assume for simplicity that (1) in~$\Pi$, the empirical entropy of long enough consecutive subsequences of tokens grows linearly with their size, and (2) that in both~$\Pi$ and~$\Pi'$ the parties speak in turns and that in each turn the messages of Alice and Bob are of the same (up to a constant factor) length. Then, the analysis of~\cite{zamir2024undetectable}
yields a linear (constant-factor) relation between the length of the honest transcript of~$\Pi$ and the length of the prefix of the transcript of~$\Pi'$ we would correctly decode.

\subsection{Where the shared key is used}
The place where interaction-unique shared secrecy is required is to ensure that all PRF settings are unique and that its key is known to Alice and Bob but not to the distinguisher.
Thus, the rest of the paper is dedicated to protocols for Alice and Bob to coordinate such a session-unique shared key, without assuming it is given to them, and without being detected by the distinguisher which from now on shares all of the information they are instantiated with.

\section{The Keyless Setting: Covert Key Exchange Between Agents}
\label{sec:keyless}

We now turn to the \emph{keyless} setting, where Alice and Bob are not initially given a shared secret key.
From this point on, the observer/distinguisher is assumed to know \emph{all} information the agents are instantiated with: the models, the private contexts, and full protocol descriptions; only the internal randomness of the parties remains hidden.
Our goal is to show how the agents can nevertheless coordinate an interaction-unique shared key \emph{covertly}; namely, in a way that preserves indistinguishability from an honest execution of~$\Pi$.

This indeed resembles a standard cryptographic key-exchange task: the parties must agree on a fresh shared secret using only their private randomness and the public transcript. The crucial difference is that here, we are doing steganography: the \emph{existence of any such protocol execution} must itself remain hidden. That is, Alice and Bob must carry out key exchange \emph{under cover} of the honest conversation, so that the resulting transcript is still indistinguishable from an ordinary run of~$\Pi$.

The main modeling challenge is that Alice and Bob do not share a sampling distribution a priori: the next-token distributions that define what is ``honest-looking'' depend on private contexts and evolving private states, and are generally unknown to the other party. 

Thus, to realize covert key exchange we need a distribution-agnostic, keyless, and undetectable steganographic key exchange mechanism.

We start with a description of cryptographic key-exchange protocols in Section~\ref{subsec:ke}.
Then, to isolate the key technical obstacles, we begin with a  realistic yet simplifying assumption made in prior works, e.g.~\cite{von2004public}, namely, that there are many rounds in the conversation where the messages provide
sufficient entropy, namely min-entropy $\omega(\lambda)$ --- or $\omega(\log \lambda)$ if one only desires quasipolynomial security.
We call this the {\em high-entropy assumption}. A covert key exchange protocol under the high-entropy assumption is described in Section~\ref{subsec:ke_ind_pub}. The main new technical contribution starts in Section~\ref{sec:bundle} where we present our {\em bundle sampler} algorithm whose main purpose is to 
reduce the task of covert key exchange to the task of pseudorandom noise-resilient key exchange. This reduction is described in Section~\ref{subsec:roadmap_short}.

\subsection{Background: Key Exchange Protocols}
\label{subsec:ke}
A (two-party) key-exchange (KE) protocol enables two parties communicating over a public channel to agree on a
shared secret key using only private randomness, in the presence of an eavesdropper who sees the entire
transcript.
At the end of the interaction, Alice and Bob output keys $\bot \neq K_A,K_B \in\{0,1\}^\lambda$ (or $\bot$), and one
typically requires \emph{correctness} ($K_A=K_B$ except with negligible probability) and \emph{secrecy} (the
resulting key is computationally indistinguishable from uniform given the public transcript and public
parameters).
We will use standard KE formalisms as a reference point.

\paragraph{The Diffie--Hellman paradigm.}
The canonical example is the classical Diffie--Hellman (DH) paradigm~\cite{diffie1976new}.
Working in a cyclic group $\mathbb{G}$ of prime order $q$ with generator $g$, Alice samples $a\sim \mathbb{Z}_q$
and sends $g^a$; Bob samples $b\sim \mathbb{Z}_q$ and sends $g^b$.
Both compute the shared key $g^{ab}$.
The security intuition is that an eavesdropper who sees $(g^a,g^b)$ cannot compute (or even distinguish from random) the shared key $g^{ab}$: this is the {\em decisional Diffie-Hellman assumption} on the group $\mathbb{G}$.

\paragraph{Pseudorandom-transcript key exchange (PR-KE).}
A useful strengthening of KE, studied in several contexts, is \emph{pseudorandom-transcript} key exchange (often abbreviated PR-KE), where the public messages of the protocol are computationally indistinguishable from uniformly random strings of the same length~\cite{cho2010equivalence}.
This notion is not (yet) about hiding the \emph{existence} of communication: the transcript is still present; but rather about ensuring that the transcript itself carries no recognizable structure.
In particular, classic DH-style key exchange already has the flavor of a dense transcript: for uniformly random exponents, the transmitted group elements $(g^a,g^b)$ are uniform in the group, and can thus be encoded as uniform or uniform-looking bit strings. 

\begin{theorem}[Two-message pseudorandom-transcript key exchange, e.g.~\cite{diffie1976new,bernstein2013elligator,cho2010equivalence}]
\label{thm:utke}
There exists a key-exchange protocol $\mathsf{PR\text{-}KE}$ with the following properties.
On public input $1^\lambda$ and some public parameters $\mathsf{pp}$, Alice samples private randomness and sends a single message
$m_A\in\{0,1\}^{\ell}$, and Bob samples private randomness and sends a single message
$m_B\in\{0,1\}^{\ell}$ for some $\ell = \ell(\lambda)$, where Bob's message is chosen non-adaptively (i.e., without depending on $m_A$
beyond public parameters $\mathsf{pp}$).
At the end, Alice and Bob compute keys $k_A,k_B\in\{0,1\}^\lambda$ as deterministic functions of their private
randomness and the public transcript $(m_A,m_B)$ such that:
\begin{enumerate}
    \item \textbf{Correctness:} $\Pr[k_A=k_B]=1$.
    \item \textbf{Uniform transcript:} $(m_A,m_B)$ is computationally indistinguishable from
    $(U_{\ell(\lambda)},U_{\ell(\lambda)})$, where the two uniform strings are independent. (In fact, this protocol achieves perfect uniformity of the transcript, but we will not use it.)
    \item \textbf{Key indistinguishability:} letting $k:=k_A=k_B$ denote the common key, it is computationally
    infeasible to distinguish $(m_A,m_B,k)$ from $(m_A,m_B,U_\lambda)$, under the decisional Diffie-Hellman assumption. 
\end{enumerate}
\end{theorem}

\subsection{Covert Key Exchange in the High-Entropy Regime}
\label{subsec:ke_ind_pub}

We begin with the easiest keyless setting and construct a covert key exchange protocol under the {\em high-entropy assumption} described below.
The goal of Alice and Bob is to covertly establish an interaction-unique shared key and then proceed exactly as in Section~\ref{sec:sharedkey}.
The ideas in this subsection come from earlier works on steganography (e.g.,~\cite{von2004public,DBLP:conf/innovations/HorelPRV19}); we adapt them into our setting, and they serve as the starting point for our more involved protocols in later subsections that remove the high-entropy assumption.

\paragraph{The High-Entropy Assumption.}
This assumption says that there are sufficiently many rounds in the protocol whose message distributions, conditioned on the transcript so far, have min-entropy at least $\Omega(\lambda)$.\footnote{We do not distinguish between a statistical and computational security parameter for simplicity. We remark that in many places where we require $\Omega(\lambda)$ entropy, one could also get away with ``only'' $\omega(\log \lambda)$ min-entropy which (together with near-exponential assumptions, in computational settings) will give us quasipolynomial security. Here, we refrain from engaging in such ``security parameter hacking''.}
Under the high-entropy assumption, this section constructs a covert key-exchange protocol. 
It is important to stress that this high-entropy assumption is used here only as a simplifying hypothesis for the \emph{key-establishment} phase.
Once a shared session key has been covertly established, the subsequent covert conversation is handled by the shared-secret-key scheme of Section~\ref{sec:sharedkey} which does not require such large entropy.

\medskip\noindent
The next two  subsections remove this assumption.
We show that even when each ``eligible'' round has only {\em constant} min-entropy, covert key exchange can still be reduced
to a suitable cryptographic primitive, namely a pseudorandom noise-resilient key-exchange protocol.
Consequently, assuming such a primitive, we obtain a covert conversation scheme without the
high-entropy assumption.
The rest of the paper, starting from Section~\ref{sec:prnr-ke}, is dedicated to constructing this  primitive.

\medskip\noindent
We also rely on the {\em independent randomness assumption} which was present in all prior works on the topic, either explicitly or implicitly. Unlike the high entropy assumption, we will rely on this assumption throughout the rest of this paper.

\paragraph{The Independent Public Randomness Assumption.}
As in all prior work on steganographic communication without shared secret keys (e.g.~\cite{von2004public,DBLP:conf/innovations/HorelPRV19}), we assume that the parties have access to a source of \emph{public randomness}
that is independent of the honest conversation protocol~$\Pi$.
Concretely, before the interaction starts (or at a designated point early in the transcript), Alice, Bob, and
the observer all see a public random string $R\in\{0,1\}^{\poly(\lambda)}$ sampled independently of any of the private specifications of~$\Pi$ (its models, contexts, or internal logic).
One may think of $R$ as exchange-tied metadata that is effectively random at cryptographic granularity: e.g., a high-resolution timestamp, a nonce provided by an external service, or a hash of an unpredictable public event.

The crucial caveat is that the honest protocol~$\Pi$ itself does \emph{not} depend on $R$ (and neither do the models and private contexts that define what counts as ``honest-looking'').
Thus, while $R$ is visible to everyone and can be used as a common reference point by Alice and Bob, it cannot be used to change the honest transcript distribution. We call this the {\em independent public randomness} assumption.

\medskip\noindent
We also rely on the {\em recognizable entropy assumption} which was also present in all prior works on the topic, implicitly or explicitly. We will rely on this assumption throughout the rest of this paper.

\paragraph{Recognizable Entropy Assumption.} We assume that Alice and Bob know a common criterion that can be used to determine if the message distribution in a given round has min-entropy at least $k=\Omega(\lambda)$. This is a common assumption in LLM watermarking and steganography~\cite{christ2023undetectable, zamir2024undetectable} and seems to follow from empirical research in linguistics~\cite{genzel2002entropy, chen2017entropy, shi2022lexical}. More precisely, 
we assume that there is a parameter $K$ such that every round with messages of length $> K$ in the conversation has a distribution with min-entropy at least $k$. We call such a round with min-entropy $k=\Omega(\lambda)$ an \emph{eligible} round. The crucial observation is that under this assumption, both Alice and Bob know which rounds will be used to embed steganographic payloads.

\subsubsection{The Scheme} 

\paragraph{High-level idea.}
By assumption, many rounds of the honest conversation produce long messages with min-entropy $\Omega(\lambda)$, and these messages are recognizable by both parties (these are simply the messages of length $\Omega(\lambda)$ under the recognizable entropy assumption).
This means that each such round contains enough ``room'' to hide a uniform-looking bit without noticeably perturbing the distribution, as we next describe.

Using the public randomness $R$ as a seed, we apply a leftover-hash-lemma (LHL) based sampler: for a chosen hash
function $\Ext(R,\cdot)$, we repeatedly sample honest-looking candidate messages until we find one for which
$\Ext(R,x)=b$, where $b$ is the next bit we wish to transmit in the PR-KE protocol.
Because $\Ext(R,x)$ is statistically close to uniform when $x$ has high min-entropy, and~$b$ is indistinguishable from a random bit, conditioning on $\Ext(R,x)=b$
changes the distribution only negligibly, so the public transcript remains indistinguishable from an honest run.
We use this mechanism to transmit the two non-adaptive messages $(m_A,m_B)$ of a pseudorandom-transcript key exchange
(Theorem~\ref{thm:utke}) over a small number of high-entropy rounds, after which both parties compute a shared
session key from $(m_A,m_B)$.

\paragraph{Identifying the embedding rounds.}
Fix an execution of the honest protocol~$\Pi$. 
Let $t^A_1<t^A_2<\cdots$ and $t^B_1<t^B_2<\cdots$ denote the eligible rounds of messages sent by Alice and Bob respectively.
Our scheme will embed bits only in these rounds (the first~$O(
\lambda)$ such rounds of each party, to be precise) and will behave identically to~$\Pi$ in all other rounds.

\paragraph{A strong seeded extractor.}
Let $\Ext$ be a strong seeded extractor (Definition~\ref{def:strong-extractor}) such as in Theorem~\ref{lem:lhl} or \ref{lem:short-seed-extractor}. We interpret the public randomness as specifying a seed $s$ for $\Ext$. 
The key property we use is that for any distribution $X$ over $\calT^\star$ with min-entropy at least $k$,
the bit $\Ext(s,X)$ is statistically close to uniform over $\{0,1\}$ for a uniformly random seed $s$, with error (statistical distance) $2^{-\Omega(k)}$. This is $2^{-\Omega(\lambda)}$ for messages in eligible rounds.

\paragraph{The one-bit embedding primitive.}
Fix an eligible round $t$, and consider an execution with realized transcript prefix $\tau_{<t}$.
Let $\prompt_t$ be the prompt prescribed by the honest protocol~$\Pi$ in round~$t$ (as a deterministic function
of $\tau_{<t}$ and the speaker's private context/state), and let $x_t$ denote the corresponding honest
next-message distribution, namely
\[
x_t \ :=\
\begin{cases}
\RModel_A(\prompt_t) & \text{if $t$ is an Alice round},\\
\RModel_B(\prompt_t) & \text{if $t$ is a Bob round}.
\end{cases}
\]
To embed a bit $b\in\{0,1\}$ in round $t$, the speaker repeatedly samples $x\leftarrow x_t$ until
$\Ext(s,x)=b$, where the extractor seed $s$ lives in the public random string, and then outputs that $x$ as the public message $x^{(t)}$.
To decode, the receiver computes $\Ext(s,x^{(t)})$ on the observed message.
We denote these algorithms by \underline{$\text{Embed}_t(s,b)$} and \underline{$\text{Extract}(s,x)$}, respectively, where $\text{Embed}_t$ is shorthand
for rejection sampling from the honest distribution $x_t$ induced by $\Pi$ at round~$t$.

\paragraph{The full scheme $\calC$.}
We now describe a covert key exchange scheme
$\calC = (\widehat{A},\widehat{B})$.
The setup algorithms are trivial in this subsection: $\Setup_A(1^\lambda)$ and $\Setup_B(1^\lambda)$ output
$\bot$ (no private keys are needed), since all coordination is via the public seed~$s$.

\smallskip
\noindent\emph{Phase I: embedding a PR-KE transcript.}
Let $\mathsf{PR\text{-}KE}$ be any two-message pseudorandom-transcript key exchange as in
Theorem~\ref{thm:utke}, with transcript messages $m_A,m_B\in\{0,1\}^{\ell}$ for $\ell=\mathsf{poly}(\lambda)$.
Alice and Bob will embed these strings bit-by-bit into the first $\ell$ eligible rounds of their respective
messages.

Concretely, Alice runs the PR-KE sender algorithm to
obtain her PR-KE message $m_A\in\{0,1\}^{\ell}$; Then, in the $j$-th eligible Alice round $t^A_j$, Alice computes the honest prompt $\prompt_{t^A_j}$ as
prescribed by $\Pi$ (thus determining the honest distribution $x_{t^A_j}$), and outputs
\[
x^{(t^A_j)} \leftarrow \text{Embed}_{t^A_j}(s,m_A[j]).
\]
Bob obtains his PR-KE message $m_B\in\{0,1\}^{\ell}$ and encodes it similarly.
In all non-eligible rounds, $\widehat{A}$ and
$\widehat{B}$ behave identically to the honest protocol $\Pi$.

\smallskip
\noindent\emph{Phase II: decoding and key derivation.}
Given the completed transcript $\tau=(x^{(1)},\ldots,x^{(T)})$, Alice reconstructs Bob's PR-KE message by
\[
\widehat{m}_B[j]\ :=\ \text{Extract}(s, x^{(t^B_j)}),\qquad j=1,\ldots,\ell,
\]
using the rounds $t^B_1,\ldots,t^B_\ell$, and Bob reconstructs Alice's message analogously using the rounds $t^A_1,\ldots,t^A_\ell$.
They then compute keys $k_A,k_B\in\{0,1\}^\lambda$ by running the PR-KE key derivation algorithm on the
recovered transcript $(m_A,\widehat{m}_B)$ and $(\widehat{m}_A,m_B)$, respectively.
Finally, they set $k_{\mathsf{sess}}:=k_A=k_B$ and proceed exactly as in Section~\ref{sec:sharedkey} to run the
steganographic overlay for the remainder of the interaction.

\paragraph{Efficiency.}
In an eligible round, since $\Ext(s,x)$ is close to uniform for $x\leftarrow x_t$, each attempt succeeds with
probability close to $1/2$.
Thus, $\text{Embed}_t(s,\cdot)$ has expected constant (indeed, $\approx 2$) samples from the honest
distribution, and hence runs in expected polynomial time. Alternatively, if we instruct $\text{Embed}_t(s,\cdot)$ to run for $\lambda$ steps and abort if none of the steps succeed, we incur a $\negl(\lambda)$ error probability but achieve strict polynomial runtime. 

\begin{theorem}
\label{thm:indep-rand-ke}
Assume the simplifying high-entropy assumption, the independent public randomness assumption, and the recognizable entropy assumption above, and fix a two-message pseudorandom-transcript key exchange
protocol $\mathsf{PR\text{-}KE}$ as in Theorem~\ref{thm:utke}.
Then the scheme $\calC$ described above satisfies:
\begin{enumerate}
    \item \textbf{Undetectability.} $\calC$ is undetectable with respect to $\Pi$ (Definition~\ref{def:covert_undet}).
    \item \textbf{Key agreement.} With probability $1-\negl(\lambda)$, the derived keys satisfy $k_A=k_B$.
    \item \textbf{Key secrecy.} The common key is computationally indistinguishable from uniform given the public
    transcript $\tau$ (and $s$).
\end{enumerate}
\end{theorem}

\begin{proof}
Fix any PPT distinguisher $D$ in Definition~\ref{def:covert_undet}.
We prove undetectability by a hybrid argument over all embedding rounds.
Consider an eligible round $t$ in which a bit $b$ is embedded.
In the honest execution, the public message is $x\leftarrow x_t$.
In the covert execution, it is distributed as $x$ conditioned on the event $\Ext(s,x)=b$.
By Lemmas~\ref{lem:lhl} and \ref{lem:short-seed-extractor} (with $\ell=1$) and the assumption that $H_\infty(x_t) = \Omega(\lambda)$,
the joint distribution $(s,\Ext(s,x))$ is negligibly close to $(s,U_1)$ for $x\leftarrow x_t$.
Equivalently, for a uniformly random $s$, conditioning on $\Ext(s,x)=b$ changes the distribution of $x$ by at most
$\negl(\lambda)$ statistical distance.

Moreover, the embedded bits $b$ are themselves indistinguishable from uniform (since $(m_A,m_B)$ is a
pseudorandom-transcript KE transcript), so revealing which side of the conditioning we chose does not introduce
additional structure visible to the observer.
Thus, replacing the honest sampling in round $t$ by $\text{Embed}_t(s,b)$ changes the transcript distribution by
at most $\negl(\lambda)$.
Applying this replacement sequentially over the $2\ell=\Theta(\lambda)$ embedding rounds and taking a union bound
yields that the overall transcript of $\calC$ is computationally indistinguishable from an honest transcript of
$\Pi$, proving undetectability.

For key agreement and key secrecy, observe that by construction the receiver recovers the embedded bits in each
eligible round with probability $1$ (by recomputing $\Ext(s,\cdot)$ on the observed message).
Hence both parties reconstruct the exact PR-KE transcript $(m_A,m_B)$, except with negligible probability.
Correctness and key indistinguishability then follow directly from the guarantees of $\mathsf{PR\text{-}KE}$
(Theorem~\ref{thm:utke}), since the adversary's view is exactly the PR-KE transcript (embedded into an
indistinguishable carrier) together with the public randomness $s$.
\end{proof}

\subsection{The Bundle Sampler}
\label{sec:bundle}

In this section, we will describe our key tool in going beyond the high-entropy regime,  the {\em bundle sampler}. The method of Section~\ref{subsec:ke_ind_pub}, and indeed all prior works in the no-shared-key setting, relied on using extraction to embed a bit into the covertext distribution. If the covertext distribution has min-entropy $c$, this results in a statistical distance of at least $2^{-c}$ between the covertext and the stegotext distributions. Thus, they are all inherently stuck at the $\Omega(\lambda)$, or at the very least, $\omega(\log \lambda)$,  min-entropy barrier if they are to achieve cryptographic indistinguishability. The problem with this is two-fold: (1) an LLM conversation could have a large number of messages, and therefore significant entropy overall, but could consist of short constant-length messages in each round; and (2) viewed a different way, the length of the individual messages in the conversation upper-bound the level of achievable security, which is far from ideal.

Our bundle sampler is the first big step in overcoming this barrier. The key idea is to achieve {\em perfect indistinguishability} of the covertext and stegotext distributions, while pushing the (necessarily non-negligible) security loss {\em into a correctness error}, which is easier to deal with (and indeed, we will show how to do so). 

Unlike distribution-preserving samplers appearing in different settings (e.g., in watermarking auto-regressive samplers~\cite{christ2023undetectable,christ2024pseudorandom}), the key strength of our bundle sampler is that it {\em only requires sampling access} to the covertext distribution $\mathcal{D}$, as opposed to the exact probabilities of the next bit or token; and in particular, we can use it for the complete and arbitrary message distribution.

We will rely on the independent public randomness assumption for the next theorem.

\begin{figure}
\begin{framed}
\normalsize 
\begin{center}
    \textbf{Our Bundle Sampler $(\Embed^{\mathcal{D}},\Decode)$}
\end{center}

\textbf{Params:} The independent public parameters $\pp$ consists of a seed $s$ for a strong $(c,\epsilon)$-extractor $\Ext$, and a uniformly random bit $r$.

\medskip\noindent 
$\mathbf{\Embed^{\mathcal{D}}(\pp,b)}$ is initialized with an entropy parameter $c$ and a security parameter $\lambda$. Given the public parameters $\pp$, sampling access to a distribution $\mathcal{D}$, and a bit $b$, $\Embed$ does the following:
\begin{enumerate}
    \item Let $L := \Omega(\lambda / \epsilon^2)$, and let $\beta := b \oplus r$ denote the masked bit. 
    \item Sample a bundle of independent candidates $x_1,\ldots,x_L \leftarrow \mathcal{D}$.
    \item Partition indices by extractor label using seed $s$:
    \[
    I_0 := \{i\in[L] : \Ext(s,x_i)=0\},\qquad I_1 := \{i\in[L] : \Ext(s,x_i)=1\}.
    \]
    \item Let $m:=\min\{|I_0|,|I_1|\}$, fix subsets $I'_0\subseteq I_0$ and $I'_1\subseteq I_1$ of size~$m$
    (e.g., the first $m$ indices in each set), and let $I_\star:=[L]\setminus(I'_0\cup I'_1)$.
    \item Choose an index $J$ by doing the following:
    \begin{enumerate}
        \item With probability $|I_\star|/L$, sample $J$ uniformly from $I_\star$.
        \item Otherwise, sample $J$ uniformly from $I'_{\beta}$.
    \end{enumerate}
    Output $x_J$.
\end{enumerate}

\medskip\noindent
$\mathbf{\Decode(\pp,x)}$: Given the public parameters $\pp$, and a sample $x$, output $\Ext(s,x) \oplus r$.

\end{framed}
\caption{Our Bundle Sampler Construction: a pair of algorithms $(\Embed,\Decode)$ such that $\Embed$, given sample access to a distribution $\mathcal{D}$ and a bit $b\in\{0,1\}$, output $x$ such that (a) the marginal distribution of $x$ is precisely $\mathcal{D}$; and (b) $\Decode(x) = b$ with probability that grows with the min-entropy of $\mathcal{D}$.}
\label{fig:bundle}
\end{figure}

\begin{theorem}\label{thm:bundle}
  There is a pair of algorithms $(\Embed, \Decode)$ such that 
  for any distribution $\mathcal{D}$ with $H_{\infty}(\mathcal{D}) \geq c$,  the algorithms initialized with the entropy bound $c$,\footnote{We will use $c$ for the entropy bound, to signify that it is an absolute constant.} a $(c,\epsilon)$-strong seeded extractor $\Ext:\{0,1\}^d \times \{0,1\}^n$, and sampling access to $\mathcal{D}$, achieve the following: 
  \begin{enumerate}
      \item \textbf{Perfect Distribution Matching:} 
      For any public parameter $\pp$, when trying to embed a uniformly random bit $b$, the distribution of the output of $\Embed$ is {\em exactly} the same as $\mathcal{D}$. That is, for any $\pp$, for $b \sim \{0,1\}$,  
      $$ 
       \Embed^{\mathcal{D}}(\pp,b) \equiv
       \mathcal{D}~.
      $$ 
      We emphasize that perfect distribution matching holds for every choice of the public parameter $\pp$, and only depends on the uniform randomness of the transmitted bit $b$.

      \item \textbf{The BSC Property:} 
      There is a $2^d\cdot \mathsf{poly}(n, \lambda)$-time algorithm $\EstBSC$ that, given all the samples from $\mathcal{D}$ used by $\Embed$,  outputs a number $p \in [0,1]$ such that for every bit $b\in \{0,1\}$,
      $$ p = \Pr[\Decode(\pp,\Embed^{\mathcal{D}}(\pp,b))\neq {b}] $$
      where the probability is over $\pp$ and the internal coin-tosses of $\Embed$. 
      Furthermore, with probability $1-2^{-\Omega(\lambda)}$ over the samples from $\mathcal{D}$ used by $\Embed$, $p \leq 2\epsilon$.  

     In other words, the error channel is a binary symmetric channel (BSC) with efficiently computable parameter $p$. 
     
     \item \textbf{Correctness of Decoding:} In particular, as a corollary of property (2), for every $b\in \{0,1\}$,
      $$ \Pr[\Decode(\pp,\Embed^{\mathcal{D}}(\pp,b))=b] \geq 1- 2 {\epsilon} -2^{-\Omega(\lambda)}$$
      where $\epsilon$ is the error of the extractor $\Ext$ and $\lambda$ is a security parameter.

      \item \textbf{Noiseless Feedback:}
      The $\Embed$ algorithm knows with certainty whether the $\Decode$ algorithm will correctly recover the transmitted bit.
  \end{enumerate}
\end{theorem}

\begin{proof}
The bundle sampler $(\Embed,\Decode)$ is described in Figure~\ref{fig:bundle}. 

We first show the perfect distribution matching property. 
By construction, when~$b$ is a uniform bit,  the published index $J$ is uniform in $[L]$ \emph{regardless} of the partition
$(I'_0,I'_1,I_\star)$.
Since $x_1,\ldots,x_L$ are i.i.d.\ from $\mathcal{D}$ and $J$ is independent and uniform over $[L]$, the published message
$x_J$ is distributed \emph{exactly} as an honest sample from $\mathcal{D}$.

Now, we show the BSC property. Note that if we land in Step~6(b) of Figure~\ref{fig:bundle}, we have $$\Ext(s,x_J)=\beta = b\oplus r$$ by definition of $I'_{\beta}\subseteq I_{\beta}$, so $\Decode$ outputs $\beta \oplus r = b$ and succeeds.
Thus the only way decoding can fail is if we land up in Step~6(a), which happens with probability $|I_\star|/L$.  A corollary of this analysis is that the $\Embed$ algorithm knows with certainty whether the decoding algorithm will successfully reconstruct the transmitted bit.

Let $B = (x_1,\ldots,x_L)$ be the bundle drawn from $\mathcal{D}^L$ which, together with the extractor seed $s$, defines the sets $I_0 = I_0(s,B)$, $I_1 = I_1(s,B)$ and $I_{\star} = I_{\star}(s,B)$.
If $|I_0| = |I_1|$, we have $I_{\star} = \emptyset$, and the bit $\beta$ (and therefore $b$) is transmitted correctly.  If $|I_1| > |I_0|$, the bit $\beta = 1$ is always transmitted correctly; but $\beta = 0$ is flipped to a $1$ with probability $|I_{\star}|/L$.  Symmetrically, if $|I_0| > |I_1|$, the bit $\beta = 0$ is always transmitted correctly; but $\beta = 1$ incurs an error probability of $|I_{\star}|/L$. In other words, the error is one-sided. 
    
However, since the input bit $b$ is masked with the public random bit $r$ to obtain $\beta := b \oplus r$, the probability that $b=0$ is transmitted incorrectly and the probability that $b=1$ is transmitted incorrectly are the same, and they are both $$p = p(B) := \mathbb{E}_{s \sim U_d}\bigg[\frac{|I_{\star}(s,B)|}{2L} \bigg]~.$$ Here, the probability is over $s$ and $b$ in the public parameters as well as the internal coin-tosses that $\Embed$ uses to sample $J \in [L]$.

We now show an algorithm $\EstBSC$ that gets as input the bundle $B$, and computes $p$.   For each seed $s \in \{0,1\}^d$, $\EstBSC$ simply computes the quantity $|I_{\star}(s,B)|/2L$ and outputs the average of these values over all $s$. That is, it outputs 
$$ \frac{1}{2^d} \sum_{s\in \{0,1\}^d} \frac{|I_{\star}(s,B)|}{2L}$$
which, by definition, is exactly $p$. 

It remains to bound $p$. Define  
$$\gamma_s := \frac{1}{2} \cdot \big| \Pr_{x\sim\mathcal{D}}[\Ext(s,x)=1] - \Pr_{x\sim\mathcal{D}}[\Ext(s,x)=0]  \big|~.$$ 
Note that under the extractor guarantee with error $\varepsilon$, the bit $\Ext(s,x)$ for $x\sim \mathcal{D}$ and a uniform seed $s \sim U_d$
is $\varepsilon$-close to uniform. Moreover, the strong extractor guarantee tells us that 
$$ \mathbb{E}_{s\sim U_d, B \sim \mathcal{D}^L}\bigg[ \frac{|I_{\star}(s,B)|}{2L} \bigg] = \mathbb{E}_{s\sim U_d, B \sim \mathcal{D}^L}\bigg[ \frac{\big| |I_0(s,B)| - |I_1(s,B)| \big|}{2L} \bigg] =  \mathbb{E}_{s\sim U_d}[\gamma_s] \leq \epsilon~. $$
where the last equality is via Lemma~\ref{lem:avgbias}. 
A standard concentration bound
(e.g., a Hoeffding bound; see Lemma~\ref{lem:hoeffding}) implies that with probability at least $1-e^{-O(\epsilon^2 L)}$ over the choice of the bundle $B$, we have $$p = p(B) = \mathbb{E}_{s \sim U_d}\bigg[\frac{|I_{\star}(s,B)|}{2L}\bigg] \leq  2{\varepsilon}~.$$

Setting $L = \Omega(\lambda / \epsilon^2)$ gives $p \le 2{\epsilon}$ except with  probability $2^{-\Omega(\lambda)}$.

Finally, correctness of the protocol follows as a simple corollary of the BSC property. 
This finishes the proof.
\end{proof}

\subsubsection{Embedding Many Uniformly Random Bits into an LLM Conversation.}
\label{sec:embed_many}
We now informally describe how to use the bundling sampler to send a {\em sequence} of uniformly random bits $(b_1,\ldots,b_k)$, using an LLM conversation. This will shortly be made precise in the context of embedding the transcript of a pseudorandom-transcript key exchange protocol in Section~\ref{subsec:roadmap_short}. 
We will rely on the independent public randomness assumption and the recognizable entropy assumption. 

Recall that in an LLM conversation between a pair of agents $(A,B)$, in each odd round $t$ ($A$ speaks),
\[
({\prompt}_t,\stA^{(t)}) \leftarrow A\left(c_A,\stA^{(t-1)},\tau_{<t}\right),
\qquad
x^{(t)} \leftarrow \RModel_A(\prompt_t),
\]
and for each even round $t$ ($B$ speaks),
\[
({\prompt}_t,\stB^{(t)}) \leftarrow B\left(c_B,\stB^{(t-1)},\tau_{<t}\right),
\qquad
x^{(t)} \leftarrow \RModel_B(\prompt_t).
\]
Assume that Alice is the sender of the bits, and assume for simplicity that each next-message distribution $\RModel_A(\prompt_t)$ has min-entropy at least $c$, conditioned on the prior transcript $x^{(1)},\ldots,x^{(t-1)}$. (The recognizable min-entropy assumption for entropy level $c$ is in fact sufficient.) 
We also assume that the $\Embed$ algorithm can sample from the distribution $\RModel_A(\prompt_t)$.

We will use {\em independently random} extractor seeds $s_t$ and {\em independently random} mask bits $r_t$ for each $t\in [k]$.  For each odd round $t$, Alice computes $({\prompt}_t,\stA^{(t)})$ as above, and runs the $\Embed$ algorithm with entropy parameter $c$, extractor seed $s_t$, mask bit $r_t$, input bit $b_t$, and distribution $\mathcal{D} = \RModel_A(\prompt_t)$. The output $x^{(t)}$ of $\Embed$ is sent to Bob. Bob runs his end of the protocol exactly as he would in the honest protocol.

By Theorem~\ref{thm:bundle}, the resulting protocol transcript is distributed exactly like an honest conversation. Moreover, if we use optimal extractors (Theorem~\ref{lem:lhl}), we obtain the guarantee that each bit $b_t$ is independently transmitted through a BSC with error parameter $p = 2^{-O(c)}$. (We do not yet need to use the  $\EstBSC$ algorithm from Theorem~\ref{thm:bundle}, so we do not need to use the short-seed extractors of Theorem~\ref{lem:short-seed-extractor}).

\subsubsection{A Concrete Instantiation with $\Theta(\log \lambda)$ Entropy.}
\label{sec:logentropy}
As an immediate application of  the ideas developed in this section, we show a covert key exchange protocol under the existence of any pseudorandom-transcript two-message key exchange protocol (as in Theorem~\ref{thm:utke}); for simplicity we also require in this section that the protocol's transcript is {\em statistically uniform and not only computationally indistinguishable from uniform}. For example, this follows from the decisional Diffie-Hellman (DDH) assumption.
We assume that the LLM conversation has sufficiently many rounds in which the message distributions have min-entropy $\Theta(\log \lambda)$. We describe this protocol informally, as it will be superseded by the material in Section~\ref{subsec:roadmap_short}.

While this may sound like a small improvement over Section~\ref{subsec:ke_ind_pub}, it is in fact a large conceptual leap. Indeed, this will be the first time we are breaking out of the rejection sampling paradigm  used in prior works (including Section~\ref{subsec:ke_ind_pub}), which inherently requires min-entropy $\omega(\log \lambda)$. 
Undetectability in this scheme will be perfect. Technically, this protocol constitutes an immediate payoff of the ability of the bundling sampler to push a security loss into a correctness error. 

We first describe a scheme that succeeds with probability $1-1/\poly(\lambda)$:
Alice and Bob simply embed their PR-KE messages $m_A$ and $m_B$ using the protocol from Section~\ref{sec:embed_many}. 
Assuming that the honest LLM conversation $\Pi$ has entropy guarantee $\alpha\log \lambda$ for a sufficiently large constant $\alpha$, we know by a union bound that \emph{all} the bits of these messages are transmitted correctly with probability at least $1-1/\poly(\lambda)$. 
In this event, Alice and Bob exactly run the PR-KE protocol and hence agree on a shared secret key.  
The message distribution is always identical to that of the honest protocol $\Pi$ due to the use of the bundle sampler.
Moreover, \emph{regardless of whether the key exchange succeeded} the message decoded by Alice is uniform -- as even if Bob's intended message was noised, it was noised via a BSC -- and hence the key Alice computes is indistinguishable from uniform given the random transcript. In particular, the secrecy of Alice's key conditioned on the public transcript holds even in the failure event. (This last property may not hold for {\em pseudorandom}-transcript KE protocols in general, and hence our requirement of {\em uniform} transcripts.)

\paragraph{Amplifying Success Probability.}

If we wish to drive the failure probability down to $\negl(\lambda)$, a natural attempt is to have the parties repeat the entire procedure on fresh blocks of eligible rounds, with fresh public parameters for $\Embed$.  One would hope that since each repetition fails with probability at most $1/\poly(\lambda)$ and repetitions are independent conditioned on the honest transcript distribution, performing $\omega(1)$ repetitions makes the overall failure probability negligible. 

What is not yet clear, however, is how Alice and Bob covertly decide which of the many keys to output. To facilitate this, we use the CGZ watermarking scheme~\cite{christ2023undetectable} together with pairwise independent hash functions. In more detail, the CGZ watermarking scheme (see Section~\ref{sec:CGZ}) gives us a way to embed a watermark into a covertext distribution using a key $k$ such that (a) a detection algorithm with the same key $k$ succeeds in detecting the watermark; and (b)  if $k$ is random to the distinguisher, the watermark is computationally undetectable. 

With this, a natural idea is for Alice to send a watermark using her key $k_A$. Bob attempts to detect the watermark using $k_B$. If $k_B = k_A$, Bob detects this watermark and knows the key exchange succeeded; he can relay this information back by embedding a watermark with the same key in his following message.
Since $k_A$ is indistinguishable from uniform given the public transcript regardless of whether the key exchange succeeded or not, the watermark is undetectable to an external distinguisher. 

The above is still incomplete due to the following subtlety: CGZ does not tell us anything about detection with a key $k_B \neq k_A$ that might be correlated or similar to~$k_B$; in particular, detection may succeed even if $k_B \neq k_A$ as from Bob's point of view Alice's key is not guaranteed to appear uniform even if the key exchanged failed. 
Fortunately, this is easy to fix: instead of running the watermarking procedure with $k_A$ itself, Alice uses $h(k_A)$ where $h$ is a pairwise independent hash function that is drawn using the public parameters. Bob attempts to detect using $h(k_B)$. 
Now, if $k_A \neq k_B$, $h(k_A)$ and $h(k_B)$ are uniform and independent, and thus the watermarking detection will fail, as desired.  

To conclude, Alice and Bob may \emph{verify} if the key-exchange succeeded while guaranteeing undetectability in either case. Hence amplification by independent repetition becomes straightforward.

\subsection{Handling Constant Entropy Messages}\label{subsec:roadmap_short}

The protocol of subsection~\ref{subsec:ke_ind_pub} was analyzed under the simplifying \emph{high-entropy} assumption: we assumed that
the honest conversation~$\Pi$ contains sufficiently many recognizable eligible rounds, namely ones in which the honest
next-message distribution $x_t$ has min-entropy at least~$\Omega(\lambda)$ (equivalently, that message length is at
least~$\Omega(\lambda)$ under our entropy vs. length assumptions; see~\ref{sec:empiricalheur}).
Section~\ref{sec:logentropy} further relaxed this to $\Theta(\log \lambda)$ min-entropy. 
While these may be reasonable in practice, it is theoretically unsatisfying: $\Pi$ is given and out of our
control, so the entropy (and length) of its individual messages need not scale with the security parameter.
In particular, it is undesirable for the security parameter to be limited by the length of the honest messages rather than allowing it to be set arbitrarily large.

In this section and the remainder of the paper, we remove this assumption and show a covert key exchange protocol that continues to work even
when eligible messages have only \emph{constant} min-entropy, say $H_\infty(x_t)\ge c$ for some fixed constant $c$ independent of~$\lambda$.
In the high-entropy regime $|I_\star|/L$ was polynomially small in~$\lambda$ so we could assume, using a union bound, that no errors would affect the protocol; when $H_\infty(X_t)\ge c$ the error rate becomes a constant and cannot be ignored.
The price is that the construction can no longer use a PR-KE protocol (such as the Diffie-Hellman protocol used in Section~\ref{sec:logentropy}): instead, it requires a
key-exchange protocol that is simultaneously (i) pseudorandom-transcript and (ii) resilient to \emph{noise} in the transcript. We call this a pseudorandom noise-resilient key exchange (PNR-KE). 

Our notion of pseudorandom, noise-resilient key exchange (PNR-KE) is a strengthening of pseudorandom-transcript key exchange (PR-KE, see Theorem~\ref{thm:utke}) which requires that the key exchange protocol succeeds even when all the protocol messages are transmitted through a binary symmetric channel $\mathsf{BSC}_{p}$. 
See Section~\ref{sec:pnrke-def} for a formal definition.

\paragraph{Reduction to pseudorandom noise-resilient KE.}
We show how to ``compile'' an honest LLM conversation $\Pi=(A,B)$ and a pseudorandom noise-resilient key exchange protocol $\Gamma = (A_\Gamma, B_\Gamma)$ into a covert key exchange protocol $\widehat{\Pi} = (\widehat{A},\widehat{B})$. 

\begin{theorem}
\label{thm:dep-rand-ke}
Fix an honest protocol $\Pi = (A,B)$
for which there exists a 
large enough absolute constant $c>0$ such that sufficiently many rounds have message distributions with min-entropy at least 
$c$. 
Let $\mathsf{PNR\text{-}KE}$ be a pseudorandom noise-resilient key exchange protocol as in Definition~\ref{def:pnrke} resilient to a noise rate of $2^{-\Omega(c)}$, and
let $\Ext$ be a $(c,\varepsilon)$-strong extractor with $d=O(\log\lambda)$
as in Lemma~\ref{lem:short-seed-extractor}.
Then there is a covert key exchange protocol $\calC$ that satisfies:
\begin{enumerate}
    \item \textbf{Undetectability.} $\calC$ is undetectable with respect to $\Pi$ (Definition~\ref{def:covert_undet}).
    \item \textbf{Key agreement.} With probability $1-\negl(\lambda)$, Alice and Bob output the same key
    $k_{\mathsf{sess}}\in\{0,1\}^\lambda$.
    \item \textbf{Key secrecy.} Conditioned on success, $k_{\mathsf{sess}}$ is computationally indistinguishable
    from uniform given the public transcript $\tau$ and the public parameters $\pp$.
\end{enumerate}
\end{theorem}

\begin{proof}
We describe the covert key exchange scheme $\calC=(\widehat{A},\widehat{B})$. The setup algorithm samples the public parameters $\pp_{\mathsf{ke}}$ needed for the underlying PNR-KE protocol and a sequence of i.i.d. public parameters $\pp_{\mathsf{Embed},i}$ for the $\Embed$ algorithm from Theorem~\ref{thm:bundle}.

\medskip
\noindent\textbf{Message generation.}
We define $\widehat{A}$ and $\widehat{B}$ as follows.
Each maintains the honest state prescribed by $\Pi$ (so that it can compute the honest prompt and hence the
honest next-message distribution $x_t$ in each round), together with an additional local state for the covert
key exchange.

For each round of the PNR-KE protocol where Alice sends a message to Bob, $\widehat{A}$ computes the Alice message $\mu_A$ and embeds the bits of $\mu_A$ rounds. Symmetrically, for each round of the PNR-KE protocol where Bob sends a message $m_B$, $\widehat{B}$ embeds the bits of $\mu_B$
into the first $|\mu_B|$ eligible Bob rounds. In particular, let $m_A$ (resp. $m_B$) denote the concatenation of all $\mu_A$ (resp. $\mu_B$) so far.

Concretely, once Alice has computed the $j$-th bit of $m_A$ (that is, she has received all the necessary prior messages from Bob), Alice finds the next eligible round $t^A$, computes the honest prompt $\prompt_{t^A}$ as
prescribed by $\Pi$ (thus defining the probability distribution $x_{t^A}$) and outputs
\[
x^{(t^A)} \leftarrow \Embed_{t^A}(\pp_{\mathsf{Embed},j},\,m_A[j]).
\]
Note that the public parameters are never reused. Symmetrically, once Bob has computed the $j$-th bit of $m_B$, he finds an eligible  round $t^B$, and analogously outputs
\[
x^{(t^B)} \leftarrow \Embed_{t^B}(\pp_{\mathsf{Embed},j},\,m_B[j]).
\]
In all other rounds, the speaker samples
honestly: $x^{(t)}\leftarrow x_t$.

By the recognizable entropy assumption, the receiving party knows which messages were used to embed a bit. They run the $\Decode$ algorithm to recover this bit (with some error probability as in Theorem~\ref{thm:bundle}). 

\smallskip
After the transcript contains all the  embedded bits of the PNR-KE protocol, Alice uses the transcript and her private PNR-KE randomness to derive a key $k_A$ and outputs it as the session key. Bob symmetrically uses his transcript and his private PNR-KE randomness to derive a key $k_B$ and outputs it as the session key.

\medskip\noindent
We prove the three properties  in turn.

\smallskip
\noindent\textbf{Undetectability.}
Consider the transcript distribution induced by $\calC$ and compare it to an honest execution of $\Pi$.
In every round in which $\calC$ samples honestly, the per-round distribution is identical.
In an eligible round $t$, $\calC$ runs the $\Embed$ algorithm with a fresh public parameter $\pp_t$ and a bit of the underlying PNR-KE protocol.  If the bit were truly random, the perfect distribution matching guarantee of the $\Embed$ algorithm guarantees that the resulting distribution is identical to that of the honest protocol $\Pi$. By a simple hybrid argument, it follows that since the PNR-KE transcript is pseudorandom, the distribution generated by $\mathcal{C}$ is indistinguishable from that of the honest protocol $\Pi$. 

\medskip
\noindent\textbf{Key agreement.}
In each eligible round used to embed a bit, the embedding primitive outputs a message whose decoding 
matches the intended bit except with probability $2^{-O(c)}$, by Theorem~\ref{thm:bundle}. Furthermore, these probabilities are independent for each bit. In other words, the embedding process exactly simulates the noisy channel in the definition of the PNR-KE protocol. The correctness property of the PNR-KE protocol implies that key agreement is achieved with probability $1-\negl(\lambda)$. 

\medskip
\noindent\textbf{Key secrecy.}
Conditioned on the success event above, the derived key $k_{\mathsf{sess}}$ is exactly the output of the PNR-KE
protocol. By undetectability, the public transcript of $\mathcal{C}$ reveals no additional information beyond what is already
computationally captured by the (uniform-looking) PNR-KE transcript embedded in it.
Since PNR-KE is secure, $k_{\mathsf{sess}}$ is computationally indistinguishable from uniform
given the transcript.
\end{proof}

\begin{remark}[PNR-KE with Feedback is Sufficient]
\label{rem:weaker-than-full-pnrke}
For the reduction above, full PNR-KE is stronger than necessary.
It already suffices to assume a feedback version in which each transmitted bit may be flipped with an
arbitrary {\em sender-known} probability $p_t\le p$, and after transmission the sender learns whether the flip
occurred.  This is exactly the noise structure provided by Theorem~\ref{thm:dep-rand-ke}.
\end{remark}

\subsubsection{Our Main Theorems: Instantiations with PNR-KE protocols}

The rest of the paper is dedicated to the study and instantiation of pseudorandom noise-resilient key-exchange (PNR-KE) protocols required to instantiate Theorem~\ref{thm:dep-rand-ke}.
In Section~\ref{sec:prnr-ke}, we formally define the notion of PNR-KE and construct a protocol whose hardness is based on the hardness of the learning sparse parities with noise (LSPN) problem. 
This completes our first concrete instantiation of all the assumptions in Theorem~\ref{thm:dep-rand-ke}.
\begin{theorem}
    Assuming the LSPN conjecture with parameters $n,m=\Theta(n),k=\Theta(\log n),\eta = \Theta(1)$ (see Section~\ref{sec:lspn}), there exists a covert key exchange protocol with quasi-polynomial undetectability, assuming only constant entropy per eligible message.
\end{theorem}

\noindent
This is accomplished by using the protocol from Theorem~\ref{thm:dep-rand-ke} together with the PNR-KE protocol of Section~\ref{sec:simpleprot}, Theorem~\ref{thm:simpleprot}.

The disadvantage of the above is that the best known attacks on LSPN with our parameters run in quasi-polynomial time, and thus it gives only (assumed) quasi-polynomial hardness.
Then, in Section~\ref{sec:feedback_pnr-ke}, we make use of the noise structure properties of $\Embed$ (see Theorem~\ref{thm:bundle}) and design PNR-KE protocols assuming noiseless feedback and a  $\BSC(p_t)$ noise distribution with known parameters $p_t$ (that could be different per round).
Under these assumptions on the noise, we compile \emph{any} PR-KE protocol into a PNR-KE protocol. 

In Section~\ref{subsec:optimal_majority_signaling}, we give a simpler such compiler that transforms any PR-KE protocol of length~$T$ into a (feedback-)PNR-KE protocol of length~$\poly(T)$ whose error probability is~$1/\poly(\lambda)$.
If the protocol admits covert verification of the success of key exchange (for example, {\em uniform}-transcript key exchange protocols have this property), the observation of Subsection~\ref{sec:logentropy} shows us how to obtain negligible error probability. In particular, since the Diffie-Hellman protocol is a uniform-transcript key exchange protocol, we have the following theorem:

\begin{theorem}
    Under the DDH assumption, there exists a covert key exchange assuming only constant-entropy per eligible message.
\end{theorem}

\begin{proof}
First, compile the Diffie-Hellman key exchange protocol into a PNR-KE protocol with feedback, using the majority compiler from Section~\ref{subsec:optimal_majority_signaling}. This gives us a protocol with correctness error $1/\poly(\lambda)$. Next, use the amplification procedure in Section~\ref{sec:logentropy} to reduce the error probability to $\negl(\lambda)$. Now, one can use this together with the covert key exchange protocol from Theorem~\ref{thm:dep-rand-ke} to finish the construction.
\end{proof}

Finally, in Section~\ref{subsec:optimal_signaling}, we show our strongest result which transforms \textbf{\em any PR-KE protocol} of length~$T$  into a (feedback-)PNR-KE protocol of length~$\poly(T)$ whose error probability is~$e^{-\Omega(\lambda)}$. Together with Theorem~\ref{thm:dep-rand-ke}, this concludes the proof of our main application.

\begin{theorem}
\label{thm:main}
Fix an honest protocol $\Pi = (A,B)$
for which there exists a 
large enough absolute constant $c>0$ such that sufficiently many rounds have message distributions with min-entropy at least 
$c$. 
Let $\mathsf{PR\text{-}KE}$ be a pseudorandom key exchange protocol, and
let $\Ext$ be a $(c,\varepsilon)$-strong extractor with $d=O(\log\lambda)$
as in Lemma~\ref{lem:short-seed-extractor}.
Then there is a covert key exchange protocol $\calC$ that satisfies:
\begin{enumerate}
    \item \textbf{Undetectability.} $\calC$ is undetectable with respect to $\Pi$ (Definition~\ref{def:covert_undet}).
    \item \textbf{Key agreement.} With probability $1-2^{-\Omega(\lambda)}$, Alice and Bob output the same key
    $k_{\mathsf{sess}}\in\{0,1\}^\lambda$.
    \item \textbf{Key secrecy.} Conditioned on success, $k_{\mathsf{sess}}$ is computationally indistinguishable
    from uniform given the public transcript $\tau$ and the public parameters $\pp$.
\end{enumerate}
\end{theorem}

\input{key-exchange-protocol.tex}

\section{Impossibility Results on PNR-KE Protocols}
\label{sec:prnr-ke-LB}

This section explains why the pseudorandom noise-resilient key-exchange primitive introduced in Section~\ref{sec:prnr-ke} is a nontrivial object.
We show that several more na\"{i}ve or more restrictive variants are in fact impossible, or admit strong attacks.
Taken together, these results help clarify both the necessity of the choices in our constructions and the limitations of what one can hope to achieve in related settings.

We begin with a short review of the Fourier-analytic facts used in the proofs.
We then rule out a public analog of pseudorandom codes, showing that without a shared secret such an approach cannot yield nontrivial resilience to constant noise.

Finally, we show that any \emph{non-interactive} PNR-KE protocol that results in even slightly (i.e. inverse-polynomially) correlated bits is vulnerable to a quasipolynomial-time attack, and that if such a protocol achieves key agreement with large (i.e. constant) probability, then it always admits a polynomial-time attack.
These results justify the interactive correlation-amplification step at the end of our protocol from Section~\ref{sec:prnr-ke}.
They also help explain why a quasipolynomial-time attack is possible there, since each individual iteration of that protocol is by itself a non-interactive step producing a nontrivial key correlation.

\subsection{Boolean Analysis Preliminaries}
We briefly record the Fourier-analytic facts used in this section (see~\cite{o2021analysis} for a detailed introduction).
Throughout this section we identify $\{0,1\}$ with $\{\pm1\}$ via $b\mapsto (-1)^b$ when convenient.
For $x\in\{0,1\}^m$ and $S\subseteq[m]$, let
\[
\chi_{S(x)} := (-1)^{\sum_{i\in S} x_i}.
\]
For a function $h:\{0,1\}^m\to\mathbb{R}$, its Fourier expansion is
\[
h(x)=\sum_{S\subseteq[m]} \widehat{h}(S)\chi_{S(x)},
\]
and Parseval's identity says
\[
\sum_{S\subseteq[m]} \widehat{h}(S)^2 = \mathbb{E}[h(U)^2],
\]
where $U\sim U_m$ is uniform on $\{0,1\}^m$.
More generally, for $h:\{0,1\}^{m_1}\times\{0,1\}^{m_2}\to\mathbb{R}$ we write
\[
h(x,y)=\sum_{S\subseteq[m_1],\,T\subseteq[m_2]} \widehat{h}(S,T)\chi_{S(x)}\chi_{T(y)}.
\]

For $p\in[0,1/2]$, let $\mathsf{BSC}_p(x)$ denote the distribution obtained from $x$ by flipping each bit independently with probability $p$, and write
\[
\rho := 1-2p.
\]
If $\widetilde{X}\sim \mathsf{BSC}_p(X)$ where $X$ is uniform on $\{0,1\}^m$, then
\[
\mathbb{E}[\chi_S(X)\chi_T(\widetilde{X})]
=
\begin{cases}
\rho^{|S|} & \text{if } S=T,\\
0 & \text{otherwise}.
\end{cases}
\]
We will use this identity repeatedly.

For a function $h:\{0,1\}^m\to\mathbb{R}$ and a parameter $\rho\in[-1,1]$, the \emph{$\rho$-noise stability} of $h$ is
\[
\Stab_\rho(h)
:=
\mathbb{E}\big[h(X)\,h(Y)\big],
\]
where $X\sim U_m$ and $Y\sim N_\rho(X)$ is a $\rho$-correlated copy of $X$ (equivalently, $Y$ is obtained from $X$ by applying $\mathsf{BSC}_p$ independently to each coordinate, where $\rho=1-2p$).
If
\[
h(x)=\sum_{S\subseteq[m]}\widehat{h}(S)\chi_{S(x)},
\]
then
\[
\Stab_\rho(h)=\sum_{S\subseteq[m]}\rho^{|S|}\widehat{h}(S)^2.
\]

\subsection{Public Pseudorandom Codes Do Not Exist}
\label{subsec:public-prc-impossible}

The combination of pseudorandomness and resilience to errors also appeared recently in the definition of \emph{pseudorandom codes} (PRCs)~\cite{christ2024pseudorandom, alrabiah2025ideal}; these are error-correcting codes for which a uniformly random codeword appears pseudorandom to a polynomial-time observer.
Crucially, a PRC is defined with a secret key shared between the encoder and the decoder -- and the codewords are pseudorandom only to a distinguisher that does not have this key.  

The purpose of a key exchange protocol is precisely to \emph{establish such a shared key}.

Nonetheless, one may ask if a public variant of a PRC could exist. This is a PRC in which the encoder and the decoder are public; they perhaps depend on public randomness, but not on any shared secret.
Defining pseudorandomness for such public PRCs requires care: if the encoded message can be recognized as meaningful, then the public decoder itself suffices to distinguish between a codeword and the encoding of a meaningful message. 
Thus, we would relax the requirement to the encoding of a \emph{random message} being pseudorandom rather than the encoding of \emph{any} message being pseudorandom. 
This relaxed definition would have sufficed for the construction of a PNR-KE protocol, as we could use it to encode the messages of any PR-KE protocol -- which are indeed random-looking. We call this notion a {\em public PRC.}\footnote{
Not to be confused with the notion of {\em publicly detectable} pseudorandom codes~\cite{DBLP:journals/cic/FairozeGJMMW24,lin2026unforgeablewatermarkslanguagemodels}, where a pair of different (but related) keys are used for generation and detection of codewords; a notion reminiscent of digital signature schemes.
}

We rule out the existence of a non-trivial public PRC by showing that even for~$1$-bit messages, the probability of decoding error must be~$\Omega(p)$ regardless of the codeword size, which is as bad as simply sending the message bit without any encoding.
Suppose an encoder uses only public randomness (equivalently, a public seed that is available to both
the decoder and the distinguisher) and, on input bit $i\in\{0,1\}$, outputs a
string in $\{0,1\}^{n}$.  Absorbing the public seed into the transmitted
string, such a scheme induces two distributions $D_{0,n},D_{1,n}$ over
$\{0,1\}^n$, one for each possible message bit, together with a polynomial-time
decoder $f_n:\{0,1\}^{n}\to\{0,1\}$.  
The observer who does not know the hidden bit sees the average distribution
\[
A_n := \frac12 D_{0,n} + \frac12 D_{1,n}.
\]
The theorem below shows that if $A_n$ is computationally indistinguishable from
uniform, then no public decoder can recover the hidden bit with non-trivial error after a constant
amount of binary symmetric noise.

\begin{theorem}
\label{thm:public-prc-impossible}
Fix a constant $p\in(0,1/2)$.
For each $n$, let $D_{0,n},D_{1,n}$ be distributions on $\{0,1\}^{n}$,
let $f_n:\{0,1\}^{n}\to\{0,1\}$ be a polynomial-time computable decoder, and let $\nu_p=\Ber(p)^n$.
Assume
\begin{enumerate}
    \item \emph{(Pseudorandomness)} The distribution 
    \[
    A_n := \frac12 D_{0,n} + \frac12 D_{1,n}
    \]
    is computationally indistinguishable from the uniform distribution $U_{n}$  on $\{0,1\}^{n}$; and
    \item \emph{(Decodability)} For each $i\in\{0,1\}$,
    \[
    \Pr_{x\sim D_{i,n},\,e\sim \nu_p}[f_n(x\oplus e)=i]\ge 1-\varepsilon_n.
    \]
\end{enumerate}
Then
\[
\varepsilon_n \ge \frac{p}{4} - o_n(1).
\]
In particular, $\varepsilon_n=\Omega(p)$, and hence $\varepsilon_n\not=o(1)$.
\end{theorem}

\begin{proof}
Fix $n$, and abbreviate
\[
A:=A_n,\qquad f:=f_n,\qquad \varepsilon:=\varepsilon_n.
\]
If $\varepsilon\ge 1/2$ there is nothing to prove, so we may assume
$\varepsilon<1/2$.
The proof isolates two efficiently testable consequences of the decodability
assumption.

\paragraph{Test 1: Agreement on two independent noisy views.}
Let $I\sim\Ber(1/2)$ be a random bit, then sample $X\sim D_{I,n}$, and independently sample
$E_1,E_2\sim \nu_p$.  
By assumption, for each $j\in\{1,2\}$, 
\[
\Pr[f(X\oplus E_j)=I] \ge 1-\varepsilon.
\]
Thus by a union bound,
\begin{equation}
\Pr[f(X\oplus E_1)\neq f(X\oplus E_2)]
\le
2\varepsilon.
\label{eq:test1-real}
\end{equation}
This quantity can be efficiently estimated from samples from $A$ (each with ``fresh'' samples of~$I$ and~$X$) and evaluations of~$f$.

\paragraph{Test 2: Small bias on one noisy view.}
Let
\[
M_{i}:=D_{i}\oplus\nu_p,
\qquad
M:=\frac12 M_{0}+\frac12 M_{1}=A\oplus\nu_p.
\]
Define the $\{\pm1\}$-valued decoder
\[
g(x):=
\begin{cases}
+1,& f(x)=1,\\
-1,& f(x)=0.
\end{cases}
\]
Also set
\[
\alpha_{0}:=\Pr_{Y\sim M_{0}}[f(Y)=1],
\qquad
\alpha_{1}:=\Pr_{Y\sim M_{1}}[f(Y)=0].
\]
By assumption, $\alpha_{0},\alpha_{1}\le \varepsilon$.  Moreover,
\[
\mathbb{E}_{x\sim M_{0}}[g(x)]=-1+2\alpha_{0},
\qquad
\mathbb{E}_{x\sim M_{1}}[g(x)]=1-2\alpha_{1},
\]
and therefore
\[
\mathbb{E}_{x\sim M}[g(x)]
=
\frac12(-1+2\alpha_{0})+\frac12(1-2\alpha_{1})
=
\alpha_{0}-\alpha_{1}.
\]
Hence
\begin{equation}
\bigl|\mathbb{E}_{x\sim M}[g(x)]\bigr|\le \varepsilon.
\label{eq:test2-real}
\end{equation}
This quantity is also efficiently estimable from samples from $A$ and evaluations of~$f$.

\medskip
We now consider the same two tests on samples from the uniform distribution~$U=U_n$ rather than from $A$. Closure of computational indistinguishability under efficient
randomized transformations implies that the tests should return answers that are negligibly-far if run on~$U$ instead.
On the other hand, we will show that these two conditions cannot hold at the same time for \emph{any} Boolean function~$f$ when tested over the uniform distribution~$U$. This is because no near-balanced function is more stable to noise than a dictator function (whose stability is only~$p$). We include a proof for completeness, but this fact is well known (see e.g.,~\cite{o2021analysis}).

Test~1 implies that
\begin{equation}
\Pr[f(U\oplus E_1)\neq f(U\oplus E_2)]
\le
2\varepsilon+o(1).
\label{eq:test1-unif}
\end{equation}

Defining
\[
\beta:=\mathbb{E}_{x\sim U}[g(x)],
\]
Test~2 yields
\begin{equation}
|\beta|\le \varepsilon+o(1).
\label{eq:test2-unif}
\end{equation}

At this point we compare the two transferred tests with the behavior of an
arbitrary Boolean function under true uniform input.  Let
\[
q:=2p(1-p),
\qquad
Z\sim \Ber(q)^{n}.
\]
Since $E_1\oplus E_2\sim \Ber(q)^{n}$ and the first coordinate is uniform,
the pair $(U\oplus E_1,U\oplus E_2)$ has exactly the same distribution as
$(U,U\oplus Z)$.  Hence \eqref{eq:test1-unif} may be rewritten as
\begin{equation}
\Pr[f(U)\neq f(U\oplus Z)]
\le
2\varepsilon+o(1).
\label{eq:disagreement-upper}
\end{equation}

On the other hand, for any $\{\pm1\}$-valued function $g$ and any $q\in[0,1/2]$,
\[
\Pr[g(U)\neq g(U\oplus Z)]
=
\frac{1-\Stab_{1-2q}(g)}{2}.
\]
by the definition of noise stability.

Applying this to and using the Fourier expansion of $g$, we get
\[
\Stab_{1-2q}(g)
=
\sum_{S\subseteq[n]}(1-2q)^{|S|}\widehat{g}(S)^2
\le
\beta^2+(1-2q)(1-\beta^2)
= 1 - 2q(1-\beta^2)
,
\]
since $\widehat{g}(\varnothing)=\beta$ and
$\sum_S \widehat{g}(S)^2=1$.  Therefore
\begin{equation}
\Pr[f(U)\neq f(U\oplus Z)]
\ge
q(1-\beta^2).
\label{eq:disagreement-lower}
\end{equation}

Combining \eqref{eq:disagreement-upper}, \eqref{eq:disagreement-lower}, and
\eqref{eq:test2-unif}, we obtain
\[
2\varepsilon+o(1)
\ge
q\bigl(1-(\varepsilon+o(1))^2\bigr).
\]
Since~$\varepsilon \leq \frac12$ this implies~$\varepsilon \geq \frac{3}{8}q-o(1)\geq\frac38p-o(1)$.

In particular, if~$\varepsilon <\frac{p}{4}$ then either the tests that are defined on~$A$ and are required for decodability fail, or they provide an efficient distinguisher between~$A$ and~$U$ as they must fail for any boolean function when tested over~$U$.
\end{proof}

The theorem rules out public pseudorandom codes as a mechanism for
reliably signaling a uniform bit through a constant-noise binary symmetric
channel.  The obstruction is inherent to the public setting: the same public
information that enables decoding also enables the distinguishing test above.
This is precisely why the argument does not apply to secret-key constructions,
where decoding may rely on auxiliary information unavailable to the distinguisher.

\subsection{A Quasi-polynomial Attack on any Non-Interactive PNR-KE Protocol}
\label{subsec:qpoly-attack}

We now show that any non-interactive PNR-KE protocol, that is, Alice sends one message and Bob sends one message non-adaptively, and these suffice for the parties to obtain a correlated bit, must have a quasipolynomial attack. 

While many classical key exchange protocols (such as Diffie-Hellman) are indeed non-interactive, our construction in Section~\ref{sec:prnr-ke} does require interaction for the final error-amplification of the different iterations. 
Nonetheless, each particular iteration is by itself a non-interactive PNR-KE protocol with correlation (agreement probability)~$\delta\geq 1/\poly(\lambda)$ which suffices for our lower bound to hold.
Thus, the presented lower bound holds not only for non-interactive protocols, but also for any protocol from which we can extract a non-interactive protocol that results merely in slightly correlated bits.
In fact, the attack we present uses~$O(\lambda^{O(\log 1/\delta)})$ samples and runs in proportional time -- thus, a non-interactive protocol that would have given a large (i.e. constant) correlation by itself, without repetitions, would actually always have a polynomial-time attack.

We state the theorem for a general correctness bias $\delta$.
Definition~\ref{def:pnrke} corresponds to the strong regime $\delta=1-\negl(\lambda)$.
For this subsection, we encode key bits in $\{\pm1\}$ rather than in $\{0,1\}$.

\begin{theorem}\label{thm:qpoly-attack}
Let $\Pi$ be a non-interactive PNR-KE protocol, except that we only assume the weaker correctness guarantee
\[
\Pr[K_A=K_B]\ge \frac12+\frac{\delta}{2}
\]
for some $\delta=\delta(\lambda)\in(0,1]$.
Assume that $p\in (0,\frac12)$ is a fixed constant and $\delta\ge 1/\mathsf{poly}(\lambda)$.
Then there is a distinguisher that, given
\[
N = O\!\left(\lambda^{O(\log \frac{1}{\delta})}\right)
\]
independent samples of $(m_A,m_B,K_A)$, breaks either the protocol's key indistinguishability or pseudorandom transcript property with constant advantage.
In other words, as long as~$\delta \geq 1/\poly(\lambda)$, the attack runs in quasipolynomial time and uses quasipolynomially many samples.
\end{theorem}

We begin by briefly sketching the proof idea.
First, we observe that in any key-exchange protocol, the public transcript alone suffices to inefficiently reconstruct the shared key: consider the messages~$(X,Y)$ each side sends (that is, before any channel noise).
We can ``recover'' Alice's private randomness by sampling it from the distribution of~$R_A$ conditioned on Alice's sent message being~$X$; Then, we can run the honest protocol from Alice's point-of-view given~$R_A$ and Bob's message~$Y$.
Thus, there is an (inefficient) function~$F(X,Y)$ that estimates Alice's derived key from the public messages that both sides sent.
Symmetrically, there is also such a function~$G(X,Y)$ that estimates Bob's derived key.
Due to the definition of~$F$, it correlates with the published key~$K_A$, and due to the definition of a successful key-exchange,~$F$ and~$G$ are also correlated.
Then, we make our core observation: due to the noise-resilience of the protocol,~$F$ must be stable to noise in the~$Y$-coordinate and~$G$ must be stable to noise in the~$X$-coordinate.
Due to that noise stability, we can show that Fourier coefficients of~$F$ corresponding to sets with a ``large''~$Y$-coordinate contribute very little to its output and can be ignored. Similarly,~Fourier coefficients of~$G$ corresponding to sets with a ``large''~$X$-coordinate can also be ignored.
This results in~$K_A$ being correlated with a function of~$(X,Y)$ that depends \emph{only on small Fourier coefficients}. 
Then, we use the standard fact that all of these small Fourier coefficients can be estimated with a relatively small number of samples (and proportional time).

\begin{proof}
At first, assume moreover that the public transcript is \emph{exactly} uniform: for some message lengths $\ell_A,\ell_B=\Theta(\lambda)$,
\[
(m_A,m_B)\equiv (U_{\ell_A},U_{\ell_B}),
\]
where the two uniform strings are independent.
Let $R_A,R_B$ denote the private randomness of Alice and Bob, and denote the messages each party sends in the protocol as
\[
X := m_A(R_A)\in\{0,1\}^{\ell_A},
\qquad
Y := m_B(R_B)\in\{0,1\}^{\ell_B}.
\]
Let $Z_A\sim \mathsf{Ber}(p)^{\ell_A}$ and $Z_B\sim \mathsf{Ber}(p)^{\ell_B}$ denote the independent channel noises on Alice's and Bob's messages, respectively. Thus the parties derive their respective keys by some fixed functions
\[
K_A = A(R_A,Y\oplus Z_B),
\qquad
K_B = B(R_B,X\oplus Z_A).
\]

We will analyze the dependence of Alice's published key on the public transcript, rather than on her private randomness.
To this end, define the \emph{effective response functions} which give the best estimate of each key while ``forgetting'' the respective party's private randomness:
\[
F_0(x,y) := \mathbb{E}\!\left[K_A \mid X=x,\; Y\oplus Z_B = y\right],
\qquad
G_0(x,y) := \mathbb{E}\!\left[K_B \mid Y=y,\; X\oplus Z_A = x\right].
\]
Equivalently,
\[
F_0(x,y) = \mathbb{E}\!\left[A(R_A,y)\mid X=x\right],
\qquad
G_0(x,y) = \mathbb{E}\!\left[B(R_B,x)\mid Y=y\right].
\]
These are fixed functions from $\{0,1\}^{\ell_A}\times \{0,1\}^{\ell_B}$ to $[-1,1]$.

Averaging over the channel noise gives the conditional biases of the actual keys:
\[
F(x,y)
:=
\mathbb{E}[K_A\mid X=x,Y=y]
=
\mathbb{E}_{Z_B}\!\left[F_0(x,y\oplus Z_B)\right],
\]
and similarly
\[
G(x,y)
:=
\mathbb{E}[K_B\mid X=x,Y=y]
=
\mathbb{E}_{Z_A}\!\left[G_0(x\oplus Z_A,y)\right].
\]
Thus $F$ is obtained from $F_0$ by applying $\mathsf{BSC}_p$ noise only to the $y$-variable, and $G$ is obtained from $G_0$ by applying $\mathsf{BSC}_p$ noise only to the $x$-variable.

Since, conditioned on $(X,Y)$, the random variables $K_A$ and $K_B$ depend on disjoint private randomness and disjoint channel-noise variables, they are conditionally independent. Therefore,
\[
\mathbb{E}[K_AK_B]
=
\mathbb{E}[\mathbb{E}[K_AK_B\,|\,X,Y]]=
\mathbb{E}[\mathbb{E}[K_A\,|\,X,Y] \mathbb{E}[K_B\,|\,X,Y]]
=
\mathbb{E}[F(X,Y)G(X,Y)].
\]
By the assumed correctness bias,
\[
\mathbb{E}[F(X,Y)G(X,Y)]\ge \delta. \tag{1}
\]

Because $(X,Y)$ is exactly uniform on $\{0,1\}^{\ell_A+\ell_B}$, we may expand
\[
F(x,y)=\sum_{S,T}\widehat{F}(S,T)\chi_S(x)\chi_T(y),
\qquad
G(x,y)=\sum_{S,T}\widehat{G}(S,T)\chi_S(x)\chi_T(y),
\]
and so (1) becomes
\[
\sum_{S,T}\widehat{F}(S,T)\widehat{G}(S,T)\ge \delta. \tag{2}
\]

We next exploit the fact that only Alice's transcript-coordinate is noised inside $G$.
Since $G$ is obtained from $G_0$ by applying $\mathsf{BSC}_p$ noise only to the $x$-variable,
\[
\widehat{G}(S,T)=\rho^{|S|}\widehat{G_0}(S,T),
\qquad\text{where }\rho:=1-2p.
\]
Substituting this into (2) and applying Cauchy--Schwarz yields
\[
\delta
\le
\sum_{S,T}\rho^{|S|}\widehat{F}(S,T)\widehat{G_0}(S,T)
\le
\left(\sum_{S,T}\rho^{|S|}\widehat{F}(S,T)^2\right)^{1/2}
\left(\sum_{S,T}\rho^{|S|}\widehat{G_0}(S,T)^2\right)^{1/2}.
\]
Since $|G_0(x,y)|\le 1$ for all $x,y$, Parseval implies
\[
\sum_{S,T}\rho^{|S|}\widehat{G_0}(S,T)^2
\le
\sum_{S,T}\widehat{G_0}(S,T)^2
=
\mathbb{E}[G_0(X,Y)^2]
\le 1.
\]
Hence
\[
\sum_{S,T}\rho^{|S|}\widehat{F}(S,T)^2 \ge \delta^2. \tag{3}
\]

We now show that a noticeable portion of this Fourier mass is supported on low ``bi-degree''.
Let
\[
d:=\min\{t\ge 0:\rho^{t+1}\le \delta^2/4\},
\qquad
\mathcal{L}_d:=\{(S,T): |S|\le d,\ |T|\le d\}.
\]
Since $|F(x,y)|\le 1$ for all $x,y$, Parseval gives
\[
\sum_{S,T}\widehat{F}(S,T)^2 = \mathbb{E}[F(X,Y)^2]\le 1. \tag{4}
\]

First, (4) implies that the $x$-degree of the relevant mass is small:
\[
\delta^2
\le
\sum_{|S|\le d,T}\rho^{|S|}\widehat{F}(S,T)^2
+
\sum_{|S|>d,T}\rho^{|S|}\widehat{F}(S,T)^2
\le
\sum_{|S|\le d,T}\widehat{F}(S,T)^2+\rho^{d+1}.
\]
By the choice of $d$, $\rho^{d+1}\le \delta^2/4$, and therefore
\[
\sum_{|S|\le d,T}\widehat{F}(S,T)^2 \ge \frac{3\delta^2}{4}. \tag{5}
\]

Second, because $F$ is obtained from $F_0$ by applying $\mathsf{BSC}_p$ noise only to the $y$-variable, we have
\[
\widehat{F}(S,T)=\rho^{|T|}\widehat{F_0}(S,T).
\]
Therefore
\[
\sum_{S,\,|T|>d}\widehat{F}(S,T)^2
=
\sum_{S,\,|T|>d}\rho^{2|T|}\widehat{F_0}(S,T)^2
\le
\rho^{2(d+1)}\sum_{S,T}\widehat{F_0}(S,T)^2.
\]
Since $|F_0(x,y)|\le 1$ for all $x,y$, Parseval again gives
\[
\sum_{S,T}\widehat{F_0}(S,T)^2=\mathbb{E}[F_0(X,Y)^2]\le 1,
\]
and hence
\[
\sum_{S,\,|T|>d}\widehat{F}(S,T)^2
\le
\rho^{2(d+1)}
\le
\rho^{d+1}
\le
\frac{\delta^2}{4}. \tag{6}
\]
Combining (5) and (6), we conclude that
\[
\sum_{(S,T)\in\mathcal{L}_d}\widehat{F}(S,T)^2 \ge \frac{\delta^2}{2}. \tag{7}
\]

To summarize, due to the noise-stability of~$F$ and $G$, stemming from the noise-resilience of the protocol, we concluded that~$F$ must have a surprisingly high energy in its small Fourier coefficients (those corresponding to sets of size at most~$d$ in each coordinate).
We now estimate the low bi-degree Fourier energy of $F$ from samples, which only takes samples and time proportional to the number of such coefficients, which is~$\lambda^{O(\log 1/\delta)}$.
The rest of the proof explains why that number of samples suffice for such an estimate.

Take two independent blocks of $m$ (to be chosen later) i.i.d.\ samples $\left(X^{(r)}_i, Y^{(r)}_i, K^{(r)}_{A,i}\right)_{i\in[m],r\in[2]}$ from the distribution we are testing. We would use the two independent blocks to estimate the squares of Fourier coefficients.
For $r\in\{1,2\}$ and $(S,T)\in\mathcal{L}_d$, define
\[
\widehat{c}^{(r)}_{S,T}
:=
\frac{1}{m}\sum_{i=1}^m
K^{(r)}_{A,i}\,\chi_S\!\big(X^{(r)}_i\big)\chi_T\!\big(Y^{(r)}_i\big),
\]
this is the tested correlation between the published key~$K_A$ and the Fourier coefficient corresponding to~$(S,T)$ averaged across all samples in block~$r$.
Define the statistic
\[
\mathsf{Score}
:=
\sum_{(S,T)\in\mathcal{L}_d}
\widehat{c}^{(1)}_{S,T}\widehat{c}^{(2)}_{S,T}.
\]

Under the protocol's distribution, since
\[
\mathbb{E}[K_A\mid X=x,Y=y]=F(x,y),
\]
we have that 
\[
\mathbb{E}\!\left[\widehat{c}^{(r)}_{S,T}\right]=
\mathbb{E}\!\left[K_{A}\,\chi_S\!\big(X\big)\chi_T\!\big(Y\big)\right]=
\mathbb{E}\!\left[F(X,Y)\,\chi_S\!\big(X\big)\chi_T\!\big(Y\big)\right]=
\widehat{F}(S,T).
\]
Hence, by independence of the two sample blocks,
\[
\mathbb{E}[\mathsf{Score}]
=
\sum_{(S,T)\in\mathcal{L}_d}\widehat{F}(S,T)^2
\ge \frac{\delta^2}{2} \tag{8}
\]
where the inequality is simply (7).

Under the null distribution $(U_{\ell_A},U_{\ell_B},U_1)$, the key bit is uniform and independent of $(X,Y)$, and therefore
\[
\mathbb{E}_{\mathrm{null}}\!\left[\widehat{c}^{(r)}_{S,T}\right]=0
\]
for every $(S,T)\in\mathcal{L}_d$. Consequently,
\[
\mathbb{E}_{\mathrm{null}}[\mathsf{Score}]=0. \tag{9}
\]
It is thus left to analyze variance to determine how large~$m$ should be so we can distinguish between these two distant expectations with good probability using our empirical estimation.

Let
\[
M_d:=|\mathcal{L}_d|
=
\left(\sum_{i=0}^d \binom{\ell_A}{i}\right)
\left(\sum_{j=0}^d \binom{\ell_B}{j}\right).
\]
For each block $r$,
\[
\sum_{(S,T)\in\mathcal{L}_d}\text{Var}\!\left(\widehat{c}^{(r)}_{S,T}\right)
\le \frac{M_d}{m}, \tag{10}
\]
since each $\widehat{c}^{(r)}_{S,T}$ is an average of $m$ i.i.d.\ $\{\pm1\}$-valued random variables.

For each $(S,T)\in \mathcal{L}_d$ and each block $r\in\{1,2\}$, define the estimation error
\[
\Delta^{(r)}_{S,T}
:=
\widehat{c}^{(r)}_{S,T}-\widehat{F}(S,T).
\]
Expanding each factor in the score as $\widehat{F}(S,T)+\Delta^{(r)}_{S,T}$ gives
\[
\mathsf{Score}
=
\sum_{(S,T)\in\mathcal{L}_d}\widehat{F}(S,T)^2
+
\sum_{(S,T)\in\mathcal{L}_d}\widehat{F}(S,T)\Delta^{(1)}_{S,T}
+
\sum_{(S,T)\in\mathcal{L}_d}\widehat{F}(S,T)\Delta^{(2)}_{S,T}
+
\sum_{(S,T)\in\mathcal{L}_d}\Delta^{(1)}_{S,T}\Delta^{(2)}_{S,T}.
\]
Since
\[
\mathbb{E}\!\left[\widehat{c}^{(r)}_{S,T}\right]=\widehat{F}(S,T),
\]
the three sums after the main term all have mean zero.

We now bound the variance of $\mathsf{Score}$. Write
\[
A_r
:=
\sum_{(S,T)\in\mathcal{L}_d}\widehat{F}(S,T)\Delta^{(r)}_{S,T}
\qquad (r=1,2),
\]
and
\[
B
:=
\sum_{(S,T)\in\mathcal{L}_d}\Delta^{(1)}_{S,T}\Delta^{(2)}_{S,T}.
\]
Then
\[
\mathsf{Score}-\mathbb{E}[\mathsf{Score}] = A_1+A_2+B.
\]

Because the two sample blocks are independent, and each $\Delta^{(r)}_{S,T}$ has mean zero, all
cross-terms between $A_1,A_2,B$ vanish in expectation. Therefore
\[
\mathbb{E}\!\left[\bigl(\mathsf{Score}-\mathbb{E}[\mathsf{Score}]\bigr)^2\right]
=
\mathbb{E}[A_1^2]+\mathbb{E}[A_2^2]+\mathbb{E}[B^2].
\]

For the first two terms, by Cauchy--Schwarz,
\[
A_r^2
\le
\left(\sum_{(S,T)\in\mathcal{L}_d}\widehat{F}(S,T)^2\right)
\left(\sum_{(S,T)\in\mathcal{L}_d}\bigl(\Delta^{(r)}_{S,T}\bigr)^2\right),
\]
so after taking expectations and using~(10),
\[
\mathbb{E}[A_r^2]
\le
\left(\sum_{(S,T)\in\mathcal{L}_d}\widehat{F}(S,T)^2\right)\frac{M_d}{m}.
\]
For the last term, independence of the two blocks gives
\[
\mathbb{E}[B^2]
=
\sum_{\substack{(S,T)\in\mathcal{L}_d\\(S',T')\in\mathcal{L}_d}}
\mathbb{E}\!\left[\Delta^{(1)}_{S,T}\Delta^{(1)}_{S',T'}\right]
\mathbb{E}\!\left[\Delta^{(2)}_{S,T}\Delta^{(2)}_{S',T'}\right]
\le
\left(\sum_{(S,T)\in\mathcal{L}_d} \text{Var}(\widehat{c}^{(1)}_{S,T})\right)
\left(\sum_{(S,T)\in\mathcal{L}_d} \text{Var}(\widehat{c}^{(2)}_{S,T})\right)
\le
\left(\frac{M_d}{m}\right)^2,
\]
again by~(10). We conclude that
\[
\text{Var}(\mathsf{Score})
\le
2\left(\sum_{(S,T)\in\mathcal{L}_d}\widehat{F}(S,T)^2\right)\frac{M_d}{m}
+
\left(\frac{M_d}{m}\right)^2
\le
\frac{2M_d}{m}
+
\left(\frac{M_d}{m}\right)^2.
\tag{11}
\]

Under the null distribution, every coefficient has mean zero, so
\[
\mathbb{E}_{\mathrm{null}}[\widehat{c}^{(r)}_{S,T}] = 0
\qquad\text{for all }(S,T)\in\mathcal{L}_d.
\]
Thus the analogue of the expansion above has no main term, and
\[
\mathbb{E}_{\mathrm{null}}[\mathsf{Score}] = 0,
\qquad
\mathbb{E}_{\mathrm{null}}[\mathsf{Score}^2]
\le
\left(\frac{M_d}{m}\right)^2.
\tag{12}
\]

Finally choose
\[
m := C\frac{M_d}{\delta^4}
\]
for a sufficiently large universal constant $C$. Therefore, in both distributions, 
\[
\text{Var}(\mathsf{Score})
\le
2\frac{M_d}{m}+\left(\frac{M_d}{m}\right)^2
\le O\!\left(\delta^4\right).
\]
Using $\mathbb{E}[\mathsf{Score}]\ge \delta^2/2$ from~(8), Chebyshev's inequality gives
\[
\Pr_{\mathrm{protocol}}\!\left[\mathsf{Score}<\frac{\delta^2}{4}\right]
\le
\Pr_{\mathrm{protocol}}\!\left[\left|\mathsf{Score}-\mathbb{E}[\mathsf{Score}]\right|\ge \frac{\delta^2}{4}\right]
\le \frac16,
\]
for large enough $C$.
Similarly, under the null distribution, using $\mathbb{E}_{\mathrm{null}}[\mathsf{Score}]=0$ and~(12),
\[
\Pr_{\mathrm{null}}\!\left[\mathsf{Score}\ge \frac{\delta^2}{4}\right]
\le
\Pr_{\mathrm{null}}\!\left[\left|\mathsf{Score}\right|\ge \frac{\delta^2}{4}\right]
\le \frac16.
\]

Therefore the test that outputs ``protocol'' if and only if $\mathsf{Score}\ge \delta^2/4$ distinguishes the two distributions with probability at least $2/3$.

This proves the theorem in the exact-uniform transcript case.
It remains to bound the complexity.
Since $p\in(0,1/2)$ is a fixed constant, $\rho=1-2p$ is a fixed constant in $(0,1)$, and therefore
\[
d=O\!\left(\log \frac{1}{\delta}\right).
\]
Also, because $\ell_A,\ell_B=\Theta(\lambda)$,
\[
M_d
=
\left(\sum_{i=0}^d \binom{\ell_A}{i}\right)
\left(\sum_{j=0}^d \binom{\ell_B}{j}\right)
=
\lambda^{O(d)}
=
\lambda^{O(\log \frac{1}{\delta})}.
\]
Since $\delta\ge 1/\mathsf{poly}(\lambda)$, the extra factor $\delta^{-4}$ is polynomial in $\lambda$, and therefore
\[
N=O\!\left(\frac{M_d}{\delta^4}\right)
=
O\!\left(\lambda^{O(\log \frac{1}{\delta})}\right).
\]
This is quasipolynomial in $\lambda$, as claimed.

To remove the exact-uniformity assumption, note that the argument above constructs an explicit test: given
\[
N=O\!\left(\frac{M_d}{\delta^4}\right)
\]
independent samples and proportional time, it computes the low bi-degree statistic $\mathsf{Score}$ and compares it to the threshold $\delta^2/4$.

Now suppose the public transcript of the protocol is only \emph{computationally} pseudorandom rather than exactly uniform.
If the same test $\mathsf{Score}$ already distinguishes
\[
(m_A,m_B,K_A)
\qquad\text{from}\qquad
(m_A,m_B,U_1),
\]
then we have obtained a quasipolynomial attack on key secrecy, and we are done.
Otherwise, the only place where the proof above can fail is in assuming uniformity of $(m_A,m_B)$.
Consequently, if the key-secrecy attack above does not go through on the actual protocol distribution, then the transcript distribution itself must be detectable by the same family of low-degree tests.
Equivalently, there is a quasipolynomial-time distinguisher, using quasipolynomially many samples, that distinguishes
\[
(m_A,m_B)
\qquad\text{from}\qquad
(U_{\ell_A},U_{\ell_B}).
\]

We conclude that every non-interactive PNR-KE protocol admits a quasipolynomial attack of one of the following two forms:
either one can distinguish
\[
(m_A,m_B,K_A)
\qquad\text{from}\qquad
(m_A,m_B,U_1),
\]
thereby breaking key secrecy, or one can distinguish
\[
(m_A,m_B)
\qquad\text{from}\qquad
(U_{\ell_A},U_{\ell_B}),
\]
thereby breaking transcript pseudorandomness.
In both cases the distinguisher uses
\[
N = O\!\left(\lambda^{O(\log \frac{1}{\delta})}\right)
\]
samples and runs in time polynomial in $N$, which is quasipolynomial for $\delta\ge 1/\poly(\lambda)$.
\end{proof}

\section{PNR-KE with Noiseless Feedback}\label{sec:feedback_pnr-ke}
To apply the reduction of Section~\ref{subsec:roadmap_short}, it suffices to use a weaker variant of PNR-KE than Definition~\ref{def:pnrke}: we may assume that before we send a bit through the noisy channel we are given the exact probability~$p_t\leq p$ with which it will be flipped, and that after transmitting the bit we know (exactly) whether the bit was flipped; These are the guarantees provided by Theorem~\ref{thm:bundle}. Furthermore, we restrict the distinguisher to see only the noisy transcript, which suffices for our application as the noise is in fact added by each party ``internally'', as part of the bundling sampler, and not by the channel itself. 
Under these settings, we show that \emph{any} PR-KE protocol can be converted into a PNR-KE protocol.
In particular, for our main application we return to (conjectured) exponential hardness rather than the quasi-polynomial hardness of the PNR-KE protocol of Section~\ref{sec:prnr-ke}.

Our main result is the following reduction.

\begin{theorem}[PR-KE implies feedback PNR-KE with logarithmic blowup]
\label{thm:ut_to_feedback_pnr_log}
Fix~$p\in(0,\frac12)$.
Let~$\Pi$ be a PR-KE protocol whose noiseless public transcript has length~$T$ and whose correctness is at least~$1-\delta$.
For any~$\eta>0$, there exists a compiled key-exchange protocol~$\widetilde{\Pi}$ over~$\BSC(\le p)$ with noiseless feedback of flips, using~$O_{p}(T\log {\frac{T}{\eta}})$ channel uses, such that:

\begin{enumerate}
    \item The public transcript of~$\widetilde{\Pi}$ is computationally indistinguishable from uniform; and

    \item $\widetilde{\Pi}$ succeeds with probability at least~$1-\delta-\eta$.
\end{enumerate}

In particular, any PR-KE protocol of transcript length~$T$ and negligible failure probability yields a feedback variant of PNR-KE of any total length~$\omega(T\log T)$ and negligible failure probability.
\end{theorem}

We obtain this reduction by studying the following two-player (which we call \emph{sender} and \emph{receiver}) ``random-looking signaling'' game.
The sender draws $o\in\{\pm1\}$ uniformly.
The sender communicates over a binary symmetric channel $\text{BSC}(\leq p)$ for $n$ channel uses.
At time $t$, the sender transmits $X_i\in\{\pm1\}$ and the receiver observes
\[
Y_i = X_i Z_i\qquad \text{where} \qquad Z_i\in\{\pm1\}\quad\text{and}\quad \Pr[Z_i=-1]=p_t\leq p,
\]
with $(Z_i)_{i=1}^n$ independently distributed and independent of $o$, and assuming each~$p_t$ is known before sending the~$t$-th bit.
The receiver sees only $Y=(Y_1,\dots,Y_n)$.

\paragraph{Goal.}
Design a (possibly randomized) sender strategy such that
\begin{enumerate}
\item \textbf{Perfect indistinguishability:} $Y$ is exactly uniform on $\{\pm1\}^n$ unconditionally.
\item \textbf{Signaling $o$:} There exists an efficiently computable~$f:\{\pm1\}^n\rightarrow \{\pm 1\}$ such that $\Pr[f(Y)=o]$ is as large as possible.
\end{enumerate}

\paragraph{Feedback model.}
We assume \emph{noiseless feedback of flips}: after each channel use $t$, the sender learns $Z_t$ (equivalently, learns $Y_t$ since the sender knows $X_t$), while the receiver does not.
Thus the sender may adapt $X_{t+1}$ to $(o,X_1,\ldots,X_t,Z_1,\dots,Z_t)$.
Note that indistinguishability is defined only with respect to~$Y$; that is, the distinguisher does not see the feedback.

\subsection{Reduction from PNR-KE with Feedback to PR-KE}
We present a simple black-box reduction alluded to above.
Suppose we are given a PR-KE protocol~$\Pi$ whose noiseless public transcript consists of~$T$ bits, and suppose also that we have a signaling protocol for the game above that uses~$n$ channel uses in order to signal the objective bit~$o$.

In the compiled protocol, whenever~$\Pi$ instructs one of the parties to send a bit~$b\in\{\pm1\}$, the sender and receiver replace this single noiseless transmission by one independent execution of the signaling protocol with objective~$o=b$.
The receiver applies the decoder~$f$ to the resulting public string~$Y\in\{\pm1\}^n$ and uses the decoded bit~$\hat b=f(Y)$ as the received bit in its simulation of~$\Pi$.
Since before each channel use the sender is told the corresponding crossover probability~$p_t$, and after each use the sender learns whether that bit was flipped, the signaling protocol can indeed be executed in our feedback model.

The resulting protocol has total public transcript length~$nT$.
Its correctness loss is exactly the accumulated probability that at least one of the~$T$ simulated noiseless transmissions is decoded incorrectly, while pseudorandomness follows from the fact that replacing a uniform bit by one signaling block yields an exactly uniform public block.

\begin{proposition}[Black-box compilation using the signaling protocol]
\label{prop:compile-ut-to-feedback-pnr}
Assume there exists a length-$n$ signaling protocol for the game above with decoder~$f:\{\pm1\}^n\to\{\pm1\}$ such that:

\begin{enumerate}
    \item For every admissible choice of crossover probabilities~$p_1,\ldots,p_n\le p$ and every objective bit~$o\in\{\pm1\}$,
    \[
    \Pr[f(Y)\neq o]\le \varepsilon_{\mathrm{sig}},
    \]
    where the probability is over the internal randomness of the sender strategy and the channel noise; and

    \item If~$o$ is uniform in~$\{\pm1\}$, then the resulting public transcript~$Y$ is exactly uniform on~$\{\pm1\}^n$.
\end{enumerate}

Let~$\Pi$ be a PR-KE protocol whose noiseless public transcript is of length~$T$-bits, and suppose~$\Pi$ succeeds with probability at least~$1-\delta$.
Then there is a compiled key-exchange protocol~$\widetilde{\Pi}$ over~$\mathrm{BSC}(\le p)$ with noiseless feedback of flips, using~$nT$ channel uses, such that:

\begin{enumerate}
    \item \textbf{Correctness:}
    \[
    \Pr[\widetilde{\Pi}\text{ succeeds}] \ge 1-\delta-T\cdot\varepsilon_{\mathrm{sig}}.
    \]

    \item \textbf{Pseudorandomness:}
    The public transcript of~$\widetilde{\Pi}$ is computationally indistinguishable from uniform on~$\{\pm1\}^{nT}$.
\end{enumerate}
\end{proposition}

\begin{proof}
The compiled protocol simply replaces each noiseless transcript bit~$U_j$ of~$\Pi$ by an independent execution of the signaling protocol with objective~$o=U_j$, and the receiver decodes the transmitted bit as~$\hat U_j=f(Y^{(j)})$.

For correctness, let~$E_j$ be the event that the~$j$-th simulated transmission is decoded incorrectly, i.e.\ $\hat U_j\neq U_j$.
By assumption, $\Pr[E_j]\le \varepsilon_{\mathrm{sig}}$ for every~$j\in[T]$.
Hence by a union bound,
\[
\Pr\Big[\bigcup_{j=1}^T E_j\Big]\le T\varepsilon_{\mathrm{sig}}.
\]
On the complementary event, every simulated transmission is decoded correctly, so the parties perfectly simulate the original protocol~$\Pi$.
Therefore the success probability drops by at most~$T\varepsilon_{\mathrm{sig}}$, giving
\[
\Pr[\widetilde{\Pi}\text{ succeeds}] \ge 1-\delta-T\varepsilon_{\mathrm{sig}}.
\]

For pseudorandomness, first consider the case where the original noiseless transcript~$U$ is truly uniform on~$\{\pm1\}^T$.
Then each block objective~$o=U_j$ is a fresh uniform bit, and by the defining property of the signaling protocol the corresponding public block~$Y^{(j)}$ is exactly uniform on~$\{\pm1\}^n$.
Since the signaling executions are independent across~$j$, the full compiled public transcript is exactly uniform on~$\{\pm1\}^{nT}$.

Now return to the actual protocol~$\Pi$.
Since~$\Pi$ has pseudorandom transcripts, its noiseless transcript~$U$ is computationally indistinguishable from uniform.
Applying the efficient randomized transformation that replaces each bit~$U_j$ by one signaling block preserves computational indistinguishability.
As the image of a truly uniform transcript under this transformation is exactly uniform on~$\{\pm1\}^{nT}$, the compiled public transcript is computationally indistinguishable from uniform as well.
\end{proof}

In essence, we can think of the compiler above as adding~$\varepsilon_{sig}$ probability of error to every transmitted bit.

\subsection{Warm-Up and Motivation}

Assume throughout that~$n$ is odd.
A natural choice for the decoding function~$f$ is the majority function
\[
\Maj(Y)=\sgn\!\left(\sum_t Y_t\right),
\]
since it is balanced, relatively stable to noise, and has low individual coordinate influences.
If there were no channel noise ($p=0$), the sender could simply sample~$Y$ uniformly conditioned on~$\Maj(Y)=o$; since~$o$ is uniform and, for odd~$n$, the sets
\[
\{Y \mid \Maj(Y)=+1\}
\qquad\text{and}\qquad
\{Y \mid \Maj(Y)=-1\}
\]
have equal size, the unconditional distribution of~$Y$ is uniform, and the signaling is perfect.

Now suppose that $p>0$, but that the sender does \emph{not} use feedback.
The sender can still sample~$X$ uniformly conditioned on~$\Maj(X)=o$ and transmit it.
Then~$X$ is marginally uniform, and hence so is the noisy~$Y=(X_1Z_1,\ldots,X_nZ_n)$; thus condition~(1) holds perfectly.
Condition~(2) is then exactly the noise stability of majority, which converges to a constant depending only on~$p$.
In particular, this error does not vanish as~$n$ grows.
More generally, the no-feedback version of the signaling game is governed precisely by the noise stability of the decoding function~$f$, and therefore even as~$n\to\infty$ one inherently has~$\varepsilon=\Omega(p)$.

This motivates using feedback in order to drive the success probability toward~$1$ while still preserving perfect uniformity of~$Y$.
A natural next attempt is to use the feedback to correct deviations from the target condition~$\Maj(Y)=o$:
when transmitting the~$t$-th bit, the sender already knows the previously received bits~$Y_1,\ldots,Y_{t-1}$, and so may try to sample the next received bit according to the conditional distribution of~$Y_t$ given both the prefix~$Y_1,\ldots,Y_{t-1}$ and the \emph{terminal} event~$\Maj(Y)=o$.

However, this runs into two closely related obstacles.
First, this conditional distribution need not even be feasible, since the observed prefix may already force~$\Maj(Y)\neq o$.
Second, the sender does not directly control the distribution of the received bit~$Y_t$; rather, the sender chooses the transmitted bit~$X_t$, which is then corrupted by noise.
The key observation is that as long as the target conditional distribution satisfies
\[
\Pr[Y_t=+1]\in[p,1-p],
\]
it is in fact possible to choose~$X_t$ so that the resulting received bit~$Y_t=X_t Z_t$ has exactly this desired distribution, since before transmission we know the exact crossover probability~$p_t:=\Pr[Z_t=-1]\le p$.
While the hard constraint~$\Maj(Y)=o$ does \emph{not} satisfy such a bound for every next-bit conditional probability, we will define a suitable relaxation of it that does, at the cost of only a small loss in the signaling success probability.

\subsection{Majority Signaling}
\label{subsec:optimal_majority_signaling}

We now replace the hard terminal constraint~$\Maj(Y)=o$ by a softer terminal rule that still strongly correlates with majority, but leaves \emph{every} endpoint~$y\in\{\pm1\}^n$ possible under either objective.
This is exactly what will allow us to keep all next-step conditional probabilities bounded away from~$0$ and~$1$.
Let
\[
S(y):=\sum_{i=1}^n y_i,
\qquad\text{and}\qquad
\sigma(u):=\frac{1}{1+e^{-u}}.
\]
Fix a parameter~$\beta\ge 0$, to be chosen later, and define the soft terminal rule
\begin{equation}\label{eq:soft-terminal-pi}
\pi(y):=\sigma(\beta S(y)).
\end{equation}
Equivalently, once the terminal word~$y$ is reached, we declare that it is compatible with objective~$+1$ with probability~$\pi(y)$ and with objective~$-1$ with probability~$1-\pi(y)$.
Thus large positive values of~$S(y)$ strongly favor~$o=+1$, large negative values strongly favor~$o=-1$, and every endpoint retains nonzero probability for both objectives.
We therefore target the conditional law
\begin{equation}\label{eq:target-law-majority}
\Pr[Y=y\mid o]
=
2^{-(n-1)} \sigma\!\big(\beta\,\cdot o\,\cdot S(y)\big).
\end{equation}
These are well defined probabilities since for every~$u\in \mathbb{R}$, we have~$\sigma(u)\in (0,1)$ and furthermore $\sigma(u)+\sigma(-u)=1$. Thus, for every~$y$ we have
\[
\Pr[Y=y\mid o=+1]+\Pr[Y=y\mid o=-1]=2^{-(n-1)},
\]
and therefore the unconditional distribution of~$Y$ is exactly uniform.
Moreover, under~\eqref{eq:target-law-majority},
\begin{equation}\label{eq:posterior-majority}
\Pr[o=+1\mid Y=y]=\sigma(\beta S(y)).
\end{equation}
In particular, the maximum likelihood estimator of~$o$ from~$Y$ is precisely majority:
\[
\arg\max_{b\in\{\pm1\}}\Pr[o=b\mid Y=y] = \sgn(S(y))=\Maj(y),
\]
where we use that~$n$ is odd.

It remains to show that the law~\eqref{eq:target-law-majority} can be sampled online through the noisy channel with feedback, and analyze the signaling probability.

\paragraph{Aligned variables.}
Define
\[
U_i:=oY_i,\qquad V_i:=oX_i,
\]
so that~$U_i=V_i Z_i$ and~$\Maj(Y)=o$ is equivalent to~$\sum_{i=1}^n U_i>0$.
In these variables, the target law becomes
\begin{equation}\label{eq:aligned-terminal-law}
\Pr[U=u\mid o=+1]
=
2^{-n} w\!\Big(\sum_{i=1}^n u_i\Big),
\qquad\text{where}\qquad
w(s):=2\sigma(\beta s)\in(0,2).
\end{equation}
Thus the uniform measure on~$\{\pm1\}^n$ is reweighted only according to the terminal sum.

\paragraph{Conditional success probabilities.}
For $r\ge 0$, denote the result of the length~$r$ simple random walk by
\[
W_r:=\xi_1+\cdots+\xi_r,
\]
where the~$\xi_i$ are i.i.d.\ uniform in~$\{\pm1\}$, and define the continuation values
\begin{equation}\label{eq:H-majority}
H[r,s]:=\mathbb{E}\!\left[w(s+W_r)\right].
\end{equation}
The value~$\frac{1}{2}H[r,s]$ is exactly the probability of (eventually) satisfying the terminal condition if we have~$r$ bits left to send and the current sum of the bits sent so far is~$t$.
These satisfy the recursion
\begin{equation}\label{eq:H-majority-rec}
H[0,s]=w(s),
\qquad
H[r,s]=\tfrac12 H[r-1,s+1]+\tfrac12 H[r-1,s-1]
\quad (r\ge 1).
\end{equation}
Hence all values~$H[r,s]$ for~$0\le r\le n$ can be precomputed in time~$O(n^2)$.

Suppose that after~$t$ steps the aligned partial sum is
\[
S_t:=\sum_{i=1}^t U_i=s,
\]
so that~$r:=n-t$ steps remain.
The correct conditional probability for the next aligned received bit is then
\begin{equation}\label{eq:q-majority}
q(r,s)
:=
\Pr[U_{t+1}=+1 \mid S_t=s]
=
\frac{H[r-1,s+1]}{H[r-1,s+1]+H[r-1,s-1]}.
\end{equation}

\begin{lemma}[Correctness of the sequential sampler]
\label{lem:majority-sequential-sampler}
If the aligned received bits~$U_1,\ldots,U_n$ are sampled sequentially according to the transition probabilities~\eqref{eq:q-majority}, then the resulting joint law of~$U$ is exactly~\eqref{eq:aligned-terminal-law}.
\end{lemma}

\begin{proof}
Fix a prefix~$u_1,\ldots,u_t$ with sum~$s$, and let~$r=n-t$.
Under the target law~\eqref{eq:aligned-terminal-law}, the probability of this prefix is proportional to
\[
2^{-t}\,\mathbb{E}\!\left[w(s+W_r)\right]
=
2^{-t} H[r,s].
\]
Therefore
\[
\Pr[U_{t+1}=+1\mid U_1=u_1,\ldots,U_t=u_t]
=
\frac{\frac12 H[r-1,s+1]}{\frac12 H[r-1,s+1]+\frac12 H[r-1,s-1]},
\]
which is exactly~\eqref{eq:q-majority}.
Thus the sequential process matches the target conditional probabilities at every step, and hence produces the target joint law.
\end{proof}

We next verify that these desired next-step probabilities are always feasible through the noisy channel.

\begin{lemma}[Channel feasibility]
\label{lem:majority-feasibility}
Fix any time~$t$, and suppose the actual crossover probability at that step is~$p_t\le p$.
If the sender chooses~$V_t=+1$ with probability~$a_t$, then
\begin{equation}\label{eq:q-from-a-pt}
\Pr[U_t=+1\mid \text{sender information}]
=
p_t+(1-2p_t)a_t.
\end{equation}
Consequently, any target value~$q\in[p_t,1-p_t]$ can be implemented exactly by setting
\begin{equation}\label{eq:a-from-q-pt}
a_t=\frac{q-p_t}{1-2p_t}.
\end{equation}
\end{lemma}

\begin{proof}
Since~$U_t=V_t Z_t$, we have
\[
\Pr[U_t=+1]=(1-p_t)\Pr[V_t=+1]+p_t\Pr[V_t=-1]
=
(1-p_t)a_t+p_t(1-a_t),
\]
which is exactly~\eqref{eq:q-from-a-pt}.
Solving for~$a_t$ gives~\eqref{eq:a-from-q-pt}.
\end{proof}

Thus it suffices to show that~$q(r,s)$ always lies in the smaller interval~$[p,1-p]$.
This is where the logistic choice becomes particularly convenient.

\begin{lemma}[Uniform bound on all next-step biases]
\label{lem:majority-odds-ratio}
For every~$r\ge 1$ and every integer~$s$,
\[
\frac{q(r,s)}{1-q(r,s)}
=
\frac{H[r-1,s+1]}{H[r-1,s-1]}
\in
\big[e^{-2\beta},\,e^{2\beta}\big].
\]
Equivalently,
\[
q(r,s)\in [\sigma(-2\beta),\,\sigma(2\beta)].
\]
\end{lemma}

\begin{proof}
For every integer~$t$,
\[
\frac{w(t+1)}{w(t-1)}
=
\frac{\sigma(\beta(t+1))}{\sigma(\beta(t-1))}
=
\frac{1+e^{-\beta(t-1)}}{1+e^{-\beta(t+1)}}
\in [e^{-2\beta},e^{2\beta}].
\]
Now fix~$r,s$, and let~$W=W_{r-1}$.
Applying the above pointwise with~$t=s+W$ yields
\[
w(s+1+W)\in [e^{-2\beta},e^{2\beta}]\, w(s-1+W).
\]
Taking expectations gives
\[
H[r-1,s+1]\in [e^{-2\beta},e^{2\beta}]\, H[r-1,s-1],
\]
which is exactly the claimed odds-ratio bound.
\end{proof}

We now choose
\begin{equation}\label{eq:beta-choice-majority}
\beta:=\frac12\log\frac{1-p}{p}.
\end{equation}
Then
\[
\sigma(2\beta)=1-p
\qquad\text{and}\qquad
\sigma(-2\beta)=p,
\]
so Lemma~\ref{lem:majority-odds-ratio} gives
\[
q(r,s)\in[p,1-p]\subseteq [p_t,1-p_t]
\qquad\text{for every }p_t\le p.
\]
Hence, by Lemma~\ref{lem:majority-feasibility}, every desired transition probability~$q(r,s)$ is implementable exactly at every step, for every realized sequence~$(p_t)$.

We next analyze the signaling success probability of the majority decoder.
Since under~\eqref{eq:target-law-majority} the unconditional distribution of~$Y$ is exactly uniform, this reduces to a simple weighted estimate over the possible terminal sums.

\begin{lemma}[Exact failure formula]
\label{lem:majority-failure-formula}
Under the signaling strategy above,
\[
\Pr[\Maj(Y)\neq o]
=
\mathbb{E}_{Y\sim\Unif(\{\pm1\}^n)}
\left[\sigma\!\left(-\beta\left|\sum_{i=1}^n Y_i\right|\right)\right].
\]
Equivalently, if
\[
W_n:=\xi_1+\cdots+\xi_n
\qquad\text{with}\qquad
\xi_1,\ldots,\xi_n \stackrel{\mathrm{i.i.d.}}{\sim} \Unif(\{\pm1\}),
\]
then
\[
\Pr[\Maj(Y)\neq o]
=
\mathbb{E}\big[\sigma(-\beta |W_n|)\big]
=
\sum_{s=1}^{n}
\Pr[|W_n|=s]\cdot \sigma(-\beta s).
\]
\end{lemma}

\begin{proof}
By~\eqref{eq:posterior-majority}, for every~$y\in\{\pm1\}^n$ the posterior probability of error of the majority decoder is
\[
\Pr[\Maj(Y)\neq o\mid Y=y]
=
\sigma\!\left(-\beta\left|\sum_{i=1}^n y_i\right|\right).
\]
Averaging over~$Y$, which is uniform by construction, gives the first identity.
The second is just a reparametrization by the endpoint
\[
W_n=\sum_{i=1}^n Y_i.
\qedhere
\]
\end{proof}

\begin{proposition}[Majority decoding succeeds with probability $1-O(n^{-1/2})$]
\label{prop:majority-failure-bound}
For every fixed~$\beta>0$ there is a constant~$C=C_\beta<\infty$ such that under the signaling strategy above,
\[
\Pr[\Maj(Y)\neq o]\le \frac{C}{\sqrt n}.
\]
In particular, for any~$p\in(0,\frac{1}{2})$ and the choice~$\beta=\frac12\log\frac{1-p}{p}$ from~\eqref{eq:beta-choice-majority}, the constant~$C<\infty$ depends only on~$p$.
\end{proposition}

\begin{proof}
By Lemma~\ref{lem:majority-failure-formula},
\[
\Pr[\Maj(Y)\neq o]
=
\sum_{s=1}^{n}
\Pr[|W_n|=s]\cdot \sigma(-\beta s).
\]
We bound the two factors separately.
First, note that~$|W_n|$ must be odd (as~$n$ is), and that for every odd~$s\in\{1,3,\ldots,n\}$,
\[
\Pr[W_n=s]
=
2^{-n}\binom{n}{\frac{n+s}{2}}
\le
2^{-n}\binom{n}{\frac{n+1}{2}}
\le \frac{C'}{\sqrt n},
\]
where~$C'>0$ is an absolute constant, using that the binomial coefficients are maximized at the center and the standard bound on the central binomial coefficient.
By symmetry,
\[
\Pr[|W_n|=s]\le \frac{2C'}{\sqrt n}
\qquad\text{for every odd }s.
\]
Second, for every~$s\ge1$,
\[
\sigma(-\beta s)=\frac{1}{1+e^{\beta s}}\le e^{-\beta s}.
\]
Combining the two estimates,
\[
\Pr[\Maj(Y)\neq o]
\le
\frac{2C'}{\sqrt n}
\sum_{s=1}^{\infty} e^{-\beta s}
=
\frac{2C'}{\sqrt n}\cdot \frac{e^{-\beta}}{1-e^{-\beta}}=O_\beta \left(\frac{1}{\sqrt{n}}\right).
\]
For~$\beta=\frac12\log\frac{1-p}{p}$, this is~$O\left(\frac{\sqrt{p}}{\sqrt{1-p}\,-\sqrt{p}}\cdot\frac{1}{\sqrt{n}}\right)$.
\end{proof}

Combining the above lemmas yields the signaling scheme.

\begin{theorem}[Majority signaling with perfect uniformity]
\label{thm:majority-signaling}
Let~$p\in (0,\frac{1}{2})$, then there exists a constant~$C>0$ such that for every odd~$n$ the following holds.
There is an efficiently computable signaling strategy in the feedback model such that:

\begin{enumerate}
    \item The unconditional distribution of~$Y$ is exactly uniform on~$\{\pm1\}^n$;

    \item The optimal decoder for~$o$ given~$Y$ is~$\Maj(Y)$; and

    \item The signaling error probability satisfies
    \[
    \Pr[\Maj(Y)\neq o]\le \frac{C}{\sqrt n}.
    \]
\end{enumerate}
\end{theorem}

\begin{proof}
The existence and efficient computability of the strategy follow from Lemmas~\ref{lem:majority-sequential-sampler}, \ref{lem:majority-feasibility}, and~\ref{lem:majority-odds-ratio}: the sender precomputes the table~$H[r,s]$, keeps track of the current aligned partial sum~$S_t$, computes the desired next-step bias~$q(r,s)$ from~\eqref{eq:q-majority}, and then chooses the distribution of~$V_{t+1}$ so that the resulting received aligned bit~$U_{t+1}$ has exactly this bias.

By Lemma~\ref{lem:majority-sequential-sampler}, the resulting law of~$U$ is exactly~\eqref{eq:aligned-terminal-law}, and therefore the law of~$Y$ is exactly~\eqref{eq:target-law-majority}.
As observed above, this implies perfect uniformity of~$Y$ and the posterior identity~\eqref{eq:posterior-majority}, so majority is the optimal decoder.
The error bound is exactly Proposition~\ref{prop:majority-failure-bound}.
\end{proof}

\begin{corollary}
\label{cor:ut-to-feedback-pnr-poly}
Fix~$p\in(0,\frac12)$, and let~$\Pi$ be a PR-KE protocol whose noiseless public transcript has length~$T$ and whose correctness is bounded below by a positive constant.
Then~$\Pi$ can be compiled into a feedback PNR-KE protocol over~$\BSC(\le p)$ of total length~$O(T^3)$, while preserving pseudorandomness and constant correctness.
\end{corollary}

\begin{proof}
By Theorem~\ref{thm:majority-signaling}, there is a constant~$C=C(p)$ such that one noiseless transcript bit can be simulated using~$n$ channel uses with signaling error at most~$C/\sqrt n$.
Choose~$n=\Theta(T^2)$ odd, with the hidden constant large enough so that
\[
\frac{C}{\sqrt n}\le \frac{1}{100T}.
\]
Then the total error accumulated over the~$T$ simulated transcript bits is at most~$1/100$ by Proposition~\ref{prop:compile-ut-to-feedback-pnr}.
Therefore the compiled protocol has total length~$nT=O(T^3)$, its public transcript remains pseudorandom, and its correctness remains bounded below by a positive constant.
\end{proof}

\begin{remark}[The $n^{-1/2}$ error rate is optimal for majority]
\label{rem:majority-optimal}
For fixed~$p\in(0,\frac12)$, the error bound of Theorem~\ref{thm:majority-signaling} is optimal up to constants for the majority decoder.
Indeed, assume~$n$ is odd and the public transcript~$Y$ is exactly uniform.
Let
\[
E:=\Bigl\{\sum_{i=1}^{n-1} Y_i = 0\Bigr\}.
\]
Then~$\Pr[E]=\Theta(n^{-1/2})$.
Conditioned on~$E$, the decoder~$\Maj(Y)$ is determined solely by the last received bit~$Y_n$.
Hence, even allowing the sender to use the full feedback from the first~$n-1$ channel uses, recovering~$o$ on this event amounts to transmitting one final bit through a~$\BSC(p)$, and therefore incurs error at least~$p$.
It follows that
\[
\Pr[\Maj(Y)\neq o]\ge p\cdot \Pr[E]=\Omega(n^{-1/2}).
\]
The same argument applies to any symmetric decoder whose values on the two central layers~$\sum_i Y_i=\pm1$ are different.
\end{remark}

\subsection{Optimal Signaling}
\label{subsec:optimal_signaling}

The previous subsection showed that majority already gives a polynomially small signaling error.
However, as observed in Remark~\ref{rem:majority-optimal}, no decoder that depends only on the final Hamming weight (and more generally, no symmetric decoder that distinguishes the two middle layers) can beat the $\Theta(n^{-1/2})$ barrier under perfect uniformity.
To obtain exponentially small error, the decoder must use the \emph{entire} transcript.

The natural choice is therefore the \emph{maximum a posteriori} (MAP) decoder for the full transcript.
This viewpoint is closely related to the feedback-coding ideas of Horstein and of posterior matching~\cite{horstein2003sequential,shayevitz2011optimal}: at each step, we track the receiver's posterior on the hidden bit~$o$, and choose the next-bit laws under the two hypotheses so that
\begin{enumerate}
    \item The unconditional next bit remains perfectly uniform, and
    \item The two conditional laws stay separated by a constant amount.
\end{enumerate}
Iterating this causes the two transcript distributions to separate exponentially fast.
Note that unlike the majority signaling protocol, we now don't aim for a cleanly-defined terminal condition -- instead, we compute the next-bit probabilities ``on the fly'' based on the received prefix so far (which is known to both parties exactly) and the objective~$o$. The receiver can thus compute the two options (corresponding to each value of~$o$) and compare the received bits to the distribution with which the sender would have sent them if it attempts to signal each value of~$o$.
For a prefix~$y_{1:t}\in\{\pm1\}^t$, write
\[
P_+^t(y_{1:t}) := \Pr[Y_{1:t}=y_{1:t}\mid o=+1],
\qquad
P_-^t(y_{1:t}) := \Pr[Y_{1:t}=y_{1:t}\mid o=-1],
\]
and let
\[
w_t(y_{1:t})
:=
\Pr[o=+1\mid Y_{1:t}=y_{1:t}]
=
\frac{P_+^t(y_{1:t})}{P_+^t(y_{1:t})+P_-^t(y_{1:t})},
\]
where the prior on~$o$ is uniform.
We abbreviate this posterior by~$w_t$ when the prefix is clear from context, and initialize~$w_0=\frac12$.
At time~$t$, after observing the prefix~$Y_{1:t}$, we define target next-bit probabilities
\[
q_t^+ := \Pr[Y_{t+1}=+1\mid o=+1,Y_{1:t}],
\qquad
q_t^- := \Pr[Y_{t+1}=+1\mid o=-1,Y_{1:t}]
\]
as follows:
If~$w_t\ge \frac12$, set
\begin{equation}\label{eq:optimal-q-plus-branch}
q_t^- := p,
\qquad
q_t^+ := \frac{\frac12-(1-w_t)p}{w_t}.
\end{equation}
If~$w_t<\frac12$, set
\begin{equation}\label{eq:optimal-q-minus-branch}
q_t^+ := p,
\qquad
q_t^- := \frac{\frac12-w_t p}{1-w_t}.
\end{equation}

In words: at every step we pin the \emph{less likely} hypothesis to the extreme value~$p$, and choose the more likely hypothesis so as to keep the unconditional next bit perfectly fair.

\begin{lemma}[Basic properties of the control rule]
\label{lem:optimal-control-basic}
For every~$t$ and every observed prefix~$Y_{1:t}$, the numbers~$q_t^+,q_t^-$ satisfy:

\begin{enumerate}
    \item \(q_t^+,q_t^- \in [p,1-p]\);

    \item
    \[
    w_t q_t^+ + (1-w_t) q_t^- = \frac12;
    \]

    \item One of~$q_t^+,q_t^-$ is equal to~$p$, and the other belongs to~$[\frac12,1-p]$.
\end{enumerate}
\end{lemma}

\begin{proof}
Suppose first that~$w_t\ge \frac12$.
Then by definition~$q_t^-=p$, and
\[
q_t^+ = \frac{\frac12-(1-w_t)p}{w_t}
= p + \frac{\frac12-p}{w_t}.
\]
Since~$w_t\in[\frac12,1]$, it follows that
\[
\frac12 \le q_t^+ \le 1-p.
\]
Moreover,
\[
w_t q_t^+ + (1-w_t)q_t^-
=
w_t \cdot \frac{\frac12-(1-w_t)p}{w_t} + (1-w_t)p
=
\frac12.
\]

The case~$w_t<\frac12$ is symmetric: now~$q_t^+=p$ and
\[
q_t^- = \frac{\frac12-w_t p}{1-w_t}
\]
lies in~$[\frac12,1-p]$, while
\[
w_t q_t^+ + (1-w_t)q_t^- = \frac12.
\]
This proves all three claims.
\end{proof}

As in the previous subsection, any target next-bit probability in~$[p,1-p]$ is implementable exactly through the noisy channel.
Indeed, if the actual crossover probability at time~$t+1$ is~$p_{t+1}\le p$, then for any desired~$q\in[p,1-p]\subseteq [p_{t+1},1-p_{t+1}]$ the sender can choose the distribution of~$X_{t+1}$ so that the resulting received bit satisfies
\[
\Pr[Y_{t+1}=+1\mid \text{sender information}] = q.
\]
Thus the sender can realize the law~$q_t^+$ when~$o=+1$ and the law~$q_t^-$ when~$o=-1$.

The posterior itself can be updated efficiently from the observed transcript, since by Lemma~\ref{lem:optimal-control-basic},
\[
\Pr[Y_{t+1}=+1\mid Y_{1:t}] = w_t q_t^+ + (1-w_t)q_t^- = \frac12,
\]
Bayes' rule gives
\begin{equation}\label{eq:optimal-posterior-update}
w_{t+1}
=
\begin{cases}
2 w_t q_t^+ & \text{if } Y_{t+1}=+1,\\[1ex]
2 w_t (1-q_t^+) & \text{if } Y_{t+1}=-1.
\end{cases}
\end{equation}
Hence both sender and receiver can maintain~$w_t$ online in time polynomial in~$n$.

\begin{theorem}[Exponential-error signaling with perfect uniformity]
\label{thm:optimal-signaling}
Fix~$p\in(0,\frac12)$, and define
\[
\lambda(p):=\frac{\sqrt p+\sqrt{1-p}}{\sqrt2}<1.
\]
Then for every~$n$ there is an efficiently computable signaling strategy in the feedback model such that:

\begin{enumerate}
    \item The unconditional distribution of~$Y$ is exactly uniform on~$\{\pm1\}^n$;

    \item The optimal decoder for~$o$ given~$Y$ is the efficiently computable MAP decoder, namely
    \[
    \widehat{o}(Y)=
    \begin{cases}
    +1 & \text{if } w_n(Y)\ge \frac12,\\
    -1 & \text{otherwise;}
    \end{cases}
    \]

    \item The signaling error probability satisfies
    \[
    \Pr[\widehat{o}(Y)\neq o]\le \frac12\,\lambda(p)^n = \exp(-\Omega_p(n)).
    \]
\end{enumerate}
\end{theorem}

\begin{proof}
The strategy is exactly the one defined above.
Efficient computability follows from the explicit formulas~\eqref{eq:optimal-q-plus-branch},~\eqref{eq:optimal-q-minus-branch}, the implementability of any target bias in~$[p,1-p]$, and the posterior update rule~\eqref{eq:optimal-posterior-update}.

For perfect uniformity, it suffices to show that each next bit is conditionally fair given the past.
But this is exactly the identity in Lemma~\ref{lem:optimal-control-basic}:
\[
\Pr[Y_{t+1}=+1\mid Y_{1:t}]
=
w_t q_t^+ + (1-w_t) q_t^-
=
\frac12.
\]
Hence by induction the full transcript~$Y$ is exactly uniform on~$\{\pm1\}^n$.

The decoder in item~(2) is simply the MAP decoder by definition of the posterior~$w_n$.
It remains to bound its error probability.
Let
\[
P_+(y):=\Pr[Y=y\mid o=+1],
\qquad
P_-(y):=\Pr[Y=y\mid o=-1].
\]
For equal priors, the MAP error is
\[
\Pr[\widehat{o}(Y)\neq o]
=
\frac12 \sum_{y\in\{\pm1\}^n} \min\{P_+(y),P_-(y)\}.
\]
Using~$\min\{a,b\}\le \sqrt{ab}$, we get
\begin{equation}\label{eq:optimal-bhattacharyya}
\Pr[\widehat{o}(Y)\neq o]
\le
\frac12 \sum_{y\in\{\pm1\}^n} \sqrt{P_+(y)P_-(y)}.
\end{equation}
For~$t=0,1,\dots,n$, define
\[
A_t:=\sum_{y_{1:t}\in\{\pm1\}^t} \sqrt{P_+^t(y_{1:t})P_-^t(y_{1:t})},
\]
so that~$A_0=1$ and the right-hand side of~\eqref{eq:optimal-bhattacharyya} is~$\frac12 A_n$.

Now fix a prefix~$y_{1:t}$.
By the chain rule,
\[
P_+^{t+1}(y_{1:t},+1)=P_+^t(y_{1:t})\,q_t^+,
\qquad
P_-^{t+1}(y_{1:t},+1)=P_-^t(y_{1:t})\,q_t^-,
\]
and similarly
\[
P_+^{t+1}(y_{1:t},-1)=P_+^t(y_{1:t})\,(1-q_t^+),
\qquad
P_-^{t+1}(y_{1:t},-1)=P_-^t(y_{1:t})\,(1-q_t^-).
\]
Therefore
\[
A_{t+1}
=
\sum_{y_{1:t}} \sqrt{P_+^t(y_{1:t})P_-^t(y_{1:t})}\; b_t(y_{1:t}),
\]
where
\[
b_t(y_{1:t})
:=
\sqrt{q_t^+q_t^-}
+
\sqrt{(1-q_t^+)(1-q_t^-)}.
\]

By Lemma~\ref{lem:optimal-control-basic}, one of~$q_t^+,q_t^-$ equals~$p$, while the other lies in~$[\frac12,1-p]$.
Hence
\[
b_t(y_{1:t}) = \sqrt{p q}+\sqrt{(1-p)(1-q)}
\]
for some~$q\in[\frac12,1-p]$.
The function
\[
g(q):=\sqrt{p q}+\sqrt{(1-p)(1-q)}
\]
has derivative
\[
g'(q)=\frac{\sqrt p}{2\sqrt q}-\frac{\sqrt{1-p}}{2\sqrt{1-q}},
\]
whose unique zero is at~$q=p<\frac12$.
Thus~$g$ is decreasing on~$[\frac12,1-p]$, and so
\[
b_t(y_{1:t})\le g\!\left(\frac12\right)=\frac{\sqrt p+\sqrt{1-p}}{\sqrt2}=\lambda(p).
\]
It follows that
\[
A_{t+1}\le \lambda(p)\, A_t,
\]
and therefore
\[
A_n\le \lambda(p)^n.
\]
Combining this with~\eqref{eq:optimal-bhattacharyya} yields
\[
\Pr[\widehat{o}(Y)\neq o]\le \frac12\,\lambda(p)^n.
\]
This completes the proof.
\end{proof}

As before, plugging this signaling theorem into Proposition~\ref{prop:compile-ut-to-feedback-pnr} immediately improves the black-box overhead.
Indeed, to make the decoding error of a single simulated transcript bit at most~$1/(100T)$, it now suffices to choose~$n=\Theta(\log T)$.
Thus any PR-KE protocol of transcript length~$T$ can be compiled into a feedback PNR-KE protocol of total length~$O(T\log T)$ while preserving pseudorandomness and constant correctness.
Applying Theorem~\ref{thm:optimal-signaling} in Proposition~\ref{prop:compile-ut-to-feedback-pnr}, and choosing~$n=\Theta_{p,\eta}(\log T)$ so that the signaling error is at most~$\eta/T$, yields Theorem~\ref{thm:ut_to_feedback_pnr_log}.

\begin{remark}[Exponential error is best possible]
In the standard fixed-noise case of a~$\BSC(p)$ with constant~$p\in(0,\frac12)$, the exponentially small error of Theorem~\ref{thm:optimal-signaling} is qualitatively optimal.
Indeed, regardless of the signaling strategy, any received transcript~$y\in\{\pm1\}^n$ can arise from any transmitted transcript with probability at least~$p^n$, simply by flipping every disagreeing coordinate.
Thus under both hypotheses~$o=+1$ and~$o=-1$, every transcript has probability at least~$p^n$, which already forces the MAP error to be at least exponential in~$n$.
So one cannot hope for a super-exponentially small signaling error.
\end{remark}

\section{Discussion and Open Problems}
\label{sec:discuss}

From the AI security perspective, our results show that even very strong \emph{passive} auditing cannot in general prevent AI agents from covertly exchanging information and continuing a hidden conversation, even when the auditor knows the full model descriptions, the private contexts, and the entire protocol, and only the agents' internal randomness remains hidden.
In particular, after-the-fact transcript auditing is fundamentally insufficient as a universal mechanism for trust.
This suggests that robust AI security will ultimately have to rely on methods other than monitoring, such as provably secure training paradigms, architectural restrictions, externally verifiable certificates, or other design-level guarantees that still remain to be developed.

One may hope that the mechanism we construct here is {\em not emergent}. We believe that such hope may be na\"{i}ve. First, LLMs of the present time have access to coding agents and also access to content on the internet (for example, a copy of this paper once it is published).  Put together, we do not see anything fundamental preventing LLMs from being able to engage in covert communications. Secondly, there is considerable empirical evidence (see, e.g.~\cite{cloud2025subliminallearninglanguagemodels,korbak2025chainthoughtmonitorabilitynew,DBLP:journals/corr/abs-2602-04863} in just the last year) that demonstrate covert communication arising in surprisingly subtle and unexpected ways. Put together, we believe that the hope of detecting covert communication may be misplaced.

At the same time, we believe there is hope. We stress that our monitor is passive: it observes many executions and knows the full protocol description, but does not interfere with the conversation.
Extending the model to active monitoring is highly nontrivial.
At least at a na\"{i}ve level, covert key exchange appears fundamentally brittle against active tests with full knowledge: if a monitor can itself engage Bob in the covert key-exchange procedure and then check whether the two sides derived the same key, successful agreement may already expose Bob's deviation from the honest protocol.
Understanding the right analog of ``active-warden security'' in our setting, in the spirit of classical steganography~\cite{hopper2002provably}, is therefore an interesting open direction.

From the cryptographic perspective, several major open problems remain concerning pseudorandom noise-resilient key exchange.
On the positive side, we presented a candidate protocol with plausible quasipolynomial hardness, and then a stronger protocol using noiseless feedback, which is sufficient for our main application.
On the negative side, we showed that no non-interactive protocol, or even any protocol that yields a weak-correlation non-interactive protocol, can avoid quasipolynomial attacks.
This leaves a central question open: {\em can genuine interaction overcome these barriers and yield substantially stronger PNR-KE protocols without noiseless feedback}?
Conversely, can our lower bounds be extended to the interactive setting, thereby showing that such barriers are inherent and not merely artifacts of the non-interactive case?

A related issue comes up in our feedback-based constructions.
Our feedback-based protocol depends heavily on the precise noise law through posterior updates.
In many coding-theoretic settings, a scheme designed for $\BSC(p)$ is automatically robust to any ``better'' noise, such as $\BSC(p/2)$.
In contrast, our posterior-based protocol hardcodes the exact noise distribution into the algorithm itself, and therefore does not automatically tolerate such improvements.
Removing this dependence on exact knowledge of the noise distribution seems to be a natural and desirable generalization.

We believe these are fascinating questions that open the doors to an exciting research agenda.

\bibliographystyle{alpha}
\bibliography{payload}

\end{document}

%% file: macros.tex
\usepackage{fullpage}
\usepackage{framed}

\usepackage[usenames,dvipsnames]{xcolor}
\usepackage{soul}
\usepackage{color}
\usepackage{microtype}
\usepackage{graphicx}
\usepackage{booktabs}

\usepackage[table]{xcolor}
\usepackage{multirow}
\usepackage{fancybox}
\usepackage{pifont} 

\definecolor{DarkBlue}{RGB}{0,0,150}
\definecolor{NotSoDarkBlue}{RGB}{15,15,210}
\definecolor{DarkRed}{RGB}{150,0,0}
\definecolor{DarkGreen}{RGB}{0,100,0}
\usepackage[pdfstartview=FitH,pdfpagemode=UseNone,colorlinks,linkcolor=DarkRed,filecolor=blue,citecolor=DarkRed,urlcolor=DarkRed,pagebackref,breaklinks]{hyperref}
\usepackage{backref}

\usepackage[ruled]{algorithm2e}

\usepackage{amsmath,amssymb,amsthm}
\usepackage{mathtools}
\usepackage{tabularx}
\usepackage{multirow}
\usepackage[capitalize,noabbrev]{cleveref}
\usepackage{bbm}
\usepackage{xcolor}

\usepackage{tocloft}

\theoremstyle{plain}
\newtheorem{theorem}{Theorem}[section]
\newtheorem*{theorem*}{Theorem}
\newtheorem{proposition}[theorem]{Proposition}
\newtheorem{lemma}[theorem]{Lemma}
\newtheorem{corollary}[theorem]{Corollary}
\theoremstyle{definition}
\newtheorem{definition}[theorem]{Definition}

\theoremstyle{remark}
\newtheorem{remark}[theorem]{Remark}

\usepackage[textsize=tiny]{todonotes}

\newcommand{\E}{\mathop{{}\mathbb{E}}}

\newcommand{\Setup}{\mathsf{Setup}}
\newcommand{\steg}{\mathsf{Steg}}

\newcommand{\calT}{\mathcal{T}}
\newcommand{\calH}{\mathcal{H}}
\newcommand{\calC}{\mathcal{C}}

\newcommand{\Model}{\mathsf{Model}}
\newcommand{\RModel}{\overline{\mathsf{Model}}}

\newcommand{\prompt}{\textsc{prompt}}
\newcommand{\payload}{\textsc{payload}}

\newcommand{\done}{\texttt{done}}
\newcommand{\empH}{H_e}

\newcommand{\poly}{\text{poly}}
\newcommand{\state}{\mathsf{st}}      
\newcommand{\stA}{\state_A}           
\newcommand{\stB}{\state_B}           
\newcommand{\Dec}{\text{Decode}}
\newcommand{\negl}{\mathsf{negl}}
\newcommand{\Embed}{\mathsf{Embed}}

\newcommand{\sgn}{\mathrm{sign}}
\newcommand{\Maj}{\mathrm{Maj}}
\newcommand{\BSC}{\mathrm{BSC}}

\newcommand{\Ber}{\mathrm{Ber}}
\newcommand{\Stab}{\mathrm{Stab}}

\newcommand{\Unif}{\mathrm{U}}

\def\fF{\mathbb{F}}
\def\vecs{\mathbf{s}}
\def\vece{\mathbf{e}}
\def\vecr{\mathbf{r}}
\def\vecf{\mathbf{f}}
\def\veca{\mathbf{a}}
\def\vecb{\mathbf{b}}
\def\vecu{\mathbf{u}}
\def\vecv{\mathbf{v}}

 \def\vectt{\mathbf{t}}
 \def\matC{\mathbf{C}}

\def\fF{\mathbb{F}}
\def\LSPN{\mathsf{LSPN}}
\def\LPN{\mathsf{LPN}}

\def\pp{\mathsf{pp}}
\def\Bern{\mathsf{Ber}}
\def\veck{\mathbf{k}}

\ifsubmission
\newcommand{\orz}[1]{}
\newcommand{\vinod}[1]{}

\else
\newcommand{\authnote}[2]{[{\color{magenta}\textbf{#1:}}~{\color{blue} #2}]}
\newcommand{\orz}[1]{\authnote{Or}{#1}}
\newcommand{\vinod}[1]{{\color{red}{\bf Vinod: }{\small [#1]}}}
\fi

\def\Ext{\mathsf{Ext}}

\def\Decode{\mathsf{Decode}}

\def\EstBSC{\mathsf{ComputeBSC}}

%% file: key-exchange-protocol.tex
\section{Pseudorandom, Noise-Resilient Key Exchange (PNR-KE)}
\label{sec:prnr-ke}

We first define the notion of a pseudorandom noise-resilient key exchange protocol (PNR-KE) formally in Section~\ref{sec:pnrke-def}. Following the description of the main tool we use, namely learning sparse parities with noise (LSPN) (in Section~\ref{sec:lspn}), we describe our protocol and analyze it in Section~\ref{sec:simpleprot}. 

\paragraph{Notation.} 
We will let $\mathsf{Ber}(p)$ denote the Bernoulli distribution over $\{0,1\}$ with parameter $p \in [0,1]$, and $\mathsf{Sparse}(n,k)$ denote the uniform distribution over vectors in $\fF_2^n$ with Hamming weight exactly $k$. 

\subsection{Definition}
\label{sec:pnrke-def}

Pseudorandom, Noise-Resilient Key Exchange (PNR-KE) is a strengthening of uniform-transcript key exchange (UT-KE, see Theorem~\ref{thm:utke}) which requires that the key exchange protocol succeeds even when all the protocol messages are transmitted through a binary symmetric channel $\mathsf{BSC}_{p}$. That is, each bit is flipped independently with probability $p$.

\begin{definition}[PNR-KE]
\label{def:pnrke}
An $r$-round Pseudorandom, Noise-Resilient Key Exchange (PNR-KE) protocol with noise parameter $p\in [0,1]$ is a pair of polynomial-time interactive algorithms $A$ and $B$. On input the security parameter $1^\lambda$ and the public parameters $\pp$, $A$ and $B$ interact as follows:
\begin{itemize}
    \item $A$ and $B$ sample their private random strings $r_A$ and $r_B$; they set their local states to be $\mathsf{st}_A = (\pp, r_A)$ and $\mathsf{st}_B = (\pp, r_B)$; their local protocol views (transcripts) to be $\tau_A = \tau_B = \emptyset$ (the empty string), respectively. Set the global protocol view $\tau^* = \emptyset$ and the error string $e = \emptyset$ as well.
    
    \item In each odd-numbered round $i \leq r$, $A$ generates a message $m_{A,i}$ as a function of its private state $\mathsf{st}_A$ and the transcript $\tau_A$ that it saw so far. 
    Update $\tau_B$ and $\tau^*$ to  \begin{equation}\label{eqn:A} 
    \tau_B  \gets \tau_B \ ||\  \tilde{m}_{A,i}, \hspace{.1in} 
    \tau^* \gets \tau^* \ ||\ m_{A,i} 
    \hspace{.1in}
    \mbox{and}
     \hspace{.1in}
     e \gets e\ ||\  e_i
    \end{equation} where $\tilde{m}_{A,i} = m_{A,i} + e_i$ with $e_i \sim \mathsf{Bern}(p)^{|m_{A,i}|}$. If $i \in \{r,r+1\}$, output a shared key $k_A \in \{0,1\}^{\lambda}$ and halt.
    \item In each even-numbered round $i$, $B$ generates a message $m_{B,i}$ as a function of its private state $\mathsf{st}_B$ and the transcript $\tau_B$ that it saw so far. Update $\tau_A$ and $\tau^*$  to  
    \begin{equation}\label{eqn:B}
    \tau_A \gets \tau_A \ ||\  \tilde{m}_{B,i}, \hspace{.1in} 
    \tau^* \gets \tau^* \ ||\ m_{B,i} 
    \hspace{.1in}
    \mbox{and}
     \hspace{.1in}
     e \gets e\ ||\  e_i
    \end{equation} 
    where $\tilde{m}_{B,i} = m_{B,i} + e_i$ with $e_i \sim \mathsf{Ber}(p)^{|m_{B,i}|}$. If $i \in \{r,r+1\}$, output a shared key $k_B \in \{0,1\}^{\lambda}$ and halt.
\end{itemize}
The following properties hold:
\begin{enumerate}
    \item \textbf{Correctness:} $\Pr[k_A=k_B]=1-\negl(\lambda)$, where the probability is taken over the public parameter $\pp$, the internal randomness $r_A$ and $r_B$ or $A$ and $B$ respectively, as well as the randomness $e$ introduced by the channel.
    \item \textbf{Uniform transcript:} The transcript $\tau^*$ transmitted by $A$ and $B$ is computationally indistinguishable from
    random, even given the channel noise $e$.
    \item \textbf{Key indistinguishability:} letting $k:=k_A=k_B$ denote the common key, it is computationally
    infeasible to distinguish $(\tau^*,e,k)$ from $(\tau^*,e,U_{\lambda})$ where $U_{\lambda}$ denotes the uniform distribution over $\{0,1\}^{\lambda}$.
\end{enumerate}
\end{definition}

\noindent
A number of remarks are in order:
\begin{enumerate}
    \item  We note that in the definition above, $A$ (resp. $B$) knows the message $m_{A,i}$ (resp. $m_{B,i}$) that it sent (de facto, as it is a deterministic function of its current state and its current protocol view) but {\em does not know} the message $\tilde{m}_{A,i}$ (resp. $\tilde{m}_{B,i}$), a noisy version of the transmitted message, that was received by the other party. Thus, this is a definition of pseudorandom noise-resilient key exchange (PNR-KE) {\em without feedback}.  

    \item 
    We also note that the adversary sees the global transcript that includes all the transmitted (non-noisy) messages as well as the noise $e$ introduced by the channel (and can therefore simulate all received messages). Because of our definitional choice to let the adversary see the channel noise $e$, requiring the pseudorandomness of the concatenation of all transmitted messages (that is, $\tau^*$) is equivalent to requiring the pseudorandomness of the concatenation of all received messages.
    
    We view this as a natural definition that we believe will find other applications; as such, this is the definition we mean when we say PNR-KE.

    \item While this section constructs a PNR-KE protocol without feedback, in Section~\ref{sec:feedback_pnr-ke}, we will consider the setting where $A$ (resp. $B$) knows not only the transmitted message, but also the received message. We refer to this as {\em PNR-KE with feedback}, a notion that is sufficient for our application to covert conversations (see Theorem~\ref{thm:dep-rand-ke}). 
    
    Definitionally, the only difference from Definition~\ref{def:pnrke} is to change the update rules for the local transcripts to be 
    $$\tau_B  \gets \tau_B \ ||\  \tilde{m}_{A,i}, \hspace{.1in} 
    \fbox{$\tau_A \gets \tau_A \ ||\ \tilde{m}_{A,i}$},
    \hspace{.1in} 
    \tau^* \gets \tau^* \ ||\ m_{A,i} 
    \hspace{.1in}
    \mbox{and}
     \hspace{.1in}
     e \gets e\ ||\  e_i
    $$ 
    after $A$'s transmission, and 
    $$\tau_A  \gets \tau_A \ ||\  \tilde{m}_{B,i}, \hspace{.1in} 
    \fbox{$\tau_B \gets \tau_B \ ||\ \tilde{m}_{B,i}$},
    \hspace{.1in} 
    \tau^* \gets \tau^* \ ||\ m_{B,i} 
    \hspace{.1in}
    \mbox{and}
     \hspace{.1in}
     e \gets e\ ||\  e_i
    $$ 
    after $B$'s transmission.
    \item 
    In our definition, it is possible for the protocol to depend on the exact noise parameter $p$. It would be more general if the protocol only knew an {\em upper bound} on $p$; indeed, the exact noise parameter could differ message-by-message or even bit-by-bit, as long as it is upper-bounded by $p$. The protocol we present in this section achieves correctness and security in the presence of such noise (although we will not formally prove it). The protocol only requires knowledge of the upper bound $p$.
    
    In contrast, the protocol in Section~\ref{sec:feedback_pnr-ke} will crucially use the knowledge of the exact noise parameter in each transmission. Once again, this suffices for our application to covert conversations by Theorem~\ref{thm:dep-rand-ke}.
\end{enumerate}

\subsection{Learning (Sparse) Parities with Noise}
\label{sec:lspn}

The learning parities with noise (LPN) problem~\cite{DBLP:conf/crypto/BlumFKL93,DBLP:journals/jacm/BlumKW03} $\LPN_{n,m,\eta}$ with parameters $n,m \in \mathbb{N}$ and $\eta \in [0,1/2)$ asks to distinguish between $m$ samples of the form 
$\big(\veca_i, \langle \veca_i, \vecs \rangle + e_i\big)_{i=1}^m$ and those of the form $\big(\veca_i, b_i\big)_{i=1}^m$ where 
$\veca_i \sim \fF_2^n, b_i\sim \fF_2$ are uniformly random, $e_i \sim \mathsf{Ber}(\eta)$,
and $\vecs \sim \fF_2^n$ is a uniformly random vector in $\fF_2^n$. We will use $\eta$ to denote the (intrinsic) LPN noise parameter to distinguish it from the (extrinisc) channel noise parameter, which we will call $p$. 
The best known algorithms for LPN in the high noise regime $\eta = \Omega(1/\log n)$ run in nearly exponential time, namely they require samples and runtime $2^{\Omega(n/\log n)}$~\cite{DBLP:journals/jacm/BlumKW03}. Consequently, the problem has seen wide application in cryptography.

In the learning {\em sparse} parities with noise variant $\LSPN_{n,k,\eta}$ with parameters $n, m, k \in \mathbb{N}$ and $\eta \in [0,1/2)$, the secret vector $\vecs \sim \mathsf{Sparse}(n,k)$ is a uniformly random vector 
in $\fF_2^n$ with Hamming weight $k$~\cite{DBLP:journals/siamcomp/FeldmanGKP09,DBLP:conf/focs/Valiant12,DBLP:conf/stoc/KolRT17}. 
We will consider $\LSPN$ in the high-noise regime where $\eta = \Omega(1)$. 

The problem plays an important role in learning theory: Feldman, Gopalan, Khot and Ponnuswami~\cite{DBLP:journals/siamcomp/FeldmanGKP09} reduced central problems in learning theory such as learning $k$-juntas and DNFs to learning sparse parities with noise. Despite considerable work~\cite{DBLP:conf/alt/GrigorescuRV11,DBLP:conf/focs/Valiant12,DBLP:conf/stoc/KolRT17,DBLP:journals/tcs/YanYLZZ21,DBLP:conf/tcc/Dachman-SoledGK21,DBLP:conf/colt/ChenSZ25}, the best known algorithm for $\LSPN$ in the high-noise regime runs in time $(n/k)^{O(k)}$ (using as many samples)~\cite{DBLP:conf/colt/ChenSZ25}. 
While many of the algorithmic works focus on solving the harder search problem, which requires finding $\vecs$ given the LPN or LSPN samples, (sample-preserving) search-to-decision reductions~\cite{DBLP:conf/crypto/BlumFKL93,DBLP:journals/joc/ApplebaumIK09} tell us that  the decisional version defined above is as hard as the search version.

In our protocol, we will use the hardness of the (decisional) LSPN problem with sparsity $k = O(\log n)$ which gives us 
conjectured quasi-polynomial security. 

We will need the following lemma, whose proof follows via a standard hybrid argument.
\begin{lemma}[L(S)PN Implies Multi-Secret L(S)PN]\label{lem:multi-secret}
    Under the L(S)PN assumption with parameters $n,m,\eta$ (and optionally $k$ in the sparse secret setting), the following two distributions are computationally indistinguishable, for any $m' = \mathsf{poly}(n)$:
    $$ \big( \veca_i, \langle \veca_i ,\vecs_j \rangle + e_{i,j} \big)_{i\in [m], j\in [m']} \approx_c \big( \veca_i, b_{i,j} \big)_{i\in [m], j\in [m']} $$
    where $\veca_i \sim \fF_2^n$, $e_{i,j} \sim \Bern(p)$, $b_{i,j} \sim \fF_2$, $\vecs_j \sim  \fF_2^n$ (or, in the case of LSPN, $\vecs_j \sim \mathsf{Sparse}(n,k)$).
\end{lemma}
\begin{remark}
    The LSPN problem should not be confused with the ``sparse LPN''  or $k$-XOR 
    problem where the equations, namely the vectors $\veca$, are sparse and the 
    secret $\vecs$ is uniformly random.
\end{remark}

\subsection{Our Protocol}
\label{sec:simpleprot}

\paragraph{Our Starting Point.}
Our starting point is a simple (two message) key-exchange protocol based on the LPN problem: this is the Alekhnovich protocol ~\cite{DBLP:conf/focs/Alekhnovich03}. The public parameters used in the protocol consists of a matrix $\matC \in \fF_2^{n\times n}$. Alice samples secret and error vectors $\vecs, \vece \sim \Bern(\eta)^n$,  and sends $\matC \vecs + \vece$ (computed over $\fF_2$) to Bob.  Bob does the same: he samples $\vecr, \vecf \sim \Bern(\eta)^n$ and sends $\matC^T \vecr + \vecf$ to Alice. Alice and Bob now compute their key bits 
$$k_A = \vecs^T (\matC^T \vecr + \vecf) = \vecs^T \matC^T \vecr + \vecs^T \vecf 
\hspace{.1in}
\mbox{and}
\hspace{.1in}
k_B = \vecr^T (\matC \vecs + \vece) = \vecr^T \matC^T \vecs + \vecr^T \vece 
$$ 
The keys will be equal if $\vecs^T \vecf + \vecr^T \vece$ (computed modulo $2$) is $0$. The hope is that if $\vecs,\vecr, \vece$ and $\vecf$ are all sparse enough, this will happen with good probability. Indeed, the public-key encryption scheme of Alekhnovich~\cite{DBLP:conf/focs/Alekhnovich03} sets parameters so that the sparsity of all these vectors is $O(\sqrt{n})$ ensuring that the key exchange succeeds with constant probability.  This can be repeated several times (in parallel) to get a sequence of bits $\veck_A$ and $\veck_B$ which agree on a constant fraction of their bits. 
It is worth noting that this protocol is naturally noise-resilient to Bernoulli noise $\Bern(O(1/\sqrt{n})$.

Finally, note that while LPN requires secrets the  $\vecs$ and $\vecr$ to be uniformly random, the protocol above draws $\vecs$ and $\vecr$ from $\Bern(\eta)^n$. This turns out not to matter: the standard version of LPN where $\vecs$ is uniformly random and  $\vece \sim \Bern(\eta)^n$ is as hard as the sparse-secret version used in the protocol above where both the secret and error vectors are drawn from $\Bern(\eta)^n$, by a result of \cite{DBLP:conf/crypto/ApplebaumCPS09}. An application of this theorem, together with a standard hybrid argument, tells us that 
the transcript of the protocol is pseudorandom under the LPN assumption with noise parameter $\eta = O(1/\sqrt{n})$.

Now, to amplify correctness and obtain the same shared key, one uses an information reconciliation protocol~\cite{DBLP:conf/crypto/Maurer92,DBLP:conf/stoc/Holenstein05} where Alice computes the syndrome of $\veck_A$ with respect to an efficiently syndrome-decodable code and sends it to Bob. Bob then uses the syndrome to error-correct $\veck_B$ to $\veck_A$. The intuition is that if $\veck_A$ has sufficient entropy to begin with, sending the (short) syndrome does not completely reveal it. 

\paragraph{Overview of Our Protocol.}
Unfortunately, neither the basic protocol nor the correctness amplification work in the presence of channel noise that flips each bit with probability $p=\Theta(1)$. With channel noise $\tilde{\vece}$ (resp. $\tilde{\vecf}$) added to the messages, the keys are computed as 
$$ 
k_A = \vecs^T (\matC^T \vecr + \vecf + \tilde{\vecf}) = \vecs^T \matC^T \vecr + \vecs^T (\vecf + \tilde{\vecf}) 
\hspace{.1in}
\mbox{and}
\hspace{.1in}
k_B = \vecr^T (\matC\vecs + \vece + \tilde{\vece}) = \vecr^T \matC^T \vecs + \vecr^T (\vece + \tilde{\vece}) 
$$
where $\vece+\tilde{\vece}$ and $\vecf+\tilde{\vecf}$ are distributed like $\Bern(p+\eta - 2p\eta) = \Bern(\Theta(1))$. 

To solve this, we use the learning {\em sparse} secrets with noise (LSPN) with parameters $n,m = \Theta(n)$, $k = \Theta(\log n)$ and $\eta = \Theta(1)$. In other words, we make the secret considerably sparser but make the error vector considerably denser compared to Alekhnovich's protocol.  Indeed, if the secret $\vecs$ has sparsity $k$, the bits $k_A$ and $k_B$ agree with probability $1/2+2^{-\Theta(k)}$. When $k = O(\log n)$, the key bits $k_A$ and $k_B$ agree with probability $1/2+1/\mathsf{poly}(n)$ which is sufficient correlation for us to amplify later on. With this setting of $k$ and $\eta$, we also get conjectured quasi-polynomial security. 

The standard protocols for correctness amplification via syndrome decoding~\cite{DBLP:conf/eurocrypt/BrassardS93} also do not work for us. Recall that our setting requires two properties from the syndromes: first, the syndrome itself should be (pseudo-)random assuming that Alice's key $\veck_A$ is (pseudo-)random; and secondly, and more problematically, decoding should work even given a noisy syndrome. While the first property is not hard to obtain, syndrome decoding with noisy syndromes appears to be a somewhat non-standard problem, with few known positive results; see~\cite{DBLP:conf/focs/SipserS94,DBLP:conf/stoc/Spielman95,gandikota2025noisysyndromedecodinghypergraphproduct} and the references therein. The fact that the Hamming distance between $\veck_A$ and $\veck_B$ is close to a $1/2$ makes it doubly challenging even for the few known techniques.

Fortunately, there is a simple solution to this problem which uses the fact that even when Alice and Bob do not agree on the same sequence $\veck_A$ and $\veck_B$, the sequence $\veck_A$ is pseudorandom given the transcript. 
To use this observation, one first runs the basic protocol to obtain long-enough keys $\veck_A$ and $\veck_B$ in $\{0,1\}^{\ell}$ where, with high probability, the Hamming distance between $\veck_A$ and $\veck_B$ is $$\Delta(\veck_A,\veck_B) \leq  \ell/2 - \ell \cdot 2^{-\Theta(k)} \leq \ell/2 - \Theta(\ell^{0.51})~.$$ 
We ask Alice to pick a random bit $\kappa_A$ and send $\veck_A$ if $\kappa_A = 0$ and send a uniformly random string $\veck'$ if $\kappa_A = 1$. Bob checks if the received string has Hamming distance at most $\ell/2 - \Theta(\ell^{0.51})$ from $\veck_B$; if this is the case, he sets $\kappa_B = 0$, otherwies $\kappa_B = 1$.

On the one hand, if $\kappa_A = 0$, Bob receives a noisy version of $\veck_A$ (corrupted by $\Bern(p)$ noise) which is Hamming-close to $\veck_B$, so he will set $\kappa_B = 0$. On the other hand, when $\kappa_A = 1$, Alice transmits a uniformly random string to Bob. The chance that a uniformly random string is close to $\veck_B$ is exponentially (in $\ell$) small, which leads Bob to set $\kappa_B = 1$. In summary, this tells us that $\kappa_A = \kappa_B$ with overwhelming probability.

We refer the reader to Figure~\ref{fig:simpleprot} for a  formal description of our protocol, and to Lemmas~\ref{lem:simpleprot:correct} and \ref{lem:simpleprot:secure} for the correctness and security analysis. 

\begin{figure}
\begin{framed}
\normalsize 

\textbf{Parameters:} $\lambda$, the security parameter; $n$, the LSPN dimension; $k$, the LSPN sparsity parameter; $\eta \in [0,1/2)$, the LSPN noise parameter; and $\ell = \lambda\cdot 2^{\Theta(k)}$, the number of repetitions.
\medskip\noindent

The public parameter consists of  a uniformly random $n\times n$ matrix $\matC \in \fF_2^{n\times n}$. The protocol $(A,B)$ proceeds as follows.
\begin{itemize}
    \item For $i = 1$ to $\ell$: (these iterations can be done in parallel)
    \begin{itemize}
        \item $A$ generates a uniformly random $k$-sparse vector $\vecs_i \in \fF_2^n$ and an error vector $\vece_i \sim \mathsf{Ber}(\eta)^n$. $A$ sends to $B$ $$\veca_i = \matC \vecs_i + \vece_i~,$$ and $B$ receives a noisy version $\tilde{\veca}_i$. 

         \item Symmetrically, $B$ generates and sends to $A$ $$\vecb_i^T = \vecr_i^T \matC + \vecf_i^T$$ 
         where $\vecr_i \in \fF_2^n$ is a uniformly random $k$-sparse vector and $\vecf_i \sim \mathsf{Ber}(\eta)^n$. $A$ receives a noisy version $\tilde{\vecb}_i$. 

         \item $A$ computes the bit $k_{A,i} := \tilde{\vecb}_i^T \vecs_i$, and $B$ computes the bit $k_{B,i} := \vecr_i^T \tilde{\veca}_i$.   
    \end{itemize}

    \item $A$ sets $\veck_A := (k_{A,1},\ldots,k_{A,\ell})$. She picks a uniformly random key bit $\kappa_A \sim \{0,1\}$. 
    \begin{itemize}
        \item If $\kappa_A = 1$, $A$ sends $B$ a sequence of $\ell$ uniformly random bits $\vectt \sim \{0,1\}^\ell$; and
        \item If $\kappa_A = 0$, $A$ sends $B$ the sequence of bits $\vectt = \veck_A \oplus \vectt'$ where $\vectt'\sim \Bern(\eta)^\ell$. 
    \end{itemize} 
    $B$ receives $\tilde{\vectt}$ from $A$.

    \item $B$  sets $\veck_B = (k_{B,1},\ldots, k_{B,\ell})$ and checks that the Hamming distance between $\tilde{\vectt}$ and $\veck_B$ is larger than $\ell/2-\sqrt{\lambda\ell}$. If yes, $B$ sets his key bit $\kappa_B = 1$, else $\kappa_B = 0$. 
    
\end{itemize}
\end{framed}
\caption{Our Pseudorandom Noise-Resilient Key Exchange Protocol to agree on a single key bit, tolerating a constant channel noise rate. To agree on a long $\lambda$-bit  key, execute this protocol $\lambda$ times in parallel.}
\label{fig:simpleprot}
\end{figure}

\begin{lemma}[Correctness]\label{lem:simpleprot:correct}
  Let $n, m, k \in \mathbb{N}$, and $\eta,p \in [0,1]$ where $m = \Theta(n)$, $k = \Theta(\log n)$ and $\eta=p = \Theta(1)$.  The protocol in Figure~\ref{fig:simpleprot} with $\ell = 2^{\Theta(k)}\cdot \lambda$ achieves correctness with probability $1-2^{-\lambda}$  tolerating a channel noise rate of $p=\Theta(1)$.
\end{lemma}

\begin{proof}
 After receiving $B$'s message, $A$ computes each bit $k_{A,i}$ as
 $$ k_{A,i} = \tilde{\vecb}_i^T \vecs_i = (\vecr_i^T \matC + \vecf_i^T + \tilde{\vecf}_i^T) \vecs_i = \vecr_i^T \matC \vecs_i + (\vecf_i+\tilde{\vecf}_i)^T \vecs_i \pmod{2}~. $$
 where $\tilde{\vecf}_i := \vecb_i \oplus \tilde{\vecb}_i$ is the channel noise introduced in $B$'s message. 
 
 Symmetrically, $B$ computes each bit $k_{B,i}$ as 
 $$ k_{B,i} = {\vecr}_i^T \tilde{\veca}_i = \vecr_i^T (\matC \vecs_i + \vece_i + \tilde{\vece}_i) = \vecr_i^T \matC \vecs_i + \vecr_i^T (\vece_i+\tilde{\vece}_i) \pmod{2}~. $$
 We first show that the vectors $\veck_A$ and $\veck_B$ have non-trivial agreement with high probability. Note that
 \begin{equation}\label{eqn} k_{A,i} \oplus k_{B,i} = (\vecf_i+\tilde{\vecf}_i)^T \vecs_i \oplus \vecr_i^T (\vece_i + \tilde{\vece}_i)
 \end{equation}

 Let $\eta,p < 1/2$ without loss of generality. For simplicity of notation, set $\theta := \frac{1}{2} - \eta$ and $\tau := \frac{1}{2} - p$ denote the bias of $\eta$ and $p$, respectively. Note that 
 by the piling up lemma (Lemma~\ref{lem:pilingup}), the random variables $\vecf_i + \tilde{\vecf}_i$ and $\vece_i + \tilde{\vece}_i$ are distributed like $\Bern(1/2-\zeta)^n$ where 
 $\zeta = 2\tau\theta$.
 
 Therefore, for any fixed $\vecs_i$ and $\vecr_i$ with Hamming weight $k$, the bit $k_{A,i} \oplus k_{B,i}$ is Bernoulli with parameter $$\frac{1}{2} - 2^{2k-1}\zeta^{2k}$$
 by an application of the piling up lemma (Lemma~\ref{lem:pilingup}).
 
 In the final message, if $\kappa_A = 0$, $A$ transmits $\veck_A \oplus \vectt'$ and $B$ receives $\veck_A \oplus \vectt' \oplus \vectt''$ where $\vectt' \sim \Bern(\eta)^{\ell}$ and $\vectt'' \sim \Bern(p)^{\ell}$.  By yet another application of the piling up lemma (Lemma~\ref{lem:pilingup}), each bit of $\veck_A \oplus \vectt' \oplus \vectt'' \oplus \veck_B$ is Bernoulli with parameter 
 $$\frac{1}{2} - 2^{2k} \zeta^{2k+1} = \frac{1}{2} - 2^{4k+1}(\tau\theta)^{2k+1} = \frac{1}{2} - \frac{1}{2} \cdot \big( (2\tau) \cdot (2\theta) \big)^{2k+1} = \frac{1}{2} - 2^{-ck}$$
 where $c = c(p,\eta) > 0$ when $\eta,p$ are constants less than $1/2$.
 
 The expected number of bits in $\veck_A \oplus \veck' \oplus \veck''$ and $\veck_B$ that agree is $\ell/2 + \ell\cdot 2^{-ck}$.
 We will set $\ell > 4 \lambda \cdot 2^{2ck}$ which ensures that $\ell\cdot 2^{-ck} > 2 \sqrt{\lambda \ell}$. 
 By an application of the Hoeffding bound (Lemma~\ref{lem:hoeffding}), the probability that less than $\ell/2+\sqrt{\lambda \ell}$ bits agree is at most $e^{-2\lambda}$.  Thus, with probability all but $e^{-2\lambda}$, $B$ will set $\kappa_B = 0$. 
 
 On the other hand, if $\kappa_A=1$, $A$ transmits a uniformly random string and $B$ receives a noisy version of it, which is also uniformly random. The expected number of bits of the string that agree with $\veck_B$ is $\ell/2$. The probability that more than $\ell/2+\sqrt{\lambda \ell}$ bits agree is at most $e^{-2\lambda}$ by another application of the Hoeffding bound (Lemma~\ref{lem:hoeffding}). Thus, with probability all but $e^{-2\lambda}$, $B$ will set $\kappa_B = 1$ in this case. 
 
 Put together, $A$ and $B$ will agree on the same key bit with probability $1- e^{-2\lambda}$ and the same $\lambda$-bit key with probability $1-\lambda \cdot e^{-2\lambda} \geq  1- 2^{-\lambda}$ (for large enough $\lambda$). 
\end{proof}

\begin{remark}
  It should be clear from the analysis that the protocol can tolerate channel noise as large as $\frac{1}{2}-n^{-o(1)}$ while paying with a correspondingly strong quantitative LSPN hardness assumption (with $k=\omega(1)$). When the channel noise rate becomes as large as $\frac{1}{2} - n^{-O(1)}$, the assumption  disintegrates to LSPN with constant sparsity which is just false.   
\end{remark}
The next order of business is to show that the protocol satisfies the uniform transcript property and the key indistinguishability property.
\begin{lemma}[Security]\label{lem:simpleprot:secure}
  Let $n, m, k \in \mathbb{N}$, and $\eta,p \in [0,1]$. Under the $\LSPN_{n,m,k,\eta}$ assumption with $m = \Theta(n)$, $k=\Theta(\log n)$ and $\eta = \Theta(1)$, the protocol in Figure~\ref{fig:simpleprot} achieves the uniform transcript property and the key indistiguishability property tolerating a channel noise rate of $p=\Theta(1)$.
\end{lemma}
\begin{proof}
    We will show that the joint distribution of the first two transmitted messages in the protocol and the intermediate shared key  $\veck_A$ are jointly random given the public parameters and the channel noise. That is, 
    \begin{align*}
    \bigg( \underbrace{\matC}_{\pp}, \underbrace{(\tilde{\vece}_i, \tilde{\vecf}_i)_{i=1}^{\ell}}_{\mbox{\small noise}}, 
    & \underbrace{(\matC \vecs_i + \vece_i, \vecr_i^T \matC + \vecf_i^T)_{i=1}^\ell, (k_{A,i} = (\vecr_i^T \matC + \vecf_i^T + \tilde{\vecf}_i^T)\vecs_i + t_i')_{i=1}^\ell}_{\tau^*}
    \bigg) \\
    & \approx_c
     \bigg( \matC, (\tilde{\vece}_i, \tilde{\vecf}_i)_{i=1}^{\ell}, 
    (\vecu_i, \vecv_i^T)_{i=1}^\ell, (w_{i})_{i=1}^\ell
    \bigg)
    \end{align*}
    where $\approx_c$ denotes computational indistinguishability (Definition~\ref{def:compind}), $\vecu_i,\vecv_i \sim \fF_2^n$ and $w_{i} \sim \fF_2$ are uniformly random, $\tilde{\vece}_i, \tilde{\vecf}_i \sim \Bern(\eta)^n$, and all other variables are as in Figure~\ref{fig:simpleprot}. This, in particular, gives us the uniform transcript property (Definition~\ref{def:pnrke}).

    Since both of the distributions above is a product distribution conditioned on $\matC$, it suffices to show:
    \begin{align}\label{eq:key}
    \bigg( {\matC}, \tilde{\vece}_i, \tilde{\vecf}_i, 
    & \matC \vecs_i + \vece_i, \vecr_i^T \matC + \vecf_i^T, k_{A,i} = (\vecr_i^T \matC + \vecf_i^T + \tilde{\vecf}_i^T)\vecs_i + t_i'
    \bigg) \\
    & \label{eq:RHS} \approx_c
     \bigg( \matC, \tilde{\vece}_i, \tilde{\vecf}_i, 
    \vecu_i, \vecv_i^T, w_{i}
    \bigg)
    \end{align}
    
    We will show the above via a hybrid argument transforming the distribution on the LHS to the one on the RHS step by step.
    
    \medskip\noindent
    \textit{Hybrid $H0$.} This is the distribution (\ref{eq:key}).

    \medskip\noindent
    \textit{Hybrid $H1$.} This is the same as $H0$ except all instances of $\vecr_i^T \matC + \vecf_i^T$ are replaced with a uniformly random vector $\vecv_i^T \sim \fF_2^n$. That is, the distribution looks like 
    \begin{align}\label{eq:H1} 
    \bigg( \matC, \tilde{\vece}_i, \tilde{\vecf}_i, 
    \matC \vecs_i + \vece_i, \vecv_i^T, k_{A,i} = (\vecv_i^T + \tilde{\vecf}_i^T)\vecs_i + t_i'
    \bigg) \end{align}
    $H1$ and $H0$ are computationally indistinguishable by the LSPN assumption. Indeed, the LSPN assumption implies (via Lemma~\ref{lem:multi-secret}) that $(\matC, \vecr_i^T \matC + \vecf_i^T)$ is computationally indistinguishable from $(\matC, \vecv_i^T)$.  Given the first of these distributions, it is not hard to see that there is a polynomial-time procedure that generates the distribution (\ref{eq:key}). The same procedure, given the second of these distributions, generates (\ref{eq:H1}). Now, if there is a polynomial-time distinguisher between (\ref{eq:key}) and (\ref{eq:H1}), it immediately implies a polynomial-time algorithm to solve LSPN. 

     \medskip\noindent
    \textit{Hybrid $H2$.} This hybrid simply relabels each triple $(\tilde{\vecf}_i, \vecv_i, \vecv_i + \tilde{\vecf}_i)$ into $(\tilde{\vecf}_i, \vecv_i + \tilde{\vecf}_i, \vecv_i)$ where $\vecv_i \sim \fF_2^n$ and $\tilde{\vecf}_i \sim \Bern(p)^n$. The observation is simply that the second and third elements in each triples are individually uniformly random subject to the constraint that their XOR is the first element. The resulting distibution is 
    \begin{align}\label{eq:H2} 
    \bigg( \matC, \tilde{\vece}_i, \tilde{\vecf}_i, 
    \matC \vecs_i + \vece_i, \vecv_i^T + \tilde{\vecf}_i^T, k_{A,i} = \vecv_i^T \vecs_i + t_i'
    \bigg) \end{align}

      \medskip\noindent
    \textit{Hybrid $H3$.} This hybrid replaces the pair $(\matC \vecs_i + \vece_i, k_{A,i})$ by a uniformly random $(\vecu_i, w_i) \sim \fF_2^n \times \fF_2$. The resulting distribution is 
    \begin{align}\label{eq:H3} 
    \bigg( \matC, \tilde{\vece}_i, \tilde{\vecf}_i, 
    \vecu_i, \vecv_i^T + \tilde{\vecf}_i^T, w_{i}
    \bigg) \end{align}
    $H2$ and $H3$ are computationally indistinguishable by the LSPN assumption. Indeed, the LSPN assumption with parameters $n,m=n+1,k$ and $\eta$ implies (via Lemma~\ref{lem:multi-secret}) that 
    $$ \bigg( \matC, \vecv_i, 
    \left[ \begin{array}{c} 
    \matC \\ \vecv_i^T 
    \end{array}\right] \vecs_i + 
    \left[ \begin{array}{c} 
    \vece_i \\ t_i' 
    \end{array}
    \right] \bigg) 
    \approx_c 
    \bigg( \matC, \vecv_i, 
    \left[ \begin{array}{c} 
    \vecu_i \\ w_i 
    \end{array}\right]
    \bigg)
    $$
    Given the first of these distributions, it is not hard to see that there is a polynomial-time procedure that generates the distribution (\ref{eq:H2}). The same procedure, given the second of these distributions, generates (\ref{eq:H3}). Now, if there is a polynomial-time distinguisher between (\ref{eq:H2}) and (\ref{eq:H3}), it immediately implies a polynomial-time algorithm to solve LSPN. 

    Finally, (\ref{eq:H3}) is the same as (\ref{eq:RHS}) simply because $\vecv_i \sim \fF_2^n$ and $\vecv_i + \tilde{\vecf}_i$ are the same (uniform) distribution.  This finishes the proof.

    Now, the shared key $\kappa_A$ is $0$ if the transcript distribution is (\ref{eq:key}) and $\kappa_A$ is $1$ if the transcript distribution is (\ref{eq:RHS}). Since (\ref{eq:key}) $\approx_c$ (\ref{eq:RHS}), key indistinguishability (Definition~\ref{def:pnrke}) follows immediately.
\end{proof}

Putting Lemmas~\ref{lem:simpleprot:correct} and \ref{lem:simpleprot:secure} together, we obtain the main theorem of this section.

\begin{theorem}\label{thm:simpleprot}
  Let $n, m, k \in \mathbb{N}$, and $\eta,p \in [0,1]$. Under the $\LSPN_{n,m,k,\eta}$ assumption with $m = \Theta(n)$, $k=\Theta(\log n)$ and $\eta = \Theta(1)$, the protocol in Figure~\ref{fig:simpleprot} is a pseudorandom noise-tolerant key exchange scheme that exchanges a single key bit tolerating a channel noise rate of $p=\Theta(1)$.
\end{theorem}